\documentclass[12pt,a4paper]{article}
\usepackage{jheppub}

\usepackage[utf8x]{inputenc}
\usepackage[T1]{fontenc}
\usepackage[english]{babel}
\usepackage{soul}
\usepackage{comment}
\usepackage{subcaption}
\usepackage[all]{xy}

\usepackage{tabularray}

\usepackage{graphicx}
\usepackage{float}

\usepackage{amsmath}	
\usepackage{amssymb}
\usepackage{amsfonts}
\usepackage{mathrsfs}
\usepackage{mathtools}
\usepackage{amsthm}
\usepackage{upgreek}

\usepackage{braket}
\usepackage{slashed}
\usepackage{hyperref}

\usepackage{cleveref}

\usepackage{epsf}
\usepackage{epsfig}	
\usepackage{epstopdf}
\usepackage[font=small,labelfont=bf, width=.95\textwidth]{caption}

\theoremstyle{plain}
\newtheorem*{theorem*}{Theorem}
\newtheorem{theorem}{Theorem}[section]

\newtheorem{proposition}[theorem]{Proposition}
\newtheorem*{corollary*}{Corollary}

\theoremstyle{remark}
\newtheorem{defn}{Definition}[section]

\usepackage{verbatim}
\usepackage{textcomp}
\usepackage{enumitem}   

\definecolor{bubbles}{rgb}{0.91, 1.0, 1.0}
\definecolor{aquamarine}{rgb}{0.5, 1.0, 0.83}
\definecolor{bubblegum}{rgb}{0.99, 0.76, 0.8}
\definecolor{bluebell}{rgb}{0.64, 0.64, 0.82}
\definecolor{dollarbill}{rgb}{0.72, 0.93, 0.6}

\newcommand{\fp}{x}
\newcommand{\nprops}{N}
\newcommand{\eexternal}{E}
\newcommand{\rrank}{R}
\newcommand{\feynpar}{x}

\newcommand{\sumnus}{\nu_1+\dots+\nu_{\nprops}}
\DeclareMathOperator{\nint}{N_{\mathrm{int}}}

\usepackage[toc]{appendix}

\usepackage{layout}

\usepackage{microtype}
\usepackage{bookmark}

\usepackage{booktabs}
\usepackage{longtable}	
\usepackage{multirow}
\usepackage{diagbox}
\usepackage{makecell}

\usepackage{color}
\usepackage[]{xcolor}

\newcommand{\dd}{\mathrm{d}} 

\newcommand{\sfA}{\mathsf{A}}

\newcommand{\cU}{\mathscr {U}}
\newcommand{\cF}{\mathscr {F}}
\newcommand{\cV}{\mathscr {V}}
\newcommand{\cL}{\mathscr {L}}

\renewcommand{\Re}{\operatorname{Re}}

\definecolor{unamblue}{cmyk}{1 0.79 0.12 0.59}

\usepackage[]{hyperref}
\hypersetup{
	colorlinks=true,%
	citecolor=blue,%
	filecolor=blue,%
	linkcolor=blue,%
	urlcolor=blue,
	bookmarksnumbered=true,     
	bookmarksopen=true,         
	bookmarksopenlevel=1,       
	pdfstartview=Fit,           
	pdfpagemode=UseOutlines,
	pdfpagelayout=TwoPageRight
}

\usepackage{graphicx}
\usepackage{slashed}
\usepackage{epstopdf}
\usepackage{verbatim}	
\usepackage{blkarray}
\usepackage{ytableau}

\newcommand{\undernotation}[1]{\pmb{#1}}

\usepackage[compat=1.1.0]{tikz-feynman}
\usepackage{tikz,contour}
\usetikzlibrary{calc,decorations.markings,positioning}
\usepackage{pgfplots}
\pgfplotsset{compat=1.18}
\pgfdeclarelayer{bg}    
\pgfsetlayers{bg,main}  

\tikzset{
    partial ellipse/.style args={#1:#2:#3}{
        insert path={+ (#1:#3) arc (#1:#2:#3)}
    }
  }

\tikzfeynmanset{warn luatex=false}

\usepackage[	]{todonotes}
\usepackage{subcaption}

     \newcommand{\smallNickel}[1]{{                \texttt{#1}}}

\newcommand{\ii}{\mathrm{i}}

\newcommand{\Conv}{\operatorname{Conv}}
\newcommand{\Newton}{\Delta}

\DeclareMathOperator{\aff}{aff}
\DeclareMathOperator{\supp}{supp}
\DeclareMathOperator{\ehr}{Ehr}
\newcommand{\simplexwithpars}[1]{\Newton_{\text{HS}}^{(#1)}}

\title{Fano and Reflexive Polytopes from Feynman Integrals} 

\author[a]{Leonardo de la Cruz,}
\affiliation[a]{Institut de Physique Th\'eorique, Université Paris-Saclay, CEA, CNRS, F-91191 Gif-sur-
	Yvette Cedex, France}

\author[b]{Pavel P.~Novichkov}
\affiliation[b]{Department of Physics and Astronomy, Ghent University, 9000 Ghent, Belgium}

\author[a]{and Pierre Vanhove}

\date{\today}

\abstract{We classify the Fano and reflexive polytopes that arise
  from quasi-finite Feynman integrals.  These polytopes appear as
  scaled Minkowski sums of the Newton polytopes associated with the
  Symanzik graph polynomials. 
  For one-loop graphs and  multiloop sunset graphs, we
  identify the Fano and reflexive cases by computing
  the number of interior points from the associated bivariate Ehrhart polynomials. More generally,  we utilize the properties of Symanzik polynomials and their symmetries to conduct a direct search over all Feynman graphs in generic kinematics with up to ten edges and nine loops.
  We find that such cases are remarkably sparse: for example, we find only two two-dimensional
  reflexive polytopes, three three-dimensional reflexive polytopes, and
  four  three-dimensional Fano polytopes. We also reveal a surprising feature of one-loop $N$-gon integrals in higher dimensions:  their associated reflexive polytopes encode degenerate Calabi--Yau $(N-2)$-folds.  We further
  analyze the geometric structures encoded by these polytopes and
  exhibit explicit connections with del Pezzo surfaces, $K3$ surfaces,
  and Calabi--Yau threefolds. Since reflexive polytopes naturally
  correspond to Calabi--Yau varieties, our classification demonstrates
  that quasi-finite Feynman integrals, with reflexive polytopes, are intrinsically linked to
  Calabi--Yau period integrals.
}

\begin{document}
	\maketitle


	\section{Introduction}

Feynman integrals are central objects in perturbative quantum field
theory, encoding the contributions of quantum fluctuations to physical
observables such as scattering amplitudes and the determination of effective
couplings (see, for instance, the review~\cite{Travaglini:2022uwo} for a
recent account of the breadth of applications of amplitudes).
Beyond their physical significance, these integrals exhibit
a remarkably rich mathematical structure, interconnecting quantum
field theory with algebraic geometry, number theory, and
combinatorics. It is now widely recognized that
many Feynman integrals can be interpreted as periods of
algebraic varieties defined by the vanishing of the Symanzik
polynomials~\cite{Bloch:2005bh,Bogner:2007mn,Brown:2009ta, Schnetz:2013hqa}. This
observation has opened a new perspective on quantum field theory,
emphasizing the role of algebraic geometry in understanding the
analytic and arithmetic properties of amplitudes. 

Feynman
integrals also serve as a testing ground for discovering novel
geometries that emerge at higher-loop orders. The same class of
multiloop integrals governs observables such as Higgs boson
production at the LHC,  the anomalous magnetic moments of the
electron or the muon~\cite{Laporta:2017okg, Blum:2023qou}, or the
post-Minkowskian expansion of general relativity for the dynamics of
compact binaries relevant to gravitational-wave
emission~\cite{Bjerrum-Bohr:2022blt}. Recent computations have
revealed the appearance of Calabi--Yau period integrals in physical
observables. For instance, $K3$ periods enter in the evaluation of the
three-loop contribution to two-body post-Minkowskian scattering in
general relativity~\cite{Bern:2021dqo,Bern:2021yeh,Dlapa:2021npj,Dlapa:2022lmu,Bjerrum-Bohr:2022ows} and the hadronic vacuum polarization in chiral
perturbation theory~\cite{Lellouch:2025rnz}, and Calabi--Yau three-fold periods in the four-loop
computation of post-Minkowskian scattering in
general relativity~\cite{Klemm:2024wtd,Driesse:2024feo}. The function space that appears in intermediate steps of these calculations has been studied in refs.~\cite{Frellesvig:2023bbf, Brammer:2025rqo}.

Newton polytopes associated with Feynman integrals also play an important role in understanding their geometry. They arise from the graph polynomials in the parametric representation. In this representation, the Newton polytope associated with a  Feynman graph $\Gamma$ is the Minkowski sum of the Newton polytope of the first Symanzik polynomial $\cU_\Gamma$ and the second Symanzik polynomial $\cF_\Gamma$.
Newton polytopes of Feynman integrals are useful 
	for studying their convergence properties~\cite{2010arXiv1010.5060N,2011arXiv1103.6273B},  
	the method of regions~\cite{Ananthanarayan:2018tog,Gardi:2022khw,Gardi:2024axt,Ma:2023hrt},  their evaluation using sector decomposition~\cite{Borowka:2015mxa,Heinrich:2021dbf}, and their relation to $\sfA$-hypergeometric systems~\cite{GELFAND1990255, 2010arXiv1010.5060N, 2016arXiv160504970N, Schultka:2018nrs,  delaCruz:2019skx, Klausen:2019hrg, Klausen:2021yrt,Ananthanarayan:2022ntm}. 	 Newton polytopes that arise from Feynman integrals are lattice polytopes, that is, their vertices have integer coordinates.  To a given lattice polytope we can attach a toric variety, which carries information about the convergence of a given Feynman integral~\cite{2011arXiv1103.6273B} and the differential equations it satisfies~\cite{delaCruz:2019skx}.	
    
 Fano polytopes are a particular class of lattice polytopes having a single interior lattice point (see ref.~\cite{cox2011toric}).
 An important subclass of Fano polytopes, known as reflexive polytopes, was put forward
 in ref.~\cite{Batyrev:1993oya}.
 Reflexive polytopes play a distinguished role in mirror symmetry, as being associated with Calabi--Yau varieties.
In a fixed number of dimensions, all reflexive polytopes are known up to dimension four~\cite{Kreuzer:2000xy}. 

Certain Feynman integrals lead to mirror
pairs. This is the case for example, of  multiloop sunset graphs in two dimensions, whose connection to mirror symmetry was rigorously studied
in refs.~\cite{Bloch:2014qca,Bloch:2016izu}. There it was found that these Feynman integrals  are multiple-valued holomorphic
functions, which arise as well in the context of open mirror symmetry.
The relevance of the toric Fano varieties, associated to Fano polytopes, in the context of Feynman integrals has been put
forward as well in ref.~\cite{Schimmrigk:2024xid}.

\medskip

	In another context, algorithms for classifying all locally finite integrals in four dimensions for a given Feynman graph 
    were presented in refs.~\cite{Gambuti:2023eqh,delaCruz:2024xsm}. In particular, the algorithm of ref.~\cite{delaCruz:2024xsm} relies on the convergence region of Euler--Mellin integrals studied in ref.~\cite{2010arXiv1010.5060N},  later generalized in refs.~\cite{2011arXiv1103.6273B, Schultka:2018nrs}. It is based on constructing numerators that have the property of having parametric representations that depend only on monomials whose exponents are  interior lattice points of the Newton polytope of their Symanzik polynomials.
    In this context, Fano and reflexive polytopes arise naturally as  special cases of Newton polytopes associated to finite integrals with only one interior lattice point.
    Quasi-finite integrals introduced in ref.~\cite{vonManteuffel:2014qoa}  lead to Fano and reflexive polytopes as well.

In this work, we will take a systematic step toward classifying the \emph{reflexive} and \emph{Fano} polytopes
that arise from the Newton polytopes of the Symanzik polynomials associated to finite Feynman integrals. The classification presented in this paper establishes a bridge between the combinatorics of Feynman graphs and the geometry of their associated toric varieties. By systematically identifying the reflexive and Fano polytopes corresponding to classes of Feynman integrals, we construct a geometric dictionary that links the analytic structure of amplitudes to the combinatorial and topological data of toric geometry. This correspondence provides a new framework for exploring the interplay between algebraic geometry and quantum field theory, with potential applications ranging from the theory of periods and motives to high-precision predictions in particle physics and gravitational theory.

\medskip
This paper is structured as follows.
In \cref{sec:finite}, we review the parametric representation of Feynman integrals.
In \cref{sec:newton}, we develop the framework of Newton polytopes associated with the Symanzik polynomials,
we discuss the convergence conditions that define (quasi-)finite integrals,
and explain how bivariate Ehrhart polynomials are used to determine interior lattice points.
In \cref{sec:reflexivefano}, we define Fano and reflexive polytopes and present the example of  multiloop sunset integrals.
In \cref{sec:dimloopscan}, we perform a systematic search by increasing the number of  dimensions and loops. We identify all reflexive and Fano polytopes that arise from one- and two-loop graphs in  $D=2$ and $D=4$,  isolating the rare cases where Newton polytopes have a single interior point. Using symmetry relations and Ehrhart bounds, we extend this search to graphs up to ten edges and nine loops in  
\cref{sec:edgescan}. In
\cref{sec:periods}, we evaluate certain Feynman integrals  associated to reflexive polytopes that lead to simple period integrals.
In \cref{sec:mirror}, 
we  interpret the resulting polytopes in terms of toric geometry and Calabi--Yau periods, establishing a geometric dictionary between quasi-finite Feynman integrals and some Calabi--Yau families. 
Our conclusions are presented in \cref{sec:conclusion}. 

The appendices provide the technical material supporting the main analysis. 
\Cref{app:TwoloopGraph} lists the Symanzik polynomials for all two-loop graph topologies. 
\Cref{app:oneLoopEhrhart} derives the Ehrhart polynomial used in the one-loop scan. 
\Cref{app:sunsetEhrhart} develops the full bivariate Ehrhart polynomial for the multiloop sunset polytopes and relates it to permutohedra. 
\Cref{app:fanotable} contains  tables of graphs leading to Fano polytopes. 
	
	\section{Parametric representation of Feynman integrals}\label{sec:finite}
	
\newcommand{\prefactorGamma}{\Upgamma}	
	Feynman integrals associated with a Feynman graph $\Gamma$ generally appear in quantum field theory as tensor integrals, that is,  integrals with numerator factors. Let us consider a graph  $\Gamma$ with $L$ loops and $\nprops$ internal edges in $D$ space-time dimensions. 
	A general tensor integral attached to this graph is a linear combination of integrals of the form 
	\newcommand{\multiJ}{\undernotation{J}}
    \begin{equation}\label{e:tensorFeyn}
		I_\Gamma^{\multiJ}(L,D;\undernotation{\nu}):= \int_{(\mathbb R^{1,D-1})^L}{\ell_{j_1}^{\mu_1}\cdots
			\ell_{j_R}^{\mu_R}\over \prod_{i=1}^{\nprops} d_i^{\nu_i} }
		\prod_{i=1}^L \frac{  \dd^D\ell_i }{\ii \pi^{D/2} }\,,
	\end{equation}
	where $\rrank$ is the number of times that the loop momenta appear in the numerator. 
    The double-index  $\multiJ:=(J_1, J_2, \dots, J_R)$,  with  $J_i=(j_i,\mu_i)$, denotes the $i^{\mathrm{th}}$ Lorentz index belonging to the $j^{\mathrm{th}}$ loop momentum. The powers of the denominators are denoted collectively by $\undernotation{\nu} :=(\nu_1,\dots,\nu_{\nprops})$ while the inverse propagators are given by 
	\begin{equation}
		d_i=
		\left( \sum_{j=1}^L \alpha_{ij} \ell_j+\sum_{j=1}^{\eexternal} \beta_{ij}p_j \right)^2-m_i^2+\ii \epsilon\,, 
        \label{denominators}
	\end{equation}
	where
	$\alpha_{ij}$, $\beta_{ij}$ take the values in  $\set{-1,0,1}$, $m_i$ are the  internal masses and $p_i$ are the $E$ independent
	external momenta. We will explicitly indicate  the parameters $(L,D;\undernotation{\nu})$ throughout, as they   will play a central role in the subsequent analysis.
	The parametric representation  of eq.~\eqref{e:tensorFeyn}  can be written as~\cite{Heinrich:2008si} 
	\begin{multline}\label{e:Itensor}
		I_\Gamma^{\multiJ}(L,D;\undernotation{\nu})=   \prefactorGamma\left(\sumnus-\left\lfloor
                    R\over2\right\rfloor-{LD\over2}\right)\cr
                \times\int_{ \mathbb R\mathbb P_+^{\nprops-1}} {\mathscr{U}_\Gamma^{\sumnus-{L+1\over2}D}\over
			\mathscr{F}_\Gamma^{\sumnus-{L\over2}D}}{\mathcal
			N^{\multiJ}(\undernotation x;\undernotation{\nu},\undernotation a)\over \mathscr{U}_\Gamma^R}
		\prod_{i=1}^{\nprops}
		{\feynpar_i^{\nu_i-1}\over\prefactorGamma(\nu_i) }  \Omega_0^{(\nprops)} \, ,
	\end{multline}
	where $\left\lfloor R\over2\right\rfloor$ denotes the nearest
        integer less than or equal to $R$ and the domain of integration
        is the projective positive orthant
        \begin{equation}
          \mathbb R\mathbb P_+^{\nprops-1}:=\{[x_0,\dots, x_{\nprops}]\in
\mathbb P^{\nprops-1}| x_i\in\mathbb R, x_i\geq0\},
\end{equation}
        and
	\begin{equation}
		\Omega_0^{(\nprops)}=\sum_{i=1}^{\nprops} (-1)^i x_i \dd x_1\wedge \cdots
		\wedge \dd x_{i-1}\wedge \dd x_{i+1}\wedge \cdots \wedge \dd x_\nprops
	\end{equation}
	is    the canonical   differential form on $\mathbb
        P^{\nprops-1}$. The integral depends on the Symanzik (graph) polynomials and a numerator that we define as follows.

        In order to construct  Symanzik polynomials $\cU_\Gamma$ and $\cF_\Gamma$,  we introduce the mass hyperplane
	\begin{equation}
		\label{e:Ldef}
		\cL_\Gamma:=\sum_{i=1}^\nprops m_i^2 \feynpar_i \, 
	\end{equation}
	and  the graph polynomial from the sum of the 2-forests  (we
        refer to ref.~\cite{Nakanishi:1971} for details about the
        construction of these graph polynomials)
	\begin{equation}
		\label{e:Vdef}
		\cV_\Gamma:=\sum_{\substack{\text{spanning}\\\text{2-forests}~F~\text{of}~\Gamma}}
		\quad  \left(\sum_{(v_1,v_2)\in F=T_1\cup T_2} p_{v_1}\cdot
		p_{v_2}\right)\, \prod_{j\in \Gamma\setminus F} \feynpar_j \,. 
	\end{equation}
	The Symanzik polynomials are  defined by
	\begin{align}
		\label{e:Udef}
		\cU_\Gamma&:= \sum_{\substack{\text{spanning}\\\text{trees}~T~\text{of}~\Gamma}}
		\quad
		\prod_{j\in \Gamma\setminus T} \feynpar_j\,, \\
		\label{e:Fdef}
		\cF_\Gamma&:= \cU_\Gamma \,\cL_\Gamma-\cV_\Gamma \, .
	\end{align}
	The numerator $\mathcal{N}^{\multiJ}(\undernotation{x};\undernotation{\nu},\undernotation{a})$ is a homogeneous polynomial of degree $RL$ in the edge variables $\feynpar_1,\dots,\feynpar_\nprops$,
    \begin{equation}
		\mathcal N^{\multiJ}(\undernotation{x};\undernotation{\nu},\undernotation {a}):=    \sum_{k=0}^{\left\lfloor
			R\over2\right\rfloor} \left(-\frac12\right)^k
		{ \prefactorGamma\left(\sumnus-k-{LD\over2}\right)\over  \prefactorGamma\left(\sumnus-\left\lfloor
			R\over2\right\rfloor-{LD\over2}\right)} 
		\times \,  \mathscr{F}_\Gamma^k\,
		\mathscr{P}_k^{\multiJ}(\undernotation x) \,,
	\end{equation}
	where $\mathscr{P}^{\multiJ}_k(\undernotation x)$ is a homogeneous polynomial
 of degree
	$RL-k(L+1)$,
	\begin{equation}\label{e:Pdef}
		\mathscr{P}^{\multiJ}_k(\undernotation
		x)=\sum_{\undernotation{a} \in \textrm{supp}(\Gamma;\mu)}
		p^{\multiJ }_{\undernotation a }\prod_{i=1}^\nprops x_i^{a_i} \, ,
	\end{equation}
where $p^{\multiJ }_{\undernotation a }$ is a tensor that can be constructed from metric tensor and the matrices in eq.~\eqref{denominators}. Its explicit form can be found in ref.~\cite{Heinrich:2008si} and may also be computed with \texttt{SecDec} \cite{Heinrich:2021dbf}. We won't need the explicit form of this polynomial as we explain below. 
    We have introduced the support of the homogeneous polynomial,
	\begin{equation}
		\textrm{supp}(\mathscr{P}):=\{
		a_1+\cdots +a_\nprops=RL-k(L+1); \qquad
		R(L-1)/2\leq a_i\leq RL\}\, .
	\end{equation}
        We remark that terms of the form  \eqref{e:tensorFeyn} with different values of $R$ can also be put in the form of eq.~\eqref{e:Itensor} by taking the maximum value of $R$ in the combination and adjusting the numerator to include additional powers of $\cU$. 
	
       In practical calculations, loop momenta is contracted with external momenta and thus each term in the numerator becomes a scalar. Thus, we can focus on the monomials appearing in $\mathcal N^{\multiJ}(\undernotation{x};\undernotation{\nu},\undernotation {a})$. 
		Taking one monomial at the time in the numerator, say $x_1^{\alpha_1}\cdots x_\nprops^{\alpha_\nprops}$, amounts to studying  the scalar integral (now dropping the double-index $\multiJ$)
		\begin{equation}
			I_\Gamma (L, D; \undernotation{\alpha},\undernotation{\nu}) :=   	\int_{\mathbb R\mathbb P_+^{\nprops-1}} {\mathscr{U}_\Gamma^{\sumnus-{L+1\over2}D}\over
				\mathscr{F}_\Gamma^{\sumnus-{L\over2}D}}{1\over \mathscr{U}_\Gamma^R}
			\prod_{i=1}^{\nprops}
			{\feynpar_i^{\alpha_i+\nu_i-1}\over\prefactorGamma(\nu_i) }  \Omega_0^{(\nprops)} \, .
			\label{basic-feynman-integral}
		\end{equation}
		This definition does not include the $\Upgamma$-function prefactor for reasons detailed in \cref{sec:quasifinite}.
		Since $\alpha_1 + \dots + \alpha_\nprops = RL$, we can rewrite the integral~\eqref{basic-feynman-integral} as
		\begin{multline}
			I_\Gamma (L, D; \undernotation{\alpha}, \undernotation{\nu}) =
			\prod_{i=1}^\nprops \frac{\prefactorGamma(\alpha_i + \nu_i)}{\prefactorGamma(\nu_i)} \\
            \times
			\int_{\mathbb R\mathbb P_+^{\nprops-1}} {\mathscr{U}_\prefactorGamma^{(\alpha_1+\nu_1) + \dots + (\alpha_\nprops + \nu_\nprops) - (L+1)(D/2+R)}\over
				\mathscr{F}_\Gamma^{(\alpha_1+\nu_1) + \dots + (\alpha_\nprops + \nu_\nprops) - L(D/2+R)}}
			\prod_{i=1}^{\nprops}
			{\feynpar_i^{\alpha_i+\nu_i-1}\over\prefactorGamma(\alpha_i + \nu_i) }  \Omega_0^{(\nprops)} \, ,
		\end{multline}
		manifesting that, up to a numerical prefactor, the rank~$R$ and the numerator exponents~$\alpha_i$ enter the integrand only through linear combinations $D/2+R$ and $\alpha_i + \nu_i$, respectively. We have thus recovered the well-known result~\cite{Tarasov:1996br} stating that Feynman integrals with numerator factors can be expressed as linear combinations of scalar integrals, possibly in higher dimensions and with raised propagator powers.
        
		From now on, we will redefine $D \to D + 2R$, $\nu_i \to \nu_i + \alpha_i$ effectively setting $R$ and $\nu_i$ to zero, so that the integral under study takes the form
		\begin{equation}
			I_\Gamma (L,D;\undernotation{\nu}) :=
			\int_{\mathbb R\mathbb P_+^{\nprops-1}} \frac{\prod_{i=1}^{\nprops}
			{\feynpar_i^{\nu_i-1}}}{\mathscr{U}_\Gamma^{n_\cU} \mathscr{F}_\Gamma^{n_\cF}}
			  \Omega_0^{(\nprops)} ,
			\label{parametric-integral-projective}
		\end{equation}
		where we have introduced the short-hand notation 
		\begin{equation}
			n_\cU=	- \nu_1-\cdots-\nu_\nprops+{L+1\over2}D\, , \qquad 
			n_\cF=\nu_1+\cdots+\nu_\nprops-{L\over2}D\, .
            \label{definitionsofpowers}
		\end{equation}
                In an affine patch, say $\feynpar_\nprops=1$, the
                projective integral becomes
		\begin{equation}
			I_\Gamma (L,D;\undernotation{\nu}) =
			\int_{\mathbb{R}^{\nprops-1}_+} \left. \frac{\prod_{i=1}^{\nprops-1}
			{\feynpar_i^{\nu_i}}}{\mathscr{U}_\Gamma^{n_\cU} \mathscr{F}_\Gamma^{n_\cF}} \right|_{\feynpar_\nprops = 1}
			\prod_{i=1}^{\nprops-1}
		 {\dd x_i\over x_i}.
			\label{parametric-integral-affine}
		\end{equation}
	The equivalence between the two expressions is a consequence
 of the projective nature of the parametric representation and it is 
referred to as the Cheng--Wu theorem~\cite[App.~C.10]{Cheng:1987ga} in the physics literature~\cite{Panzer:2015ida,Weinzierl:2022eaz}.

		\section{Newton polytopes attached to Feynman integrals}\label{sec:newton}

		Introducing the multi-index  notation
		\begin{equation}
			\undernotation x^{\undernotation {a}}:=   x_1^{ a_1}\cdots x_n^{a_n}\, ,
		\end{equation}
		to a  (Laurent) polynomial in $n$ variables
		\begin{equation}\label{e:LaurentDef}
			f(x_1,\dots,x_n)=\sum_{\undernotation a=(a_1,\dots,a_n)\in \mathbb
				Z^n} c_{\undernotation a} \undernotation x^{\undernotation a}, \qquad
			c_{\undernotation a} \in\mathbb C \, ,
		\end{equation}
		we associate a Newton polytope
		\begin{equation}\label{e:DeltaLaurentDef}
			\Newton(f):= \left\{\sum_{i=1}^n\lambda_i a_i \, \middle| \, \sum_{i=1}^n
			\lambda_i=1, \lambda_i\in\mathbb R_+, \undernotation a\in \textrm{supp}(f)\right\} 
		\end{equation}
		as the convex hull of its support
		$
		\supp(f):=\set{\undernotation a\in\mathbb Z^n | c_{\undernotation a}\neq0}.
		$
		The Newton polytope of the product of polynomials is                the Minkowski sum
		\begin{equation}
			\Newton(f\cdot g)=\Newton(f)+\Newton(g)=
                        \{x+y \, | \, x\in \Newton(f), y\in \Newton(g) \} \, .
			\label{Newton-property}
                      \end{equation}
For later purposes, we introduce the hypersimplex (HS) $\simplexwithpars{d,k}$. It is defined as the convex hull of $d$-dimensional vectors whose coefficients consist of $k$ ones and $d-k$ zeros. The standard simplex corresponds to 
  \begin{align}\label{e:standardsimplex-definition}
      \simplexwithpars{\nprops,1} &:= \Newton\!\left(\sum_{i=1}^\nprops x_i\right)
    = \Conv\{e_1, \dots, e_\nprops\}\cr
    &=  \left\{ (\nu_1,\dots,\nu_\nprops) \in \mathbb{R}^\nprops \, \middle| \, \sum_{i=1}^\nprops \nu_i = 1, \ 0\le \nu_i \le 1 \text{ for all } i \right\},
  \end{align}
  and the second hypersimplex
  \begin{align}\label{e:secondhypersymplex-definition}
      \simplexwithpars{\nprops,2} &:= \Newton\!\left(\sum_{1 \le r < s \le \nprops} x_r x_s\right)
    = \Conv\{ e_r + e_s \, | \, 1 \le r < s \le \nprops\}\cr
    &= \left\{ (\nu_1,\dots,\nu_\nprops) \in \mathbb{R}^\nprops \, \middle| \, \sum_{i=1}^\nprops \nu_i = 2, \ 0\le \nu_i \le 1 \text{ for all } i \right\}\, ,
  \end{align}
where \( e_i \) are the standard basis vectors of \( \mathbb{R}^\nprops \).

        \subsection{Newton polytopes in the parametric representation}
        
		To the parametric representation of the Feynman integral in
		\cref{parametric-integral-projective} or \cref{parametric-integral-affine}, we associate the Newton polytope%
		\begin{equation}\label{e:NewtonUF}
			\Delta_\Gamma(L,D;\undernotation{\nu}) =
			n_\cU\Newton (\cU_\Gamma)+n_\cF\Newton(\cF_\Gamma),                
		\end{equation}
		assuming $n_\cU, n_\cF \geq 0$.
		
		For integer values of $n_\cU$ and $n_\cF$, $\Delta_\Gamma(L,D;\undernotation{\nu})$  is a lattice polytope and \cref{e:NewtonUF} follows from eq.~\eqref{Newton-property} because 
    $\Delta_\Gamma(L,D;\undernotation{\nu})=	\Newton(\cU_\Gamma^{n_\cU}\cF_\Gamma^{n_\cF})$.       All its vertices have integer components and thus lie on the lattice~$\mathbb{Z}^n \cap \aff(\Delta_\Gamma)$, where $\aff(\Delta_\Gamma) \subset \mathbb{R}^n$ is the smallest affine space containing $\Delta_\Gamma$. Even though \cref{parametric-integral-projective} and \cref{parametric-integral-affine} define two different Newton polytopes, they are related by a projection map $(a_1, \dots, a_{N-1}, a_N) \mapsto (a_1, \dots, a_{N-1})$. Moreover, since the integrand of \cref{parametric-integral-projective} is homogeneous, its Newton polytope is restricted to a hyperplane, so the projection map defines an isomorphism of the respective lattices $\mathbb{Z}^N \cap \{a_1 + \dots + a_N = \mathrm{const}\}$ and $\mathbb{Z}^{N-1}$. Therefore, it preserves properties related to lattice polytopes, such as volume, number of interior lattice points, and reflexivity~\cite{Haase:2012,Beck:2015}. For this reason, we will use the projective and the affine form of the Newton polytope interchangeably.

                \medskip
		We give a few remarkable properties that we will use later when
		analyzing  polytopes arising from Feynman integrals.
		
		\begin{itemize}
			\item
			When all the internal masses are non-vanishing, i.e.\ $m_i\neq0$ for
			$1\leq i\leq \nprops$, we have that the Newton polytope of $\cF_\Gamma$ is
			given by the Minkowski sum
			\begin{equation}\label{e:FMinkowski}
				\Newton(  \cF_\Gamma) =\Newton(\cU_\Gamma)+\Newton(\cL_\Gamma)\, .
			\end{equation}
			This relation  is a direct consequence of the definitions of the graph polynomials when all the internal mass coefficients are non vanishing (this was also observed  e.g. in ref.~\cite{Arkani-Hamed:2022cqe,Borinsky:2023jdv}).  
             The Newton polytope $\Newton(\cL_\Gamma)=\simplexwithpars{\nprops,1}$ is the standard
simplex in \cref{e:standardsimplex-definition}.
			This means that when all the internal masses are non-vanishing the
			Newton polytope of the integrand of the Feynman integral~\eqref{parametric-integral-projective}
            is given by the weighted Minkowski sum of $\Newton (\cU_\Gamma)$ and $\simplexwithpars{\nprops,1}$. Specifically,
            \begin{equation}\label{e:NewtonUFMassive}
				\Delta_\Gamma(L,D;\undernotation{\nu}) =
				\frac{D}{2}
				\Newton (\cU_\Gamma)+\left(\nu_1+\cdots+\nu_\nprops-\frac{DL}
				2\right)\simplexwithpars{\nprops,1}                
			\end{equation}
			is determined only by the polytope of the first Symanzik polynomial
			and the linear mass polynomial.
			
\item			When some of the internal masses vanish, 
			\cref{e:FMinkowski} is not true anymore as was observed in ref.~\cite[Remark 4.16]{Schultka:2018nrs}. In particular, in the fully
			massless case $(m_1,\dots,m_\nprops)=(0,\dots,0)$ (external momenta can still be massive) the Newton polytope
			associated to the Feynman integral is given by the weighted Minkowski sum of the Newton
			polytope of $\Newton (\cU_\Gamma)$ and  $\Newton(\cV_\Gamma)$
			\begin{equation}
            \begin{aligned}
				\label{e:NewtonUFMassless}
				\Delta_{\Gamma;\undernotation 0}(L,D;\undernotation{\nu}):=    \bigg(\frac{D(L+1)}{2}-\nu_1&-\cdots-\nu_\nprops\bigg)
				\Newton (\cU_\Gamma)\\
                &+\left(\sumnus-\frac{DL}
				2\right)\Newton(\cV_\Gamma)\, .
			\end{aligned}
            \end{equation}
			We discuss the difference between the polytopes of graphs with all
            massive internal propagators and graphs with all massless internal propagators in \cref{sec:dimloopscan} in the context of one-loop graphs.

\item The Newton polytope of the first Symanzik polynomial is  a generalized permutahedron~\cite{postnikov2005}  as shown in ref.~\cite[Theorem~3.4]{Borinsky:2023jdv}.  When the Minkowski sum relation in \cref{e:FMinkowski} applies the Newton polytope is also given by a generalized permutahedron as shown in ref.~\cite[Theorem~3.5]{Borinsky:2023jdv}.
            
			\item There is a special class of polynomials with saturated Newton polytopes, meaning that
 every lattice point of the
			Newton polytope corresponds to a non-vanishing monomial~\cite{Cara:2019}
			\begin{equation}
				\Newton(f) \cap \mathbb{Z}^n= \supp (f)\, .
			\end{equation}
			The graph polynomials of one-loop integrals
                        and the multiloop sunset graph polynomials have saturated Newton
			polytopes.
		\end{itemize}

\subsubsection{Newton polytopes in Lee--Pomeransky representation}
\label{sec:LeePomeransky}

Here we explore the relation between Newton polytopes of Feynman integrals in the Feynman parameter representation~\eqref{parametric-integral-projective} and another commonly-used parametric representation---the Lee--Pomeransky (LP) representation~\cite{Lee:2013hzt}.
This representation is typically presented in deprojectivized form, which reads, up to a prefactor,
\begin{equation}
	I_\Gamma (L,D;\undernotation{\nu}) =
	\int_{\mathbb{R}^{\nprops}_+} \frac{1}{\left( \cU_\Gamma + \cF_\Gamma \right)^{D/2}}
	\prod_{i=1}^{\nprops}
	{\feynpar_i^{\nu_i-1}} \dd x_i \, .
	\label{lp-integral-affine}
\end{equation}
To make the relation between Newton polytopes more transparent, we will instead consider the projective version of this integral:
\begin{equation}
	I_\Gamma (L,D;\undernotation{\nu}) =
	\int_{\mathbb R\mathbb P_+^{\nprops}} \frac{1}{\left( x_0 \cU_\Gamma + \cF_\Gamma \right)^{D/2}}
	\prod_{i=0}^{\nprops}
	{\feynpar_i^{\nu_i-1}}  \Omega_0^{(\nprops+1)} ,
	\label{lp-integral-projective}
\end{equation}
where \(\nu_0 = \frac{L+1}{2} D - \nu_1-\cdots-\nu_\nprops\). (The affine integral~\eqref{lp-integral-affine} can be obtained from \cref{lp-integral-projective} by setting \(x_0 = 1\) in the integrand.)

In complete analogy to the parametric representation case~\eqref{parametric-integral-projective}, we can associate the following Newton polytope to the integral~\eqref{lp-integral-projective}:
\begin{equation}
	\Delta_\Gamma^{\mathrm{LP}}(L,D) =
	\frac{D}{2} \Newton (x_0 \cU_\Gamma + \cF_\Gamma)\, .
	\label{e:NewtonLP}
\end{equation}
Note that, unlike the parametric representation case, this polytope does not depend on the vector of exponents \(\undernotation{\nu}\).
Furthermore, assuming that \(D/2\) is integer, we can expand the corresponding polynomial in powers of \(x_0\) as
\begin{equation}
	\left( x_0 \cU_\Gamma + \cF_\Gamma \right)^{D/2}
	= \sum_{a_0=0}^{D/2} \binom{D/2}{a_0} \, x_0^{a_0} \, \cU_\Gamma^{a_0} \cF_\Gamma^{D/2 - a_0},
\end{equation}
so that
\begin{equation}
	\supp \left( (x_0 \cU_\Gamma + \cF_\Gamma)^{D/2} \right) =
	\bigcup_{a_0 = 0}^{D/2}
	\left\{ (a_0, \undernotation{a}) \in \mathbb{Z}^{N+1} \, \middle| \, \undernotation{a} \in \supp \left( \cU_\Gamma^{a_0} \cF_\Gamma^{D/2-a_0} \right) \right\}.
\end{equation}
It follows that \(\Delta_\Gamma^{\mathrm{LP}}(L, D)\) is the convex hull of Newton polytopes~\(\Delta_\Gamma (L,D; \undernotation{\nu})\) defined in \cref{e:NewtonUF} with all possible values of \(n_\cU = 0, \dots, D/2\), embedded into a set of parallel hyperplanes \(a_0 = n_\cU\), see \cref{fig:lp-slices}.
This explains our choice of the projective form of the Lee--Pomeransky
representation~\eqref{lp-integral-projective} over the deprojectivized
one~\eqref{lp-integral-affine}: setting \(x_0 = 1\) is equivalent to
projecting all polytopes \(\Delta_\Gamma (L, D; \undernotation{\nu})\)
onto the same hyperplane \(a_0 = 0\), which would obscure this
picture.  

  \begin{figure}[tb]
\begin{subfigure}{0.45\textwidth}
        \centering
        \includegraphics[width=5cm]{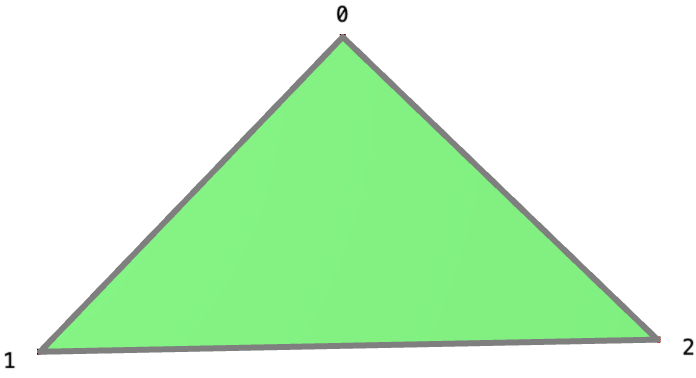}
        \caption{Polytope $\Newton(\cU_ {\rm triangle})+\Newton(\cF_{\rm triangle})$ in the parametric representation.}
        \label{fig:massivetrianglepoluf}
\end{subfigure}
\hfill
\begin{subfigure}{0.45\textwidth}
        \centering
        \includegraphics[width=5cm]{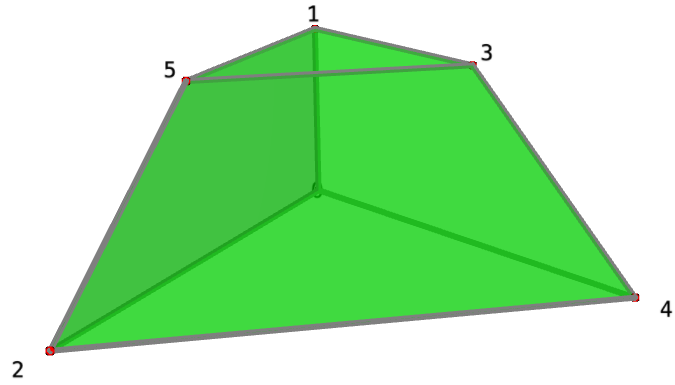}
        \caption{Polytope $2 \Newton(x_0 \cU_ {\rm triangle}+\cF_{\rm triangle})$ in Lee--Pomenransky representation.}
        \label{fig:massivetrianglepollp}
      \end{subfigure}
      	\caption{Comparison of
the 
polytopes for the massive triangle in \(D = 4\) dimensions. Polytope (\subref{fig:massivetrianglepoluf}) is a slice of the Lee--Pomeransky  polytope (\subref{fig:massivetrianglepollp}) by the plane \(a_0 = 1\), parallel to the triangular faces, corresponding to \(a_0 = 0\) and \(a_0 = 2\). }
	\label{fig:lp-slices}
  \end{figure}

It is clear now that \(\Delta_\Gamma^{\mathrm{LP}}(L, D)\) contains
interior points of \(\Delta_\Gamma (L, D; \undernotation{\nu})\)  for
all values of \(n_\cU\). In general, this number can be greater than
that for a single polytope with a fixed value of \(n_\cU\). This makes
Lee--Pomeransky representation less suitable for Fano and reflexivity
searches than Feynman parameter representation. Indeed, if
\(\Delta_\Gamma^{\mathrm{LP}}(L, D)\) is Fano, one can expect to find
a Fano polytope among the Feynman parameter polytopes \(\Delta_\Gamma
(L, D; \undernotation{\nu})\) for some value of \(n_\cU\). On the
other hand, if a Feynman parameter polytope \(\Delta_\Gamma (L, D;
\undernotation{\nu})\) happens to be Fano for some value of \(n_\cU\),
this does not guarantee that the corresponding Lee--Pomeransky
polytope \(\Delta_\Gamma^{\mathrm{LP}}(L, D)\) is Fano, as the latter
may contain additional interior points. 
Therefore, in what follows we will focus on Newton polytopes~\eqref{e:NewtonUF} associated to integrals in the parametric representation~\eqref{parametric-integral-projective}.

		\subsection{(Quasi-)finite Feynman integrals}
		\label{quasifinite}	The parametric representation of Feynman integrals in
		eq.~\eqref{parametric-integral-projective}  is a special
		case of Euler--Mellin integrals. Following ref.~\cite{2011arXiv1103.6273B} we will quote a theorem that states the convergence region of an Euler--Mellin integral.
		
		\begin{defn}[Euler--Mellin integrals]
        Let $f_1, \dots, f_q$ be Laurent polynomials in variables $x_1,\dots,x_n$, $f_i(\undernotation x):=\sum_{\undernotation a \in
				\supp(f_i)} c_i(\undernotation a)  \undernotation
			x^{\undernotation a}$. An Euler--Mellin integral is defined
			as
			\begin{equation}\label{e:Mellin}
				I(\undernotation{\nu},\undernotation t):=    \int_{\mathbb R^n_+}  {x_1^{\nu_1} \cdots x_n^{\nu_n}\over \prod_{i=1}^q (f_i(\undernotation x))^{t_i}} \, \prod_{j=1}^n {\dd x_j\over x_j}.
			\end{equation}
		\end{defn}
		\begin{defn}[Non-vanishing polynomial]
			If $F$ is a face of the Newton polytope
			$\Newton(f)$ of $f$, then the truncated polynomial with
			support $F$ is given by $f_F:=\sum_{\undernotation a\in
				F \cap \supp(f)} c(\undernotation a)\, \undernotation
			x^{\undernotation a}$. The polynomial $f$ is said to be completely
			non-vanishing on a set $X$ if for each face $F$ of
			$\Newton(f)$, the truncated polynomial $f_F$ has no zeros
			on $X$. In particular, the polynomial $f$ itself does not
			vanish on $X$.  
		\end{defn}

		\begin{theorem}[Theorem~2.2 of ref.~\cite{2011arXiv1103.6273B}]\label{thm:conv}
			If each of the polynomials $f_1,\dots,f_q$ is completely non-vanishing on the positive
			orthant $\mathbb R^n_+$, then the Euler--Mellin integral $ I(\undernotation \nu,\undernotation t) $ of~\cref{e:Mellin} converges and
			defines an analytic function in the tube domain
			\begin{equation}
				\label{e:DomainConvergence}
				\left\{  (\undernotation{\nu},\undernotation t) \in \mathbb C^{n+q} \, \middle| \,
				\Re(\undernotation t) \in \mathbb R_+^q, \, 
				\Re(\undernotation{\nu})\in{\mathrm{int}}\left(\sum_{i=1}^q \Re( t_i) \Newton(f_i)\right) \right\} .
			\end{equation}
			
		\end{theorem}
		We shall apply this theorem to the parametric representation in
		\cref{parametric-integral-projective} with $\undernotation
		t=(n_\cU,n_\cF)$ and the polytope 
		$\Delta_\Gamma(L,D;\undernotation{\nu})$ in \cref{e:NewtonUF} attached to a Feynman graph
		$\Gamma$, with the exponents $\nu_i\geq1$ for $1\leq i\leq \nprops$.
        
		\subsubsection{Convergence of Feynman integrals}
        We will apply the above theorem for integer values of all the quantities $\undernotation \nu$, $\nprops$, $D$, $L$. The convergence of the integrals 
		is then controlled by $D$, $\undernotation \nu$, $n_\cF$. We also assume that the exponent $n_\cU$ is positive. Otherwise it can be taken as part of the numerator~\cite{delaCruz:2024xsm}. 
                The positivity condition on the
		powers of the Symanzik polynomials, and taking into
                account that the powers of the propagators are
                strictly positive integers $\nu_i\geq1$, leads to the following constraints
		\begin{align}
                  \label{e:convParam}
                  n_\cU&\geq0, \qquad n_\cF\geq0,  \qquad n_\cU+n_\cF={D\over2},\cr
                  \nu_1+\cdots+\nu_{\nprops}&=(L+1)n_\cF+L n_\cU,\qquad
            \nu_i\geq1 \qquad 1\leq i\leq\nprops,\\
\nonumber                         	(\nu_1,\dots,\nu_\nprops)  &\in \textrm{int}\left(	\left({D\over2}-n_\cF\right)\Newton(\cU_\Gamma)+n_\cF \Newton(\cF_\Gamma)\right)\,.     
		\end{align}

     The finite integral reads
		\begin{equation}\label{e:Fmassive}
			I_\Gamma(L, D;\undernotation{\nu}):=
                        \int_{\mathbb R_+^{\nprops-1}}
		\left.	{ \prod_{i=1}^\nprops x_i^{\nu_i}\over \cU_\Gamma^{{D\over2}-n_\cF}
				\cF_\Gamma^{n_\cF}}\right|_{x_\nprops=1} \prod_{i=1}^{\nprops-1} {\dd x_i\over x_i}, 
		\end{equation}
with $(\nu_1,\dots,\nu_\nprops)  \in
\textrm{int}\left(\Newton_\Gamma(L,D;\undernotation{\nu})\right)$. 
		When all  the internal masses are non-vanishing, the interior point
		condition in \cref{e:Fmassive} takes a simpler form
                thanks to the relation in \cref{e:FMinkowski}
		\begin{equation}
		\Newton_\Gamma(L,D;\undernotation{\nu})= {D\over2}\Newton(\cU_\Gamma)+n_\cF \simplexwithpars{\nprops,1}\,.
		\end{equation}
        Therefore the properties of the polytope associated with the Feynman
		integral does not require to know the second Symanzik
		polynomial in full. 
        In addition, we will consider even spacetime dimensions to focus on integer lattice polytopes.        This will be particularly useful when scanning for
		reflexivity of the polytope attached to a fully massive Feynman integral.

		In the case of vanishing  internal masses
		$(m_1,\dots,m_n)=(0,\dots,0)$ we have
		\begin{equation}\label{e:Fmassless}
			 I_{\Gamma,\undernotation 0} (L, D;\undernotation{\nu}):=  \int_{\mathbb R_+^{\nprops-1}} \left.
			{ \prod_{i=1}^\nprops x_i^{\nu_i}\over \cU_\Gamma^{{D\over2}-n_\cF}
				\cV_\Gamma^{n_\cF}}\right|_{x_\nprops=1} \prod_{i=1}^{\nprops-1} {\dd x_i\over x_i}
        \end{equation}
with $(\nu_1,\dots,\nu_\nprops)  \in  \textrm{int}\left(\left({D\over2}-n_\cF\right)\Newton(\cU_\Gamma)+n_\cF \Newton(\cV_\Gamma)\right)$. 
                
		\subsubsection{Quasi-finite integrals}\label{sec:quasifinite}
		The integral $I_\Gamma(L, D; \undernotation{\nu})$ differs from the
		parametric representation of a Feynman integral in \cref{e:Itensor} by the $\Upgamma$-function prefactor
		\begin{equation}
			c_\prefactorGamma=
			{\prefactorGamma(\nu_1+\cdots+\nu_\nprops-{LD\over2})\over \prod_{
					i=1}^\nprops\prefactorGamma(\nu_i)}.
		\end{equation}
		\Cref{thm:conv} gives the conditions  for the integral
                in \cref{parametric-integral-projective} to be finite
                for a given integer value of the dimension $D$, 
		but the prefactor $c_\prefactorGamma$ could have a pole in that
                dimension. The numerator does not vanish because all
                the exponents $\nu_i$ are greater than or equal to 1.
                
		Feynman integrals $I_\Gamma(\undernotation{\nu},D)$ that have at worst a
		$1/\epsilon$ divergence from the $\Upgamma$-function prefactor are known as quasi-finite Feynman integrals~\cite{vonManteuffel:2014qoa}. In this work, we will
		consider quasi-finite integrals as well.        For instance, the $\nprops$-point one-loop scalar Feynman integral in $D=2\nprops-2\epsilon$ dimensions is a quasi-finite integral. It 
		reads
		\begin{equation}
		\int \prod_{i=1}^{\nprops}{1\over
                  (\ell+\sum_{j=1}^i p_j)^2-m_i^2+\ii\varepsilon}
                {\dd^{2\nprops-2\epsilon}\ell\over \ii\pi^{\nprops-\epsilon}}=
                        \Upgamma(\epsilon) \, I_{\nprops-\rm gon}(1,2\nprops-2\epsilon, (1,\dots,1))
                      \end{equation}
                      with $p_1+\cdots+p_\nprops=0$ and
		where $ I(1, \nprops-2\epsilon, (1,\dots,1))={1\over
			(\nprops-1)!}+O(\epsilon)$ is a finite
                      integral (see \cref{e:evalUngon} for details
                      about this integral).		
		
		\subsubsection{Interpretation of the interior point condition}
		\label{sec:class11n}

                	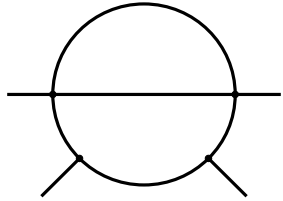
\begin{figure}[h]
			\centering
			\begin{tikzpicture}[scale=0.6]
				\filldraw [color = black, fill=none, very thick] (0,0) circle (2cm);
				\draw [black,very thick] (-2,0) to (2,0);
				\filldraw [black] (2,0) circle (2pt);
				\filldraw [black] (-2,0) circle (2pt);
				\filldraw [black] (1.414,-1.414) circle (2pt);
				\filldraw [black] (-1.414,-1.414) circle (2pt);
				\draw [black,very thick] (-2,0) to (-3,0);
				\draw [black,very thick] (2,0) to (3,0);
				\draw [black,very thick] (1.414,-1.414) to (2.25,-2.25);
				\draw [black,very thick] (-1.414,-1.414) to (-2.25,-2.25);
			\end{tikzpicture}
			\caption{A two-loop graph with edge label $(1,1,3)$.}\label{fig:11n}
		\end{figure}
		
		We illustrate the importance of the interior point condition for finiteness with an example of the two-loop graph with edge label $(1,1,3)$ (using the notation
                of ref.~\cite{Doran:2023yzu}, see also \cref{app:TwoloopGraph}) with four external legs
                attached to one of the lines of the skeleton graph. 
		
		 The momentum-space representation of this integral is
		\begin{equation}
		\mathcal{I}_{(1,1,3)}=\int_{\mathbb R^{2D}}\frac{1 }
			{		(\ell_1^2-m_{1}^2)^{\nu_1} ((\ell_1-\ell_2)^2-m_{2}^2)^{\nu_2}\prod_{i=1}^3
				(    (\ell_2+\sum_{j=1}^i p_j)^2-m_{2+i}^2)^{\nu_{2+i} }  }
			\frac{\dd^{D}\ell_1}{\ii \pi^{D\over 2}}	\frac{\dd^{D}\ell_2}{\ii \pi^{D\over2}}	
		\end{equation}
		with $p_1+\cdots+p_4=0$.
		The parametric representation of this integral is given by $\mathcal I_{(1,1,3)}=\Upgamma(\nu_1+\cdots+\nu_5-D)
		\,  I_{(1,1,3)}(2,D;\undernotation{\nu})$, where the integral without the $\Upgamma$-factor is
		\begin{equation}\label{e:I11n}
			 I_{(1,1,3)}(2,D;\undernotation{\nu})=
			\int_{\mathbb R_+^{4}} \left. x_1\cdots x_{5}\over
				\cU_{(1,1,3)}^{{3D\over2}-\sum_{i=1}^5\nu_i}\cF_{(1,1,3)}^{\sum_{i=1}^5\nu_i-D}\right|_{x_{5}=1}
                              \prod_{i=1}^{4}{\dd x_i\over x_i}\,.
		\end{equation}
        The Symanzik polynomials are obtained using the formulas of \cref{app:TwoloopGraph}
        \begin{align}
            \cU_{(1,1,3)}&=(x_1+y_1)(z_1+z_2+z_3)+x_1y_1,\cr
            \cV_{(1,1,3)}&=x_1 \sum_{1\leq i<j\leq 3} c_{ij}^x z_iz_j+y_1 \sum_{1\leq i<j\leq 3} c_{ij}^y z_iz_j+x_1y_1\sum_{i=1}^3 c_i^{xy} z_i ,\cr
            \cF_{(1,1,3)}&=\cU_{(1,1,3)}(m_1^2x_1+m_2^2y_1+m_3^2z_1+\cdots +m_5^2z_3)-\cV_{(1,1,3)},
        \end{align}
        where the kinematics coefficients  $c^x_{ij}$, $c^y_{ij}$ and $c^{xy}_i$ are linear combinations of kinematic invariants. We consider the internal masses $m_i$ and the kinematic coefficients generic and non-vanishing. 
		The positivity condition on the powers of the Symanzik
                polynomials gives
                \begin{equation}\label{e:pos113}
                {3D\over2}\geq \sum_{i=1}^5\nu_i\geq D,
                \end{equation}
                which
                gives several solutions. 
                When all powers are equal to one $\nu_1=\cdots=\nu_5=1$, then  $\sum_{i=1}^5\nu_i=5$. The positivity condition on the power of the Symanzik polynomial has a unique solution, that is, $D=4$. The resulting integral $I_{(1,1,3)}(2,4;(1,1,1,1,1))$ 
        diverges. Indeed one can check that the associated graph polytope $\Newton(\cU_{(1,1,3)})+\Newton(\cF_{(1,1,3)})$ does not have any interior points.
        But for the case when $\sum_{i=1}^5\nu_i=7$, the positivity in \cref{e:pos113} has a unique solution $D=6$.  The associated  polytope $2\Newton(\cU_{(1,1,3)})+\Newton(\cF_{(1,1,3)})$ has three interior points $(3,1,1,1,1)$, $(1,3,1,1,1)$ and $(2,2,1,1,1)$, and the resulting integrals $I_{(1,1,3)}(2,6;\undernotation{\nu})$ are finite for each choice of these interior points.  
        
        The interior point condition controls the ultraviolet behavior of the  one-loop integral over $\ell_1$ in \cref{e:I11n}. This integral is free of ultraviolet divergences  if the powers $\nu_1$ and $\nu_2$  satisfy $2(\nu_1+\nu_2)>D$.

                 \subsection{Counting the number of interior points}\label{sec:NBinterior}

                 The Ehrhart polynomial~\cite{Haase:2012,Beck:2015} is a fundamental tool in discrete geometry for enumerating the
lattice points contained in integer dilations of a convex polytope. 
For a lattice polytope \( P \subset \mathbb{R}^n \), the Ehrhart polynomial \( \ehr_P(t) \)
counts the number of integer lattice points in the dilated polytope \( tP \), that is,
\begin{equation}
\ehr_P(t) = \# \big( tP \cap \mathbb{Z}^n \big)\, .
\end{equation}
The number of interior lattice points is obtained from Ehrhart--Macdonald reciprocity,
\begin{equation}
\label{eq:n_int_ehr}
 \ehr_P(-t) = (-1)^{\dim P} \nint(tP) \, ,
\end{equation}
where
\begin{equation}
\nint(Q) := \# \big( \mathrm{int}(Q) \cap \mathbb{Z}^n \big)
\end{equation}
is the number of interior points of a polytope~$Q$.
In particular, \( \ehr_P(1) \) and \(\lvert \ehr_P(-1) \rvert\) give the total number of lattice and interior lattice points in \( P \), respectively.
Thus, the Ehrhart polynomial encodes both the boundary and interior
point structure of the polytope. We will use this in \cref{sec:upperdim} in order to put an upper 
bound on the spacetime dimensions of Feynman integrals when searching for reflexive and Fano polytopes. The knowledge of the Ehrhart polynomial gives an efficient way to identify polytope dilations containing a single interior point.

\medskip

The polytopes associated to a Feynman integral in
\cref{e:NewtonUFMassive,e:NewtonUFMassless} are the Minkowski sum of
the scaling of  two lattice polytopes:
$P$, which is 
either $\Newton(x_1+\cdots+x_\nprops) = \simplexwithpars{\nprops,1}$ for the massive case or $\Newton(\cV_\Gamma)$ for the massless case, and $Q=\Newton(\cU_\Gamma)$.
The number of lattice points depends on the scaling coefficients of both polytopes. 
To capture this dependence, one introduces the bivariate Ehrhart
polynomial \( \ehr_{P,Q}(t_1,t_2) \), defined by (see, e.g.~\cite[Chap.~19.1]{Gruber2007})
\begin{equation}
\ehr_{P,Q}(t_1,t_2) = \# \big((t_1P+t_2 Q) \cap \mathbb{Z}^n \big)\,,
\end{equation}
which counts the lattice points in the Minkowski sum of independently
dilated copies of \(P\) and \(Q\). The bivariate Ehrhart polynomial also satisfies a reciprocity relation:
\begin{equation}\label{e:countInterior}
\ehr_{P,Q}(-t_1,-t_2) = (-1)^n\, 
\nint(t_1P+ t_2Q)\,.
\end{equation}
We will use it to identify cases when Newton polytopes associated to Feynman graphs contain a single interior point.

In general, determination of the bivariate
Ehrhart polynomial is a difficult problem. The
Bernstein--McMullen's theorem~\cite[Theorem~19.4]{Gruber2007} implies that \(\ehr_{P,Q}(t_1,t_2)\) is a polynomial of total degree \(n\):
\begin{equation}
\ehr_{P,Q}(t_1,t_2) = \sum_{i=0}^{n} \sum_{j=0}^{n-i} c_{ij}\, t_1^i t_2^j,
\end{equation}
with rational coefficients \(c_{ij}\).  One can use this polynomial
form to determine the bivariate Ehrhart polynomial by computing the number of lattice points for various dilations \(t_1 P + t_2 Q\) using {\tt
  polymake}~\cite{polymake:FPSAC_2009} and fitting the coefficients~\(c_{ij}\).
We list a few properties that help to compute the Ehrhart polynomial~\cite{Beck:2015,Haase:2017,BrandenburgEtAl:2020}:
\begin{itemize}
  \item 
The leading homogeneous part (degree \(n\)) corresponds to the mixed volumes:
\begin{equation}
\mathrm{Vol}(t_1P + t_2 Q)
= \sum_{i=0}^{n} \binom{n}{i} V_i(P,Q)\, t_1^i t_2^{n-i},
\end{equation}
with $V_i(P,Q)$ the mixed volumes of $i$ copies of $P$ and $n-i$ copies of $Q$, so that $V_0(P,Q)=V(Q)$, $V_n(P,Q)=V(P)$ and $V_i(P,P)=V(P)$. 
\item Setting one of the variables to zero gives the corresponding univariate Ehrhart polynomial:
\begin{equation}
\ehr_{P,Q}(t_1,0) = \ehr_P(t_1), \qquad \ehr_{P,Q}(0,t_2) = \ehr_Q(t_2).
\end{equation}
\end{itemize}

\section{Fano and reflexive polytopes from Feynman integrals}\label{sec:reflexivefano}

		\newcommand{			\PolarPolytope}{\nabla}

We introduce the Fano and reflexive polytopes that we use in this work. 
  We refer to~\cite{CoxKatz1999,Nill2005,Kasprzyk12,Telen:2022}  for detailed discussions of the notions used in
		this work.
        
		\begin{defn}[Fano polytope~\cite{Nill2005, Kasprzyk12}]
			A lattice polytope is (canonical) Fano if the only lattice point that lies strictly
			in its interior is the origin.
		\end{defn}
		Let $N$ be
		the  dual lattice to $M$ with respect to the scalar
		product $\langle \bullet| \bullet\rangle: M\times N\to \mathbb Z$.
		Using this scalar product we can define the dual (or polar) polytope
		\begin{equation}
			\PolarPolytope:=\left\{\undernotation b=(b_1,\dots,b_n)\in N \, \middle| \, \langle
			a,b\rangle\geq-1, \, \forall \undernotation a\in \Delta\right\}\subset N\cong \mathbb Z^n\,.  
		\end{equation}

		\begin{defn}[Reflexive polytope]
			A  lattice polytope $\Delta$ is reflexive if its polar
			polytope $\PolarPolytope$ is a lattice polytope and has a single interior
			point. The polytopes $(\Delta,\PolarPolytope)$ are said to be
			mirror pairs.
			
		\end{defn}

	We make a few remarks:
        \begin{itemize}
          \item 	Newton polytopes of Laurent polynomials, defined in \cref{e:DeltaLaurentDef}, are integer lattice polytopes, i.e.\ they are included in
		some lattice $M$ of $\mathbb Z^n$. 
	Reciprocally, to a given integer lattice 
		polytope $\Delta$ one associates a Laurent polynomial
		$ f_{\Delta}=\sum_{\undernotation{\nu}\in\Delta} c_{\undernotation{\nu}} \undernotation
		x^{\undernotation{\nu}}$, with a similar definition for the dual polytope $\PolarPolytope$. This will be used in \cref{sec:mirror} when discussing the relation between polytopes from Feynman integrals and smooth Calabi--Yau varieties.

					\item The integrand of the
                                          Feynman integral~\eqref{parametric-integral-projective} is the inverse of the
			Laurent polynomial
			\begin{equation}
				f_\Newton=\frac{ 
					\cU_\Gamma^{{D\over2}-n_\cF} \cF_\Gamma^{n_\cF}} {x_1^{\nu_1}\cdots x_\nprops^{\nu_\nprops}} \, .
			\end{equation}
			When there is a single interior point $\undernotation{\nu}=(\nu_1,\dots,\nu_\nprops)$, the
			presence of the monomial in the denominator
                        (or the numerator of the Feynman integral) shifts the polytope so
			that the interior point becomes the origin.
            \item 
			Every reflexive polytope is Fano, but there are Fano polytopes that are
			not reflexive. In \cref{sec:dimloopscan,sec:edgescan} we list quasi-finite Feynman integrals whose associated polytopes are Fano, but not reflexive, for graphs up to ten edges. 
			
		\end{itemize}

		Together with the conditions in \cref{e:convParam}, we are now imposing the condition that the Newton polytope has only one interior point, and therefore it is Fano. 
        Among such polytopes, we then determine whether the polytope is reflexive.

		\subsection{Example: sunset integrals}\label{sec:sunset}

        	\begin{figure}[h]
			\centering
			\begin{tikzpicture}[scale=0.7]
				\draw[very thick] (0,0) ellipse (3cm and 2cm);
				\draw[very thick] (0,0) ellipse (3cm and 1.cm);
				\draw [very thick](0,0) ellipse (3cm and 0.5cm);
				\draw [very thick] (-3,0)--(-3.5,0);
				\draw [very thick] (3,0)--(3.5,0);
                \filldraw [black] (3,0) circle (2pt);
                \filldraw [black] (-3,0) circle (2pt);
				\node[text width=0.5cm, text centered ] at (0,0.1) {$\vdots$};
				\node [text width=0.5cm, text centered ] at (0,-2.3){$m_1$};
				\node [text width=0.5cm, text centered ] at (0,-1.3){$m_2$};
				\node [text width=0.5cm, text centered ] at (0,1.3){$m_{\nprops-1}$};
				\node [text width=1.5cm, text centered ] at (0,2.3){$m_{\nprops}$};
			\end{tikzpicture} 
			\caption{Multiloop sunset graph.}
			\label{cantaloupe-diagram}
		\end{figure}
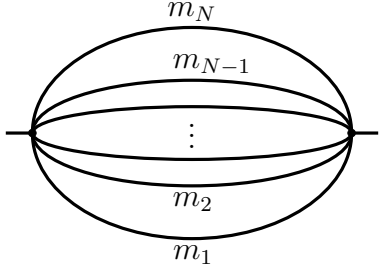

		The basic example of reflexive polytopes arising from
        Feynman integrals is the family of multiloop sunset graphs
        depicted in \cref{cantaloupe-diagram}, with $\nprops \ge 2$ edges ($\nprops=2$ corresponds to the one-loop bubble graph).
        They have the following Symanzik polynomials:
		\begin{align}
			\cU^{\nprops}_{\circleddash}&=\fp_1 \fp_2 \cdots \fp_{\nprops} \left(\frac {1}{\fp_1}+\dots + \frac {1}{\fp_{\nprops}}\right)\, ,\cr
			\cF^{\nprops}_{\circleddash}&= 	\cU^{\nprops}_{\circleddash}\,\left(m_1^2 \fp_1 + \dots + m_{\nprops}^2 \fp_{\nprops} \right) - p^2\fp_1 \cdots \fp_{\nprops}\,.\label{secondSymanzik-banana}
		\end{align}
This family is special because both of these polynomials can be written as matrix determinants.
For the first Symanzik polynomial this is a direct
consequence of its definition:
\begin{equation}
  \cU^{\nprops}_{\circleddash}= \det\left(
    \begin{pmatrix}
      x_1&0&\cdots&0\cr
      0&x_2&\cdots&0\cr
      \vdots&\cdots&\ddots&\vdots\cr
      0&0&\cdots &x_{\nprops-1}
    \end{pmatrix}
        + x_{\nprops}
      \begin{pmatrix}
        1&\cdots& 1\cr
        \vdots&\cdots&\vdots\cr
        1&\cdots& 1
      \end{pmatrix}
\right)\,.
\end{equation}
The determinantal representation for the second Symanzik polynomial is
\begin{equation}
  \cF^{\nprops}_{\circleddash}= - p^2\det\left(\begin{pmatrix}
      x_1&0&\cdots&0\cr
      0&x_2&\cdots&0\cr
      \vdots&\cdots&\ddots&\vdots\cr
      0&0&\cdots &x_{\nprops}
    \end{pmatrix}-
    {m_1^2x_1+\cdots+m_{L+1}^2x_{\nprops}\over p^2}
      \begin{pmatrix}
        1&\cdots& 1\cr
        \vdots&\cdots&\vdots\cr
        1&\cdots& 1
      \end{pmatrix}\right) \, .
  \end{equation}
               The associated Feynman integral (without the $\Upgamma$-prefactor)
                is
                \begin{equation}
                   I_{\circleddash}(\nprops-1,D;\undernotation{\nu})= \int_{\mathbb R_+^{N-1}} {x_1^{\nu_1}\cdots
                    x_{\nprops}^{\nu_{\nprops}}\over
                    (\cU^\nprops_{\circleddash})^{{D\over2}-n_\cF}
                    (\cF^\nprops_{\circleddash})^{n_\cF}                  }\Bigg|_{x_{\nprops}=1}
                  \prod_{i=1}^{\nprops-1}{\dd x_i\over x_i}\,,
                \end{equation}
                with
                $n_\cF=\nu_1+\cdots+\nu_{\nprops}-{(\nprops-1)D\over2}$ because the sunset graph has $L=\nprops-1$ loops.              %
In the generic kinematic case, thanks to \cref{e:FMinkowski} the Newton 
polytope is the scaled Minkowski sum of the  Newton polytope of the first Symanzik polynomial and the standard simplex
\begin{equation}\label{e:SunsetPolytope}
  \Delta_{\circleddash}(\nprops-1,D;\undernotation{\nu})={D\over2}\Newton(	\cU^{\nprops}_{\circleddash})+\left(\nu_1+\cdots+\nu_{\nprops}-{(\nprops-1)D\over2}\right)\simplexwithpars{\nprops,1}\,.
\end{equation}
The Newton polytope for the first Symanzik polynomial is 
\begin{equation}
  \Newton(\cU_{\circleddash}^{\nprops})=\Newton\left(\sum_{i=1}^\nprops \prod_{1\leq
    j\leq N \atop j\neq i} x_j\right) =  \Conv\set{\mathbf 1-e_i \, | \, i=1, \dots, \nprops} =\mathbf 1-\simplexwithpars{\nprops,1}\,,
    \label{simplex-sunset}
\end{equation}
with $\mathbf 1=(1,\dots,1)$.
\medskip

The number of interior points can be computed by applying
 the reciprocity relation of~\cref{e:countInterior} to the bivariate Ehrhart polynomial of the multiloop sunset polytope, which we derive in \cref{app:sunsetEhrhart}:
\begin{equation}\label{e:EhrSunsetResultMain}
    \ehr_{\circleddash}(n_\cF,D/2,\nprops)=P(n_\cF,D/2,\nprops)+\sum_{r=0}^{\nprops-1} c(r,\nprops) P(n_\cF,D/2,r)\,,
\end{equation}
with $P$ being a polynomial defined as
\begin{equation}
  P(t_1,t_2,n):=  \sum_{0\leq i,j\leq n-1\atop i+j\neq
    n-1}(-1)^{-i-j+n-1} \binom{i+j}{i} \binom{i+t_1}{i} \binom{j+t_2}{j} \,, 
\end{equation}
and the coefficients $c(r,n)$ are given by
\begin{equation}\label{e:EhrSunsetCoeff}
 c(r,n)=(-1)^n \textrm{coeff}_{x^{n-1}}   \left({2 x+1\over (1+x)^2}  \left(\frac{x}{x+1}\right)^{n-1-r}\right),
\end{equation}
where $\textrm{coeff}_{x^{r}}(f(x))$ means the coefficient of
$x^r$ in the series expansion of $f(x)$ around $x=0$.
We list the Ehrhart polynomials $\ehr_{\circleddash}(n_\cF,D/2,N)$ for up to $N=10$ edges in \cref{tab:EhrSunset}.

Using the reciprocity formula for the Ehrhart polynomial in \cref{e:EhrSunsetResultMain} we identify the polytopes with a single interior point for generic number of edges:
\begin{itemize}
\item {$D=2$ with $\undernotation{\nu}=(1,\dots,1)$:}  in this case
  $n_\cU=0$ and $n_\cF=1$, and the polytope is
  \begin{equation}\label{e:sunsetVerrill}
    \Delta_{\circleddash}(\nprops-1,2;(1,\dots,1))=\Newton(	\cU^{\nprops}_{\circleddash})+\simplexwithpars{\nprops,1}\,.
  \end{equation}

\item {$D=2\nprops$ with $\undernotation{\nu}=(\nprops-1,\dots,\nprops-1)$:}
in this case $n_\cU=\nprops$ and $n_\cF=0$, and the polytope reads
  \begin{equation}\label{e:polysunsetN}
    \Delta_{\circleddash}(\nprops-1,2\nprops;(\nprops-1,\dots,\nprops-1))=\nprops\Newton(	\cU^{\nprops}_{\circleddash})=\nprops\left(\mathbf{1}-\simplexwithpars{\nprops,1}\right)\,.
  \end{equation}
\end{itemize}

Both polytopes are known to be reflexive. The polytope of \cref{e:sunsetVerrill} is the one of the $A_{\nprops-1}$ root lattice (see \cref{e:RootAnpolytope}) that 
  has been studied in detail in ref.~\cite{verrill1996root}. 	%
  We refer to \cref{app:sunsetpolytope} for a discussion of the properties of this polytope  and \cref{sec:sunsetmirror} for a relation to Calabi--Yau geometry.  The polytope in \cref{e:polysunsetN} is a translation of the scaled simplex, therefore it is reflexive with interior point $(N,\dots,N)$.

As a side remark, we notice that the above  polytopes are saturated, because they are Newton polytopes of
determinants~\cite{Cara:2019}.

\section{Search by integral dimension and loop count}
\label{sec:dimloopscan}

 Feynman integrals are characterized by the number of
 loops $L$, the number of internal edges $\nprops$ of
 the graph and the dimension $D$ where the integral is
 evaluated.  These enter into the definition of the
 Newton polytope given in \cref{e:NewtonUF}. The
 exploration of reflexivity can be done by fixing any
 of these parameters.
 
 A common approach in physics is going up in the number of
 loops, increasing precision of theoretical predictions.
 Another possibility is to move up in dimension $D$.
 Fixing either of these parameters to their lowest possible values, $D=2$ and $L=1$, we can fully classify the resulting Fano and reflexive polytopes.
 We do this in \cref{sec:twodim} and \cref{sec:oneloop}, respectively.
 In \cref{sec:fourdimtwoloop} we consider the ``next-to-simplest'' case by setting $D=4$ and $L=2$ simultaneously.
 
 In the mathematics literature the natural parameter to fix is $\nprops$, which corresponds to fixing the dimension of the space where the polytope lives.
 This was the strategy followed in ref.~\cite{Kreuzer:2000xy} to classify reflexive
 polytopes living in four dimensions, which corresponds
 to graphs with up to $\nprops=5$ edges.  
 We follow this approach in \cref{sec:edgescan}, performing an exhaustive scan of Feynman integrals with up to $\nprops=10$ edges.
	
\subsection{$D=2$, $L=\mathrm{any}$: integrals in two dimensions}%
\label{sec:twodim}

		In $D=2$ dimensions, the conditions in \cref{e:convParam}
		give two
		solutions, namely $\sum_{i=1}^\nprops \nu_i=L+n_\cF$ with $n_\cF=0,1$. Since $\nu_i\geq1$,
		and the Euler characteristic of the graph gives $V-\nprops+L=1$, where
		$V\geq1$ is the number of vertices, we deduce that
		\begin{equation}
			L+n_\cF=\sum_{i=1}^\nprops \nu_i\geq \nprops \geq L, \qquad n_\cF=0,1  \, .
		\end{equation}
		\begin{itemize}
			\item For $n_\cF=0$, the only possibility is $\nprops = L$, $V = 1$, and $\nu_i=1$ for $1\leq i\leq \nprops$. 
			This case corresponds to a factorizable $\nprops$-bouquet Feynman graph, that is, a product of $\nprops$ tadpoles. The Symanzik polynomials are
			\begin{equation}
				\cU_{\nprops\mathrm{-bouquet}} = x_1\cdots x_\nprops,\quad \cF_{\nprops\mathrm{-bouquet}}=
				x_1\cdots x_\nprops\left(m_1^2x_1+\cdots +m_\nprops^2x_\nprops\right).
			\end{equation}
			The Newton polytope $\Newton(\cU_{\nprops\mathrm{-bouquet}})$ is a single point and thus does not have any interior points.
            Physically, this is due to the fact that tadpole integrals diverge in $D=2$.
            
			\item For $n_\cF=1$, there are two options.
            The first is $\nprops = L$, $V = 1$, $\nu_j = 2$ for some $j$, and $\nu_i = 1$ for $1 \le i \le \nprops, i \ne j$.
            This is again the case of an $\nprops$-bouquet, but now the Newton polytope is $\Newton(\cF_{\nprops\mathrm{-bouquet}})$.
            The point~$(\nu_1, \dots, \nu_\nprops)$ does not lie in the interior, hence the integral is again divergent and the associated polytope is not Fano.
            The second option is $\nprops = L+1$, $V = 2$, and $\nu_i = 1$ for $1 \le i \le \nprops$.
            This corresponds to the $L$-loop sunset integral which has been discussed in \cref{sec:sunset}.
            We conclude that this is the only class of Feynman graphs that leads to Fano and reflexive polytopes in $D=2$.
		\end{itemize}

\subsection{$D=\mathrm{any}$, $L = 1$: one-loop integrals}%
\label{sec:oneloop}

The one-loop $\nprops$-gon integrals have exponents
\begin{equation} n_\cU=-(\nu_1+\cdots+\nu_\nprops)+D, \quad n_\cF=(\nu_1+\cdots+\nu_\nprops)-{D\over2}
\end{equation}
with the  integrals taking the form
\begin{equation}
  I_{\nprops-\rm gon}(1,D; \undernotation{\nu} )=
  \int_{\mathbb R_+^{\nprops-1}} \left. {\prod_{i=1}^{\nprops} x_i^{\nu_i}\over \cU_{\nprops-\rm
      gon}^{{D\over2}-n_\cF} \cF_{\nprops-\rm
      gon}^{n_\cF}}\right|_{x_\nprops=1} \prod_{i=1}^{\nprops-1} {\dd x_i\over x_i}\,.
\end{equation}
Their Symanzik polynomials  read
\begin{equation}\label{e:ULFoneloop}
        \begin{aligned}
			\cU_{\nprops-\rm gon}&:=x_1+  \cdots +x_\nprops,\cr
			\cL_{\nprops-\rm gon}&:=m_1^2x_1+\cdots +m_\nprops^2x_\nprops,\cr
			\cV_{\nprops-\rm gon}&:=\sum_{1\leq i<j\leq \nprops}
       (p_{i}+\cdots+p_{j-1})^2 x_ix_j,\cr
                                               \cF_{\nprops-\rm gon}&:=\cU_{\nprops-\rm gon}\cL_{\nprops-\rm gon}-\cV_{\nprops-\rm gon}\,,
                                               \end{aligned}
		\end{equation}
		 where we use a cyclic labeling modulo
		$\nprops$ of the massive external momenta, which are
                constrained by momentum conservation $p_1+\cdots+p_\nprops=0$.  The convergence conditions imply that
\begin{equation}\label{e:nuvsD}
  D\geq   \nu_1+\cdots+\nu_\nprops\geq {D\over2}\,,
\end{equation}
together with the condition that the exponents lie in the interior
\begin{equation}
   \undernotation{\nu}= (\nu_1,\dots,\nu_\nprops)\in\textrm{int}(	\Delta_{\nprops-\rm gon}(D,1;\undernotation{\nu}))\,.
  \end{equation}

  \subsubsection{Massive case}\label{sec:onelooppolytope}

For the one-loop $\nprops$-gon with massive internal lines the polytope is  
	\begin{equation}\label{e:NewtonUFoneloop}
		\Delta_{\nprops-\rm gon}(1,D;\undernotation{\nu})=
			n_{\cU}\Newton(\cU_{\nprops-\rm gon})+n_\cF \Newton(\cF_{\nprops-\rm gon})=
			(\nu_1+\cdots+\nu_\nprops)\simplexwithpars{\nprops,1},
		\end{equation}
where we used \cref{e:FMinkowski} and the relation
\begin{equation}
\Newton(\cU_{\nprops-\rm
  gon})=\Newton(\cL_{\nprops-\rm
  gon})=\Newton(x_1+\cdots+x_{\nprops})=\simplexwithpars{\nprops,1}\,,
  \end{equation}
and that $n_\cU+2n_\cF=\nu_1+\cdots+\nu_\nprops$.
Therefore, the polytope for the one-loop graph with massive internal lines is the scaled standard simplex defined in \cref{e:standardsimplex-definition}. For this polytope it is immediate to compute the
Ehrhart polynomial~\cite{Beck:2015} because it is a scaled standard simplex:
\begin{equation}
  \ehr_{ \simplexwithpars{\nprops,1}}(t)=
  \prod_{r=1}^{n-1}{t+r\over r}\,. 
\end{equation}
The number of interior points can be computed using the Ehrhart
reciprocity formula in \cref{eq:n_int_ehr}:
\begin{equation}
  	\nint(k\simplexwithpars{\nprops,1}) =
        (-1)^{\nprops-1} \ehr_{  \simplexwithpars{\nprops,1}}(-k)=\prod_{r=1}^{n-1}{k-r\over r}\,.
      \end{equation}
Therefore the polytope $\Delta_{\nprops-\rm gon}(1,D;\undernotation{\nu})$
has only one interior point for
$k=\nu_1+\cdots+\nu_\nprops=\nprops$. Because $\nu_i\geq1$ for all
$i=1,\dots,\nprops$,  the interior point is
$\undernotation{\nu}=(1,\dots,1)$. This condition does not depend on the
spacetime dimension $D$, and the positivity  condition in
\cref{e:nuvsD} implies that the polytope associated to the one-loop
integral
	\begin{equation}\label{e:UFoneloop}
			 I_{\nprops-\rm gon}(1,D;(1,\dots,1))= \int_{\mathbb R_+^{\nprops-1}} \left. {1\over \cU_{\nprops-\rm
					gon}^{D-\nprops}\cF_{\nprops-\rm gon}^{\nprops-{D\over2}}}\right|_{x_\nprops=1} \dd x_1\cdots \dd x_{\nprops-1}
                                  \end{equation}
                                  is 
reflexive in $\nprops \leq D\leq 2\nprops $ dimensions. This corresponds to the vertical line at $L=1$ in
\cref{fig:PlotReflexiveLvsD}.

\subsubsection{Massless case}
In the massless case the polytope is
                      \begin{equation}
                               \Delta_{\nprops-\rm gon;0}(1,D;\undernotation{\nu})=
			n_{\cU}\Newton(\cU_{\nprops-\rm gon})+n_\cF \Newton(\cV_{\nprops-\rm gon})\,.
                      \end{equation}
The polytope associated with the $\cV_{\nprops-\rm gon}$ graph
polynomial is the standard second hypersimplex defined in \cref{e:secondhypersymplex-definition},
\begin{equation}
  \Newton(\cV_{\nprops-\rm
    gon})=\Newton\left(\sum_{1\leq i<j\leq
      \nprops}x_ix_j\right)=
      \simplexwithpars{\nprops,2} \, .
  \end{equation}
The polytope for the one-loop $\nprops$-gon with massless internal lines is the scaled Minkowski sum of two simplices,
  \begin{equation}\label{e:DeltaOneLoopMassless}
                               \Delta_{\nprops-\rm gon;0}(1,D;\undernotation{\nu})=
			n_{\cU}\simplexwithpars{\nprops,1}+n_\cF  \simplexwithpars{\nprops,2}\,.
                      \end{equation}

The number of interior points of the polytope in
\cref{e:DeltaOneLoopMassless} can be computed using the bivariate
Ehrhart polynomial associated to the Minkowski sum of the polytopes
$\Newton(\cU_{\nprops-\rm gon})=\simplexwithpars{\nprops,1}$ and $\Newton(\cV_{\nprops-\rm
  gon})=\simplexwithpars{\nprops,2}$. The bivariate Ehrhart polytope of the Minkowski sum of the standard simplex
and the second hypersimplex reads (see, e.g.~\cite[Chapter~9]{Beck:2015})
\begin{equation}\label{e:EhPolOneLoopMassless}
  \ehr(t_1,t_2)=
  \sum_{i=1}^{\nprops-1}\sum_{j=1}^{\nprops-1} E_{ij} t_1^i t_2^j \,,
\end{equation}
where $E_{ij}$ are the mixed Ehrhart coefficients. Although each polytope involved in the Minkowski sum is standard,  to the best of our knowledge the bivariate Ehrhart polynomial that we need is unknown.
The expression for the bivariate Ehrhart polynomial is
\begin{equation}\label{e:EhNgonConj}
\ehr_{\Delta_{\nprops-\rm gon;0}}(t_1,t_2)={1\over
       (\nprops-1)!} \prod_{r=0}^{\nprops-2} (t_1+2t_2+r+1)
     -{\nprops\over (\nprops-1)!} \prod_{r=0}^{\nprops-2} (t_2+r)\,,
\end{equation}
which we derive in \cref{app:oneLoopEhrhart}. The number of interior points is evaluated as in the single variable
case.  Using the binomial notation it reads
\begin{equation}
	\begin{aligned}
		\nint(\Delta_{\nprops-\rm gon;0}(1,D;\undernotation{\nu})) &=
		(-1)^{\nprops-1}  \ehr_{\Delta_{\nprops-\rm gon;0}}(-n_{\cU}, -n_{\cF}) \\
		&= \binom{n_\cU + 2n_\cF - 1}{N-1} - N \binom{n_\cF}{N-1}\,.
	\end{aligned}
\end{equation}
There is a single interior point $
\nint(\Delta_{\nprops-\rm gon;0}(1,D;\undernotation{\nu}))=1$
for
\begin{align}
  (n_\cU,n_\cF)&=(3,0), (1,1), (0,3)\qquad \nprops=3,\cr
  (n_\cU,n_\cF)&=(2m-\nprops,\nprops-m)\qquad
                 \left\lceil\nprops\over2\right\rceil\leq m\leq
                 \nprops, \nprops\geq4\,,
\end{align}
so the corresponding polytopes are Fano.
The case $ (n_\cU,n_\cF)=(n,0)$ (i.e.\ $m=\nprops$) is identical to the massive one-loop
case of the previous \namecref{sec:onelooppolytope}, so we already know that the polytope is reflexive.

In \cref{tab:oneloopmassless} we collect all Fano cases up to $\nprops=6$ edges and indicate which of them are reflexive, identified using {\tt
  polymake}~\cite{polymake:FPSAC_2009}. We give the corresponding interior points as well.

	\begin{table}[tb]
			\centering
			\setlength{\tabcolsep}{1em}
			\begin{tabular}{@{}l|r|r|r|r|@{}}
				$N$   & 3 & 4 & 5 &6 \\
                          \hline
                          		$(n_\cU,n_\cF)$
                                      &(3,0) & (4,0)&(5,0) &(6,0)\\
                          &reflexive & reflexive&reflexive
                                                  &reflexive\\
                       interior point   &(1,1,1) & (1,1,1,1)&(1,\dots,1)
                                                  &(1,\dots,1)\\
                          \hline
				$(n_\cU,n_\cF)$  &(1,1) & (0,2)&(1,2)
                                                  &(0,3)\\
                                                     &reflexive & reflexive&Fano
                                                  &Fano\\
                          interior point    &(1,1,1) & (1,1,1,1)&(1,\dots,1)
                                                  &(1,\dots,1)\\
                          \hline
                          				$(n_\cU,n_\cF)$
                                      &(0,3) & (2,1)&(3,1) &(2,2)\\
                          &reflexive & Fano&Fano
                                                  &Fano\\
                         interior point   &(2,2,2) & (1,1,1,1)&(1,\dots,1)
                                                  &(1,\dots,1)\\
                          \hline
                          
                              	$(n_\cU,n_\cF)$         &&&&(4,1)\\
                          
                          & & &
                                                  &Fano\\
                       interior point     & & &
                                                  &(1,\dots,1)\\
                          \hline
                          					\end{tabular}
			\caption{Fano and reflexive polytopes arising from one-loop graphs with massless internal lines, and the corresponding interior points.}
			\label{tab:oneloopmassless}
		\end{table} 

		\subsubsection{Example: $D = 4$}
		\begin{table}[tb]
			\centering
			\setlength{\tabcolsep}{1em}
			\begin{tabular}{ll}
				$n_\cF$ & Solutions   \\
				\midrule
				0 & $(\nu_1,\nu_2 )=(1,1)$  \\
				\addlinespace
				1 & $(\nu_1,\nu_2 )=(1,2)$, $(\nu_1,\nu_2 ,\nu_3)=(1,1,1)$
				\\
				\addlinespace
				2 &$(\nu_1,\nu_2 )=(2,2)$, $(\nu_1,\nu_2 ,\nu_3)=(1,1,2)$,  $(\nu_1,\nu_2,\nu_3,\nu_4 )=(1,1,1,1)$
			\end{tabular}
			\caption{Solutions at one loop in four dimensions. Reflexivity only depends on the sum of propagator powers~$\nu_i$.}
			\label{tab:1loop}
		\end{table} 
As an example, we collect reflexive cases for one-loop graphs in $D=4$ in \cref{tab:1loop}. 		The reflexivity conditions do not depend on the precise value of the powers of the propagators but only on their sum.
		
		\begin{itemize}
			\item The bubble graph: all the cases
                          $n_\cF=0,1,2$ lead to finite
			integrals when the internal propagators are
                        massive. Only the case with $n_\cF = 0$ is
                        reflexive; the other cases $n_\cF=1,2$ have more than one
                        interior point, so they are not Fano.
			\item The triangle graph: only the case
                          $n_\cF=1$ leads to a reflexive
			polytope. The graph polynomials are
			\begin{align}
				\cU_{\rm triangle}&=x_1+x_2+x_3, &\cL_{\rm
					triangle}&=m_1^2x_1+m_2^2x_2+m_3^2x_3,\cr
				\cV_{\rm
					triangle}&=p_3^2x_1x_2+p_2^2x_1x_3+p_1^2x_2x_3,
				&\cF_{\rm
					triangle}&= \cU_{\rm triangle}\cL_{\rm
					triangle}-  \cV_{\rm
					triangle}.                                 
			\end{align}
			
			The Newton polytopes  $\Newton(\cU_{\rm
                          triangle})+\Newton(\cF_{\rm triangle})=3\simplexwithpars{3,1}$ for
                        the massive triangle and
			$\Newton(\cU_{\rm triangle})+\Newton(\cV_{\rm
                          triangle})=\Newton(\cV_{\rm
                          triangle})+\simplexwithpars{3,1}$ for the massless triangle are two-dimensional and
                        reflexive. They are 
			shown in \cref{fig:polytopetriangle} and their relation to smooth elliptic curves is discussed in \cref{massive-triangle,massless-triangle}.

			\begin{figure}[ht]
				\centering
				\begin{subfigure}{0.45\textwidth}
					\centering
					\includegraphics[width=5cm]{polytope-massive-triangle.png}
					\caption{Massive case: 2d reflexive polytope \#15 in
						{\tt sagemath} database~\cite{sagemath}.}
					\label{fig:polytopemassivetriangle}
				\end{subfigure}
                \hfill
				\begin{subfigure}{0.45\textwidth}
					\centering
					\includegraphics[width=5cm]{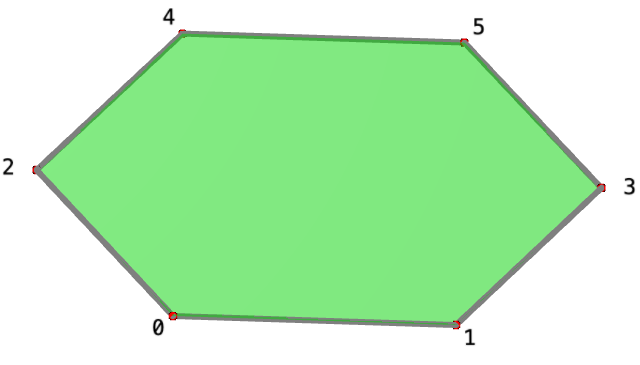}
					\caption{Massless case: 2d reflexive polytope \#9 in
						{\tt sagemath} database.}
					\label{fig:polytopemasslesstriangle}
				\end{subfigure}
				\caption{Two-dimensional reflexive
					polytopes for the triangle graph. }\label{fig:polytopetriangle}
			\end{figure}

			\item The box graph: only the case $n_\cF=2$ leads to reflexive
			polytopes. The graph polynomials are given by
			\begin{align}
				\cU_{\rm box}&=x_1+x_2+x_3+x_4, &\cL_{\rm
					box}&=m_1^2x_1+m_2^2x_2+m_3^2x_3+m_4^2x_4,\cr
				\cV_{\rm
					box}&=\sum_{1\leq i<j\leq 4} (p_{i}+\cdots+p_{j-1})^2 x_ix_j,
				&\cF_{\rm
					box}&= \cU_{\rm box}\cL_{\rm
					box}-  \cV_{\rm
					box}.    
			\end{align}
			The Newton polytopes  $2\Newton(\cF_{\rm
                          box})=4\simplexwithpars{4,1}$ for the massive box graph and
			$2\Newton(\cV_{\rm box})$ for the massless box
                        graph are three-dimensional and reflexive. They are
			shown in \cref{fig:polytopebox} and their relation to $K3$ varieties is discussed in \cref{massive-box,massless-box}.
		\end{itemize}

			\begin{figure}[ht]
				\centering
				\begin{subfigure}{0.45\textwidth}
					\centering
					\includegraphics[width=5cm]{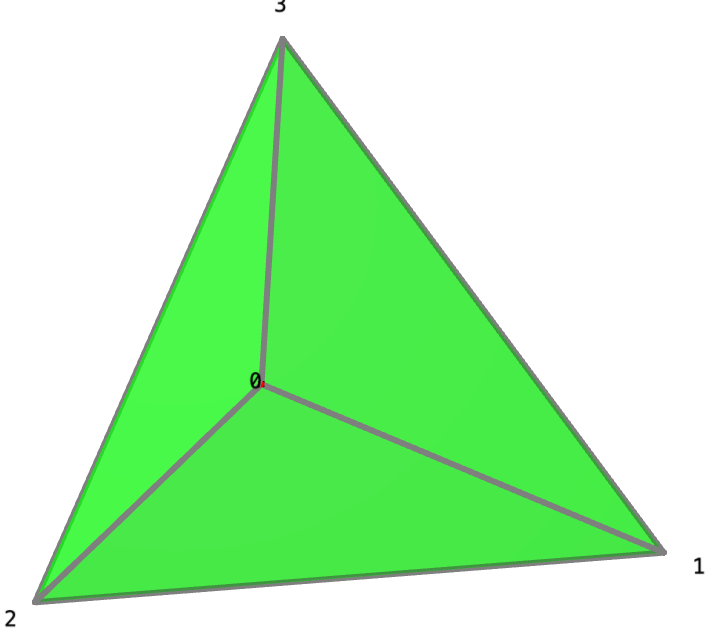}
					\caption{Massive case: 3d reflexive polytope \#4311 in
						{\tt sagemath} database.}
					\label{fig:polytopemassivebox}
				\end{subfigure}\hfill
				\begin{subfigure}{0.45\textwidth}
					\centering
					\includegraphics[width=5cm]{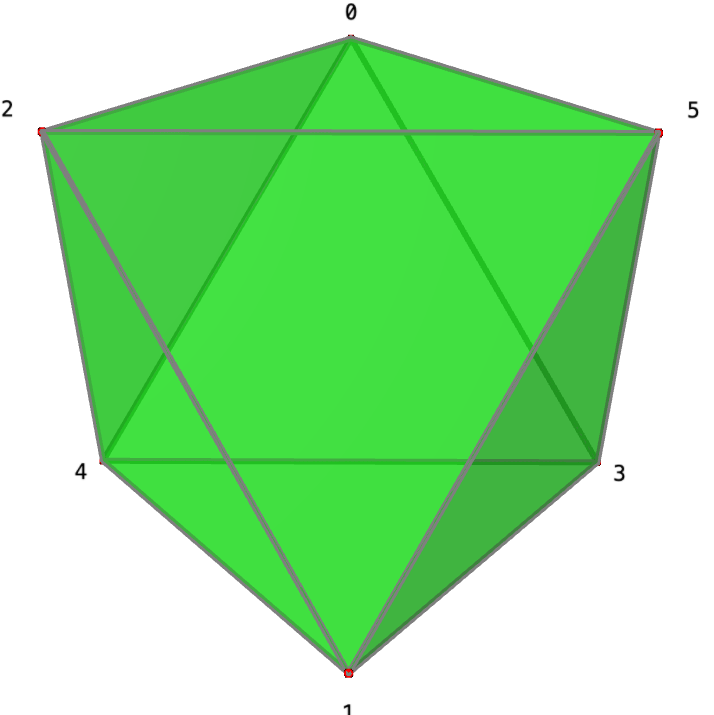}
					\caption{Massless case: 3d reflexive polytope \#3349 in
						{\tt sagemath} database.}
					\label{fig:polytopemasslessbox}
				\end{subfigure}
				\caption{Three-dimensional reflexive
					polytopes for the box graph. }\label{fig:polytopebox}
			\end{figure}

	\subsection{$D=4$, $L=2$: two-loop integrals in four dimensions}%
    \label{sec:fourdimtwoloop}

                The Symanzik graph polynomials for two-loop graphs
                take the generic expression derived
                in~\cite{Doran:2023yzu} and are recalled in
                \cref{app:TwoloopGraph}. Thanks to these generic
                expressions, it is not difficult to search for two-loop
                graph polytopes that are Fano or reflexive using
                {\tt polymake}~\cite{polymake:FPSAC_2009} or {\tt sagemath}~\cite{sagemath}.

		Solving the finiteness  conditions for two-loop graphs
                and searching for polytopes with a single interior
                point, we only find three cases: the kite, shown
                in \cref{fig:kite}, the house graph, shown in
                \cref{fig:house}, and the                tardigrade,
                shown in \cref{fig:tardigrade}.
                Polytopes for these graphs with all
                massive internal lines lead to non-reflexive Fano
                polytopes. The massless internal line cases lead to
                reflexive polytopes. 
		
		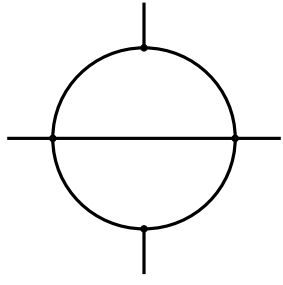
\begin{figure}[h]
			\centering
			\begin{tikzpicture}[scale=0.6]
				\filldraw [color = black, fill=none, very thick] (0,0) circle (2cm);
				\draw [black,very thick] (-2,0) to (2,0);
				\filldraw [black] (2,0) circle (2pt);
				\filldraw [black] (-2,0) circle (2pt);
				\filldraw [black] (0,2) circle (2pt);
				\filldraw [black] (0,-2) circle (2pt);
				\draw [black,very thick] (-2,0) to (-3,0);
				\draw [black,very thick] (2,0) to (3,0);
				\draw [black,very thick] (0,2) to (0,3);
				\draw [black,very thick] (0,-2) to (0,-3);
			\end{tikzpicture}
			\caption{The kite graph, edge label $(1,2,2)$.}\label{fig:kite}
		\end{figure}

        \paragraph{The kite graph.}
		The graph polynomials of the kite graph, shown in in \cref{fig:kite}, are given by
                \cref{app:TwoloopGraph}:
		\begin{align}\label{e:KitePolynomials}
			\cU_{(1,2,2)}&=(y_1+y_2+z_1+z_2)x_1 +(z_1+z_2)(y_1+y_2)\, ,\cr
			\cL_{(1,2,2)}&=m_1^2x_1+m_2^2y_1+m_3^2y_2+m_4^2z_1+m_5^2z_2\,,\cr
			\cV_{(1,2,2)}&=p_1^2 (x_1 y_1 y_2+ x_1 y_2 z_1+y_1 y_2
			z_1 +x_1 y_2 z_2 +y_1 y_2 z_2)\cr
			&+ p_3^2 (x_1 y_1 z_1+x_1 y_2 z_1+x_1
			y_1 z_2+x_1 y_2 z_2)
			\cr
			& + p_2^2 (x_1
			y_1 z_2+x_1 y_2 z_2 + x_1 z_1 z_2 + y_1 z_1 z_2 + y_2 z_1 z_2)\cr
			& - 2 p_1\cdot p_3 (x_1 y_2
			z_1 +x_1
			y_2 z_2)
			+ 2
			p_2\cdot p_3
			(x_1 y_1 z_2 +x_1 y_2
			z_2)
			- 2 p_1\cdot
			p_2 x_1 y_2 z_2\, ,\cr
			\cF_{(1,2,2)}&=\cU_{(1,2,2)}\cL_{(1,2,2)}-\cV_{(1,2,2)}\,.
		\end{align}
		The Newton polytope for the kite graph in $D=4$ dimensions
		with non-vanishing internal masses,
                \begin{equation}
                  \Delta_{(1,2,2)}(2,4;(1,1,1,1,1))=\Newton(\cU_{(1,2,2)})+\Newton(\cF_{(1,2,2)})=2\Newton(\cU_{(1,2,2)})+\simplexwithpars{5,1} \, ,
                \end{equation}
                has a single interior point, so it is Fano, but it is not reflexive. The
		Newton polytope of the kite graph with all massless
                internal lines,
                \begin{equation}
                  \Delta_{(1,2,2),0}(2,4;(1,1,1,1,1))= \Newton(\cU_{(1,2,2)})+\Newton(\cV_{(1,2,2)}) \, ,
                  \end{equation}
                  is a reflexive
		four-dimensional polytope with
         24 vertices, 12 facets, 45 lattice points and lattice volume of 196.
        Using this information and expressing the polytope in normal form, we can locate it in the Kreuzer--Skarke database~\cite{Kreuzer:2002uu}, which also gives the 
     associated Hodge numbers: $h_{11}=8$ and $h_{12}=40$.  
		
		\begin{figure}[h]
			\centering
			\begin{tikzpicture}[scale=0.6]
				\filldraw [color = black, fill=none, very thick] (0,0) circle (2cm);
				\draw [black,very thick] (-2,0) to (2,0);
				\filldraw [black] (2,0) circle (2pt);
				\filldraw [black] (-2,0) circle (2pt);
				\filldraw [black] (0,2) circle (2pt);
				\filldraw [black] (1.414,-1.414) circle (2pt);
				\filldraw [black] (-1.414,-1.414) circle (2pt);
				\draw [black,very thick] (-2,0) to (-3,0);
				\draw [black,very thick] (2,0) to (3,0);
				\draw [black,very thick] (0,2) to (0,3);
				\draw [black,very thick] (1.414,-1.414) to (2.25,-2.25);
				\draw [black,very thick] (-1.414,-1.414) to (-2.25,-2.25);
			\end{tikzpicture}
			\caption{The house graph, edge label $(1,2,3)$. 
            }\label{fig:house}
		\end{figure}
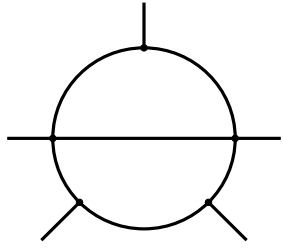
        \paragraph{The house graph.}
		The graph polynomials of the house graph in
                \cref{fig:house} read
		\begin{align}\label{e:123Polynomials}
			\cU_{(1,2,3)}&=x_{1} (y_{1}+y_{2}+z_1+z_2+z_3)+(y_{1}+y_{2})(z_1+z_2+z_3),\cr
			\cL_{(1,2,3)}&=m_1^2x_1+m_2^2y_1+m_3^2y_2+m_4^2z_1+m_5^2z_2+m_6^2z_3,\cr
			\cV_{(1,2,3)}&=-2 p_{1} \cdot p_{2}\, x_{1} y_{2} (z_{2}+z_{3})-2 p_{1} \cdot p_{3} x_{1} y_{2}
			z_{3} +2 p_{1} \cdot p_{4}x_{1} y_{2} 
			(z_{1}+z_{2}+z_{3})\cr
			&+p_{1} \cdot p_{1} y_{2} (x_{1}
			(y_{1}+z_{1}+z_{2}+z_{3})+y_{1}
                   (z_{1}+z_{2}+z_{3}))\cr
                   &+2
			 p_{2} \cdot p_{3} z_{3} (x_{1} (y_{1}+y_{2}+z_{1})+z_{1}
			(y_{1}+y_{2}))\cr
			&-2  p_{2} \cdot p_{4}x_{1} (y_{1}+y_{2})
			(z_{2}+z_{3})+p_{2} \cdot p_{2} (z_{2}+z_{3}) (x_{1}
                   (y_{1}+y_{2}+z_{1})+z_{1} (y_{1}+y_{2}))\cr
                   &-2 
			p_{3} \cdot p_{4} x_{1} z_{3}(y_{1}+y_{2})
			+p_{3} \cdot p_{3} z_{3} (x_{1}
			(y_{1}+y_{2}+z_{1}+z_{2})+(y_{1}+y_{2})
			(z_{1}+z_{2}))\cr
            &+p_{4} \cdot p_{4} x_{1} (y_{1}+y_{2})
			(z_{1}+z_{2}+z_{3}),\cr
			\cF_{(1,2,3)}&=\cU_{(1,2,3)}\cL_{(1,2,3)}-\cV_{(1,2,3)}\,.
		\end{align}
		The Newton polytope for the house graph with
                non-vanishing internal masses is
                \begin{equation}
                  \Delta_{(1,2,3)}(2,4;(1,1,1,1,1,1))=  2 \Newton(\cF_{(1,2,3)}) =2\Newton(\cU_{(1,2,3)})+2\simplexwithpars{6,1} \, .
                \end{equation}
                It has a single interior point, so it is Fano, but again it is not reflexive.
                The Newton polytope of the house graph with vanishing
                internal masses,
                \begin{equation}
                  \Delta_{(1,2,3),0}(2,4;(1,1,1,1,1,1)=  2 \Newton(\cV_{(1,2,3)}) \, ,             
                \end{equation}
                is a reflexive five-dimensional polytope.

		\begin{figure}[h]
			\centering
			\begin{tikzpicture}[scale=0.6]
				\filldraw [color = black, fill=none, very thick] (0,0) circle (2cm);
				\draw [black,very thick] (-2,0) to (2,0);
				\filldraw [black] (2,0) circle (2pt);
				\filldraw [black] (0,2) circle (2pt);
				\filldraw [black] (0,-2) circle (2pt);
				\filldraw [black] (0,0) circle (2pt);
				\filldraw [black] (-2,0) circle (2pt);
				\draw [black,very thick] (-2,0) to (-3,0);
				\draw [black,very thick] (2,0) to (3,0);
				\draw [black,very thick] (0,2) to (0,3);
				\draw [black,very thick] (0,-2) to (0,-3);
				\draw [black,very thick] (0,0) to (0,-1);
			\end{tikzpicture}
			\caption{The tardigrade graph, edge label $(2,2,2)$. }\label{fig:tardigrade}
		\end{figure}
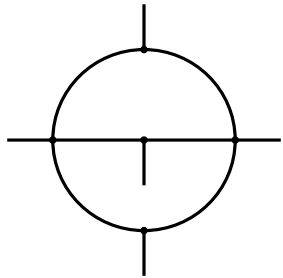
        \paragraph{The tardigrade graph.}
		The tardigrade graph has the graph polynomials
		\begin{align}\label{e:tardigradePolynomials}
			\cU_{(2,2,2)}&=(x_{1}+x_2) (y_{1}+y_{2}+z_1+z_2)+(y_{1}+y_{2})(z_1+z_2),\cr
			\cL_{(2,2,2)}&=m_1^2x_1+m_2^2x_2+m_3^2y_1+m_4^2y_2+m_5^2z_1+m_6^2z_2,\cr
			\cV_{(2,2,2)}&=-2 p_{1} \cdot p_{2} x_{2} z_{2}  (y_{1}+y_{2})+2p_{1} \cdot p_{3} x_{2} y_{2}
			 (z_{1}+z_{2})-2  p_{1} \cdot p_{4}x_{2}
			(y_{1}+y_{2}) (z_{1}+z_{2})\cr
			&+p_{1} \cdot p_{1} x_{2} (x_{1}
			(y_{1}+y_{2}+z_{1}+z_{2})+(y_{1}+y_{2})
			(z_{1}+z_{2}))-2 p_{2} \cdot p_{3} y_{2} z_{2}
                   (x_{1}+x_{2})\cr
                  	&+p_{2} \cdot p_{2} z_{2} (x_{1}
			(y_{1}+y_{2}+z_{1})+x_{2} (y_{1}+y_{2}+z_{1})+z_{1}
                     (y_{1}+y_{2}))\cr                
                       &+p_{3} \cdot p_{3} y_{2} (x_{1}
			(y_{1}+z_{1}+z_{2})+x_{2} (y_{1}+z_{1}+z_{2})+y_{1}
                         (z_{1}+z_{2}))\cr                         
                       &+p_{4} \cdot p_{4} (x_{1}+x_{2}) (y_{1}+y_{2})
                         (z_{1}+z_{2})+2  p_{2} \cdot p_{4} (x_{1}+x_{2})
                            (y_{1}+y_{2})z_{2}\cr
                          &
                            -2  p_{3} \cdot p_{4} (x_{1}+x_{2})y_{2}
			(z_{1}+z_{2}),                       \cr
			\cF_{(2,2,2)}&=\cU_{(2,2,2)}\cL_{(2,2,2)}-\cV_{(2,2,2)}\,.
		\end{align}
		
		The Newton polytope for the tardigrade with
                non-vanishing internal masses is 
                \begin{equation}
                  \Delta_{(2,2,2)}(2,4;(1,1,1,1,1,1))=2\Newton(\cF_{(2,2,2)})=2\Newton(\cU_{(2,2,2)})+2\simplexwithpars{6,1} \, .                  
                \end{equation}
                It has a single interior point, so it is Fano, but it is not reflexive.
                The Newton polytope for the tardigrade with
                all vanishing internal masses is
                \begin{equation}
                  \Delta_{(2,2,2)}(2,4;(1,1,1,1,1,1))=2\Newton(\cV_{(2,2,2)}) \, ,                                 \end{equation}
and it is a five-dimensional reflexive polytope.

	\section{Search by the number of edges}%
    \label{sec:edgescan}
			\begin{figure}[h]
			\centering
                                	\begin{tikzpicture}[scale=0.6]
				\filldraw [black] (0,0) circle (2pt);
				\filldraw [black] (1,1) circle (2pt);
				\filldraw [black] (1,-1) circle (2pt);
				\filldraw [black] (-1,1) circle (2pt);
                                \filldraw [black] (-1,-1) circle
                                (2pt);
                                \filldraw [black] (-2,-2) circle
                                (2pt);
                                \draw [black,very thick] (-2,-2) to
                                (2,2);
                                  \draw [black,very thick] (-2,2) to
                                  (2,-2);
                                    \draw [black,very thick] (-1,-1)
                                    to (-1,1);
                                      \draw [black,very thick] (1,1)
                                      to (1,-1);
                                        \draw [black,very thick] (0,-2) to (0,0);
			\end{tikzpicture}
				\caption{A factorizable graph.}
				\label{fig:doubletriangle}
		\end{figure}
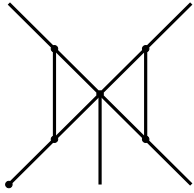
		In this \namecref{sec:edgescan}, we will present our results for an exhaustive direct search of Fano and reflexive polytopes associated to Feynman integrals with up to ten edges, $\nprops \le 10$, assuming generic external kinematics and generic propagator masses.
		
		In generic kinematics the search is facilitated by the fact that the polytope for the second Symanzik polynomial is the Minkowski
		sum of the polytope of the first Symanzik polynomial and the mass hyperplane
		\eqref{e:FMinkowski}.
		We restrict our search to 1-particle irreducible (1PI) non-factorizable graphs; an example of an excluded graph is shown in \cref{fig:doubletriangle}.
        This leads to an upper bound on the number of loops, since $V \ge 2$ implies $L=1+\nprops-V \le \nprops-1$. In addition, as external kinematics is generic, we exclude cases with multiple external momenta attached to the same vertex.
        We generate all Feynman graphs that satisfy these requirements using $\mathtt{QGRAF}$~\cite{Nogueira:1991ex} with the following configuration options:
        \begin{verbatim}
            options = onepi, cycli, topol;
            true = iprop[phi, 0, 10];
        \end{verbatim}

\subsection{Upper bound on dimension}\label{sec:upperdim}
We have just seen that fixing the number of edges~$\nprops$ puts an upper bound on the number of loops~$L$.
We will show now that the remaining parameter in our search---spacetime dimension~$D$---is also bounded.
This means that for a fixed number of edges~$\nprops$, we can perform an exhaustive search for Fano and reflexive polytopes, as we only need to consider a finite number of cases.

To establish an upper bound on $D$, we will make use of the following
\begin{proposition}[cf.~{\cite[Exercise 4.10]{Beck:2015}}]
	Let \(P\) be a \(d\)-dimensional polytope, \(d \geq 1\). Then the number of interior lattice points of its dilations satisfies
	\begin{equation}
		\begin{aligned}
			\nint((d+1)P) & \geq 1, \\
			\nint((d+2)P) & > 1.
		\end{aligned}
	\end{equation}
\end{proposition}
\begin{proof}
	This is a straightforward application of Ehrhart theory (see, e.g.,~\cite[Chapters 3 and 4]{Beck:2015}). The Ehrhart polynomial of \(P\) can be written in terms of the so-called \(h^*\)-vector as
	\begin{equation}
		\ehr_P(t) = \sum_{i=0}^d h_i^* \binom{t+d-i}{d},
		\label{eq:ehr_hstar}
	\end{equation}
	where \(h_0^* = 1\) and \(h_i^* \geq 0\). Combining \cref{eq:ehr_hstar} with the reciprocity formula~\eqref{eq:n_int_ehr}, we find
	\begin{equation}
		\nint(kP) = \sum_{i=0}^d h_i^* \binom{k+i-1}{d} \geq \binom{k-1}{d}. 
	\end{equation}
	Taking \(k=d+1\) and \(k=d+2\), we obtain the desired bound.
\end{proof}

The Newton polytope~\(\Newton_\Gamma(L, D; \undernotation{\nu})\) lives in \((N-1)\) dimensions, so the above proposition implies that it cannot be Fano if it contains an \((N+1)\)-fold dilation of some polytope.
This means that in the representation of \(\Newton_\Gamma(L, D; \undernotation{\nu})\) as a weighted Minkowski sum we cannot have coefficients exceeding \(N\).

In the case of non-vanishing internal masses, we can write 
\begin{equation}
\Newton_\Gamma(L, D; \undernotation{\nu}) = D/2 \, \Newton (\cU_\Gamma) + n_\cF \Newton(\cL_\Gamma)\,,
\end{equation}
 so that $D/2, n_\cF \leq N $.
With more general kinematics, we have \(n_\cU, n_\cF \leq N\), and in particular, \(D/2 = n_\cU + n_\cF \leq 2N\).

\subsection{Symmetries}
Graphs may share the same Symanzik polynomials and thus correspond to the same Feynman integral. The determination of graphs that share the same Symanzik polynomials is a basic problem in Feynman integral calculations. Given a set of graphs, the isomorphism problem can be solved by choosing a canonical labeling and comparing the relabeled graphs. The isomorphism between two Newton polytopes can be checked similarly by constructing a graph known as a 1-skeleton, choosing a canonical ordering of the graph vertices, and comparing the relabeled polytopes~\cite{Bremner_2014}. Another alternative is to construct a normal form for the vertices of a polytope~\cite{grinis2013normalformsconvexlattice}.  For our purposes, it is not necessary to construct all polytope symmetries of a given integral~\cite{delaCruz:2024ssb}, but only to determine a representative of a given set of integrals sharing the same polytope.  

A strategy based on normal ordering of vertices is used ordinarily in IBP-reduction packages. It is known as \emph{Pak algorithm} in this context~\cite{Pak:2011xt}.  We can compare graphs by ordering variables of the graph polynomial pairs 
$P_\Gamma=(\cU_\Gamma, \cF_\Gamma)$. Two graphs $ \Gamma_1$, $\Gamma_2$ are equivalent if their polynomial pairs are equal under their normal orderings
\begin{equation}
	\text{Normal}(P_{\Gamma_1})
	=\text{Normal}(P_{\Gamma_2}) \, . 
\end{equation}

We use Pak algorithm to identify polynomials that have the same monomials (but not necessarily having the same coefficients) as they will correspond to the same polytope. For instance, the graphs in \cref{example-symmetries} share the same polytope, because they have the same vertices, despite the fact that the graphs do not have the same number of external legs.
As a result, we can group graphs into equivalence classes and choose a representative of each class.
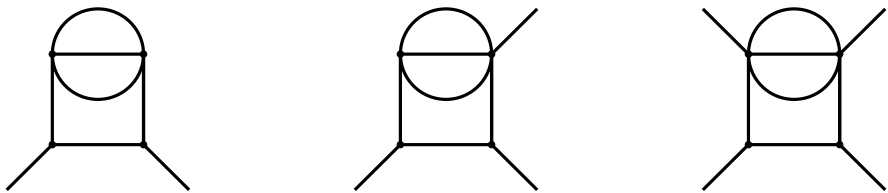
\begin{figure}[htb]
	\centering
	\begin{tikzpicture}[scale=0.6]
		\filldraw [color = black, fill=none,
		very thick] (0,1) circle (1cm);
		\draw [black,very thick] (0,0) +(-1cm,-1cm) rectangle +(1cm,1cm);
		\filldraw [black] (1,-1) circle (2pt);
        \filldraw [black] (1,1) circle (2pt);
		\filldraw [black] (-1,1) circle (2pt);
		\filldraw [black] (-1,-1) circle(2pt);
		\draw [black,very thick] (-2,-2) to
		(-1,-1);                                                             
		\draw [black,very thick] (1,1)
		to (1,-1);
		\draw [black,very thick] (2,-2) to (1,-1);
	\end{tikzpicture}\hspace{2cm}
	\begin{tikzpicture}[scale=0.6]
		\filldraw [color = black, fill=none,
		very thick] (0,1) circle (1cm);
		\draw [black,very thick] (0,0) +(-1cm,-1cm) rectangle +(1cm,1cm);
		\filldraw [black] (1,1) circle (2pt);
		\filldraw [black] (1,-1) circle (2pt);
		\filldraw [black] (-1,1) circle (2pt);
		\filldraw [black] (-1,-1) circle
		(2pt);
		\draw [black,very thick] (-2,-2) to
		(-1,-1);                               
		\draw [black,very thick] (2,2)
		to (1,1);
		\draw [black,very thick] (1,1)
		to (1,-1);
		\draw [black,very thick] (2,-2) to (1,-1);
	\end{tikzpicture}\hspace{2cm}
	\begin{tikzpicture}[scale=0.6]
		\filldraw [color = black, fill=none,
		very thick] (0,1) circle (1cm);
		\draw [black,very thick] (0,0) +(-1cm,-1cm) rectangle +(1cm,1cm);
		\filldraw [black] (1,1) circle (2pt);
		\filldraw [black] (1,-1) circle (2pt);
		\filldraw [black] (-1,1) circle (2pt);
		\filldraw [black] (-1,-1) circle	(2pt);
		\draw [black,very thick] (-2,-2) to
		(-1,-1);
		\draw [black,very thick] (-2,2) to
		(-1,1);
		\draw [black,very thick] (2,2)
		to (1,1);
		\draw [black,very thick] (1,1)
		to (1,-1);
		\draw [black,very thick] (2,-2) to (1,-1);
	\end{tikzpicture}
	\caption{Three graphs that share the same polytope in generic kinematics}
	\label{example-symmetries}
\end{figure}

 \subsection{Results}
 
The search proceeds as follows. In a first step, we generate the set of all graphs $\Gamma^\nprops_{\mathrm{all}}$, for $\nprops=1, \dots, 10$ with \texttt{QGRAF}. 
For $N=10$, this gives 16744 graphs. 
We then compute $\cU_\Gamma$ and $\cF_\Gamma$ and group our graphs based on the normal ordering of the polynomial pair $\text{Normal}(P_\Gamma)$, which we compute using the \texttt{FeynCalc} command $\mathsf{FCLoopPakOrder}$~\cite{Shtabovenko:2023idz}.
For $N=10$, this leaves us with only 565 graph classes.

In a second step, we take a single representative graph of a given class and scan the list of allowed pairs $(D/2, n_\cF)$ using the upper bound of \cref{sec:upperdim}.
For each pair $(D/2, n_\cF)$, we construct the corresponding polytope~$P$ and compute the number of interior points~$\nint(P)$.
If $\nint(P) > 1$, we discard this pair as well as all pairs with higher values of $D/2$ or $n_\cF$, because dilating the polytope further can only increase the number of interior points.
If $\nint(P) = 1$, then the polytope is Fano, and we additionally check if it is reflexive.
This procedure results in a list of pairs $(D/2, n_\cF)$, or equivalently $(n_\cU, n_\cF)$, which give Fano or reflexive polytopes.    
 
We use \texttt{polymake} in order to compute the number of interior points and check reflexivity of a given polytope.
This can be easily done, for example, by evaluating \texttt{N\_INTERIOR\_LATTICE\_POINTS} and \texttt{REFLEXIVE} properties of the polytope, respectively.
To perform these computations, \texttt{polymake} uses the Parma Polyhedra Library~(PPL)~\cite{BagnaraHZ08SCP}.

Our results are summarized in \cref{tableResultsFano} (see also \cref{app:fanotable}). We plot the 3d
distribution of Fano and reflexive polytopes in $(\nprops, L, D)$ variables in
\cref{fulldistribution}.
We also show 2d projections of this distribution, separately for reflexive polytopes in \cref{fig:plotReflexive} and for Fano polytopes in \cref{fig:plotFano}.
The distribution of reflexive
polytopes is bounded from below by graphs with $\nprops$
edges at one loop (see \cref{sec:oneloop}), and bounded from
above by the  multiloop sunset graphs with $\nprops$ edges and
$\nprops-1$ loops (see \cref{sec:sunset}). From the analyses of our results we observe:
\begin{itemize}
\item In  every even dimension
there is a reflexive polytope from a one-loop graph as we discussed
in \cref{sec:oneloop} (see \cref{fig:PlotReflexiveLvsD}).
\item Another remarkable relation is the existence of reflexive
  polytopes in dimension $D=2L+2$. Some associated Feynman integrals are evaluated  in \cref{sec:periods}.
  \item We observe that the $\nprops$-gon and the multiloop sunset graphs admit a reflexive polytope with a non-vanishing exponent~$n_\cF$. 
  \item We notice that  Fano non-reflexive cases arise from graphs with five edges on. This is the case of the kite graph of \cref{sec:fourdimtwoloop}, which has a Fano or reflexive polytope depending on the powers of the Symanzik polynomials. 
\end{itemize}

As we will see in the next \namecref{sec:periods}, some reflexive graphs are simpler to evaluate. Some of their divergent cousins have already been evaluated  in the renormalization of $\varphi^4$ theory~\cite{Kompaniets:2017yct} (see the first three graphs in \cref{representatives8edges}). On the other hand, Fano cases quickly lead to topologies that are currently under extensive investigation, see e.g.~\cite{Abreu:2024fei}. These examples typically involve an integrand where both Symanzik polynomials appear. We present all of these cases  in our search in \cref{Fanoupto9,Fanoupto10} in \cref{app:fanotable}. We do not attempt to evaluate those integrals.

	\begin{table}[tb]
			\centering
			\setlength{\tabcolsep}{1em}
			\begin{tabular}{ccccc}
				$\nprops$ & \makecell[c]{graph \\ topologies} & \makecell[c]{Fano \\ classes} & \makecell[c]{reflexive \\ classes} & \makecell[c]{non-reflexive \\ Fano classes} \\
				\midrule
				2 & 1 & 1 & 1 & 0 \\
				\addlinespace
				3 & 2 & 2 & 2 & 0 \\ 
				\addlinespace
				4 & 3 & 3 & 3 & 0 \\
				\addlinespace
				5 & 6 & 4 & 4 & 1 \\ 
				\addlinespace
				6 & 13 & 8 & 7 & 4 \\ 
				\addlinespace
				7 & 28 & 11 & 6 & 6 \\
				\addlinespace 
				8 & 70 & 23 & 11 & 16 \\ 
				\addlinespace 
				9 & 193 & 36 & 14 & 24 \\ 
				\addlinespace 
				10 & 565 & 104 & 26 & 88 \\ 
			\end{tabular}
			\caption{Non-factorizable graph topologies (with generic masses and kinematics) leading to polytopes with a single interior point that we call Fano classes. We give the number of graph topologies leading to  reflexive polytopes and the ones that are Fano but not reflexive.}
			\label{tableResultsFano}
		\end{table} 
		\begin{figure}[tb]
			\centering
			\begin{subfigure}{0.45\textwidth}
				\centering
				\includegraphics[scale=0.6]{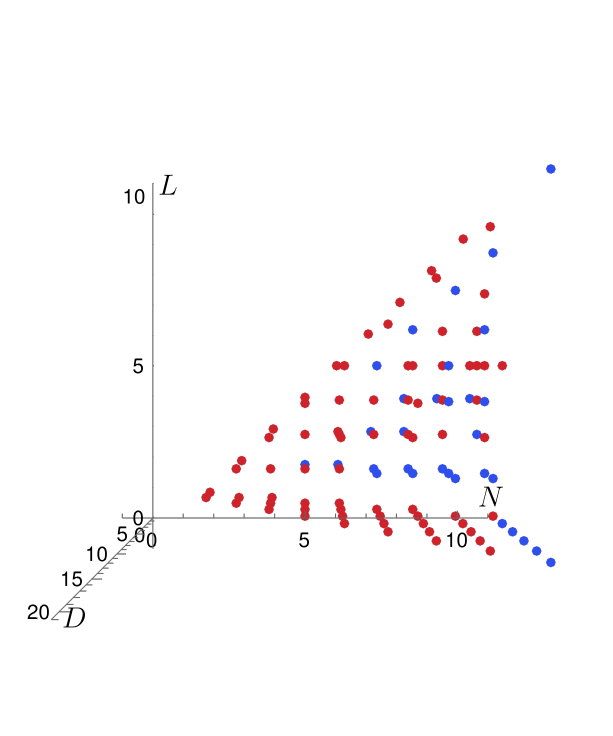}
			\end{subfigure}
			\begin{subfigure}{0.45\textwidth}
				\centering
				\includegraphics[scale=0.6]{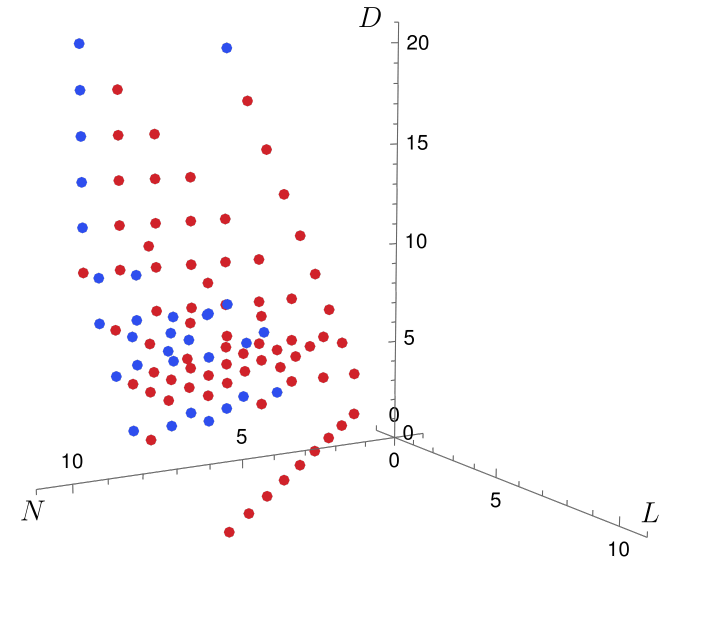}
						\end{subfigure}
			\caption{Distribution of Fano and reflexive polytopes  $(\nprops,L,D)$. Blue  points correspond to graphs that are only Fano. Red points are both Fano and reflexive. Each point represents a value for $(N,L,D)$ where at least one Fano polytope exists.}
			\label{fulldistribution}
		\end{figure}
        
\pgfplotsset{width=4.5cm, height=4.5cm}
\pgfplotsset{every axis/.append style={
    grid=both,
    font=\footnotesize}}
\pgfplotsset{scat/.style={only marks, mark size=1pt}}
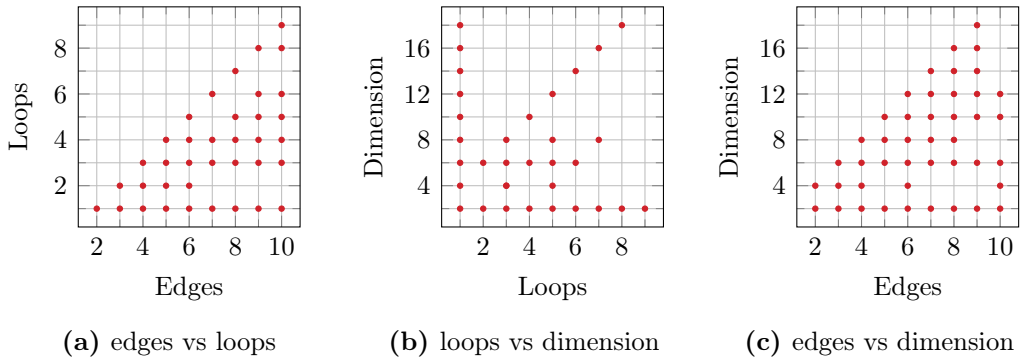
\begin{figure}[tb]
\centering
\definecolor{plot-red}{HTML}{D0222A}
\begin{subfigure}{0.3\textwidth}
\begin{tikzpicture}
    \begin{axis}[
        xlabel=Edges,
        ylabel=Loops,
        xtick distance=2,
        minor x tick num=1,
        ytick distance=2,
        minor y tick num=1,
    ]
    \addplot[scat, plot-red] coordinates {
    (2,1) (3, 1) (3, 2) (4, 1) (4, 2) (4, 3) (5, 1) (5, 2) (5, 3) (5, 4) (6, 1) (6, 2) (6, 3) (6, 4) (6, 5) (7, 1) (7, 3) (7, 4) (7, 6) (8, 1) (8, 3) (8, 4) (8, 5) (8, 7) (9, 1) (9, 3) (9, 4) (9, 5) (9, 6) (9, 8) (10, 1) (10,3) (10, 4) (10, 5) (10,6) (10,8) (10, 9)
    };
    \end{axis}
\end{tikzpicture}
\caption{edges vs loops}
\label{fig:PlotReflexiveEvsL}
\end{subfigure}
\begin{subfigure}{0.3\textwidth}
\begin{tikzpicture}
    \begin{axis}[
        xlabel=Loops,
        ylabel=Dimension,
        xtick distance=2,
        minor x tick num=1,
        ytick distance=4,
        minor y tick num=1,
    ]
    \addplot[scat, plot-red] coordinates {
    (1,2) (1, 4) (1, 6) (1,8) (1,10) (1,12) (1,14) (1,16) (1,18) (2, 2)  (2, 6)  (3, 2) (3,4) (3, 6) (3, 8) (4, 2) (4, 6) (4, 10) (3, 4)  (5, 2) (5, 4) (5, 6) (5, 8) (5, 12) (6, 2)  (6, 6) (6, 14)  (7, 2) (7,8) (7, 16)  (8, 2) (8, 18)  (9, 2)
    };
    \end{axis}
\end{tikzpicture}
\caption{loops vs dimension}
\label{fig:PlotReflexiveLvsD}
\end{subfigure}
\begin{subfigure}{0.3\textwidth}
\begin{tikzpicture}
    \begin{axis}[
        xlabel=Edges,
        ylabel=Dimension,
        xtick distance=2,
        minor x tick num=1,
        ytick distance=4,
        minor y tick num=1,
    ]
    \addplot[scat, plot-red] coordinates {
    (2,2) (2,4) (3, 2) (3, 4) (3, 6) (4, 2) (4, 4) (4, 6) (4, 8) (5, 2) (5, 6) (5, 8) (5, 10)  (6, 2)  (6, 4) (6, 6) (6, 8) (6, 10) (6, 12) (7, 2) (7, 6) (7, 8) (7, 10) (7, 12) (7, 14)  (8, 2) (8, 6) (8, 8) (8, 10) (8, 12) (8, 14) (8, 16) (9, 2) (9, 6)  (9, 10) (9, 12) (9, 14) (9, 16) (9, 18) (10, 2)  (10, 4)  (10, 6) (10, 10) (10, 12) 
    };
    \end{axis}
\end{tikzpicture}
\caption{edges vs dimension}
\label{fig:PlotReflexiveEvsD}
\end{subfigure}  
\caption{Distribution of reflexive graph polytopes from Feynman
integrals  with massive internal lines up to 10 edges.
Each point represents a value for $(N,L)$, $(L,D)$, $(N,D)$ where at least one reflexive polytope exists.}\label{fig:plotReflexive}
\end{figure}
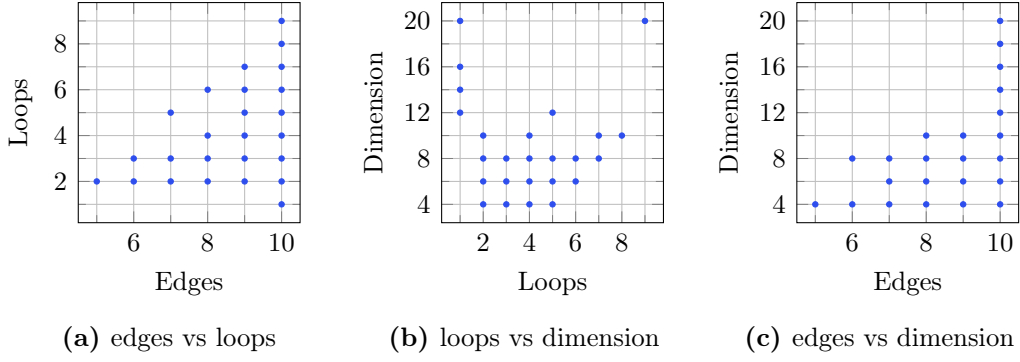
\begin{figure}[tb]
\centering
\definecolor{plot-blue}{HTML}{2E4EED}
\begin{subfigure}{0.3\textwidth}
\begin{tikzpicture}
    \begin{axis}[
        xlabel=Edges,
        ylabel=Loops,
        xtick distance=2,
        minor x tick num=1,
        ytick distance=2,
        minor y tick num=1,
    ]
    \addplot[scat, plot-blue] coordinates {
    (5, 2)  (6, 2) (6,3) (7, 2) (7, 3) (7, 5) (8, 2) (8, 3) (8, 4) (8, 6) (9, 2) (9, 3) (9, 4) (9, 5) (9, 6) (9, 7) (10,1) (10, 2) (10, 3) (10, 4) (10, 5) (10,6) (10,7) (10,8) (10,9)
    };
    \end{axis}
\end{tikzpicture}
\caption{edges vs loops}
\label{fig:PlotFanoEvsL}
\end{subfigure}
\begin{subfigure}{0.3\textwidth}
\begin{tikzpicture}
    \begin{axis}[
        xlabel=Loops,
        ylabel=Dimension,
        xtick distance=2,
        minor x tick num=1,
        ytick distance=4,
        minor y tick num=1,
    ]
    \addplot[scat, plot-blue] coordinates {
    (1,12) (1,14) (1,16) (1,20) (2, 4) (2, 6) (2, 8) (2, 10) (3, 4) (3, 6) (3,8) (4, 4) (4, 6) (4, 8) (4,10)  (5, 4) (5,6) (5, 8) (5,12) (6, 6) (6, 8) (7,8) (7, 10) (8,10) (9,20)
    };
    \end{axis}
\end{tikzpicture}
\caption{loops vs dimension}
\label{fig:PlotFanoLvsD}
\end{subfigure}
\begin{subfigure}{0.3\textwidth}
\begin{tikzpicture}
    \begin{axis}[
        xlabel=Edges,
        ylabel=Dimension,
        xtick distance=2,
        minor x tick num=1,
        ytick distance=4,
        minor y tick num=1,
    ]
    \addplot[scat, plot-blue] coordinates {
      (5, 4) (6, 4) (6, 8)  (7, 4) (7, 6) (7, 8)  (8, 4) (8, 6) (8, 8) (8, 10)  (9, 4) (9, 6) (9, 8) (9, 10) (10, 4) (10, 6)  (10, 8)  (10, 10) (10, 12) (10, 14) (10, 16) (10, 18) (10, 20)
    };
    \end{axis}
\end{tikzpicture}
\caption{edges vs dimension}
\label{fig:PlotFanoEvsD}
\end{subfigure}  
\caption{Distribution of non-reflexive Fano graph polytopes from  Feynman
integrals with massive internal lines up to 10 edges.
Each point represents a value for $(N,L)$, $(L,D)$, $(N,D)$ where at least one non-reflexive Fano polytope exists.}\label{fig:plotFano}
\end{figure}
	
\newcommand{\grapheightedgesOne}{  
\begin{tikzpicture}[scale=0.7]
\draw[very thick] (0,1) ellipse (0.5 and 0.2);
\draw [black,very thick] (0,0) to   (0,-0.3);
\draw [black,very thick] (-0.5,1) to   (-0.8,1);
\draw [black,very thick] (1,0) to   (1,-0.2);
\draw [black,very thick] (0,0) to   (0.5,1);
\draw [black,very thick] (0,0) to   (-0.5,1);
\draw [black,very thick] (0,2) to   (0.5,1);
\draw [black,very thick] (1,0) to   (0.5,1);
\draw [black,very thick] (1,0) to   (1,1);
\draw [black,very thick] (0.5,1) to   (2,1);
                                \draw [black,very thick] (0,0) to
                                (1.2,0);
                                 \filldraw [black] (1,1) circle (2pt);
                                \filldraw [black] (1,0) circle (2pt);
                                \filldraw [black] (0,0) circle (2pt);
                                \filldraw [black] (-0.5,1) circle (2pt);
                                 \filldraw [black] (0.5,1) circle (2pt);
			\end{tikzpicture}
            }		 

\newcommand{\grapheightedgesTwo}{
   \begin{tikzpicture}[scale=0.5]
                                  \draw[very thick] (0.5,0) ellipse (0.5 and 0.2);
                             \draw [black,very thick] (0,0)    to (0, 0.5);
                              \draw [black,very thick] (0,-1)    to (0, -2);
 \draw[very thick,rotate=45] (-0.7,0) ellipse (0.20 and 0.7);
                                \draw [black,very thick] (0,-1)
                                to (1,0);
                                \draw [black,very thick] (1,0)
                                to (0,1);
                                \draw [black,very thick] (-1,0)
                                to (0,1);
                                \filldraw [black] (0,0) circle (2pt);
				\filldraw [black] (1,0) circle (2pt);
                                \filldraw [black] (-1,0) circle
                                (2pt);
				\filldraw [black] (0,1) circle (2pt);
                                \filldraw [black] (0,-1) circle (2pt);
                                  \draw [black,very thick] (-2,0) to
                                  (-1,0);
                                    \draw [black,very thick] (2,0)
                                    to (1,0);                                      
                                         \draw [black,very thick] (0,2) to
                                         (0,1);
                                          \draw [black,very thick] (0,0) to
                                (-1,0);
			\end{tikzpicture}
}
\newcommand{\grapheightedgesThree}{  \begin{tikzpicture}[scale=0.7]
             \draw[very thick] (0,1) ellipse (0.5 and 0.2);
          \draw [black,very thick] (0,0) to   (0,-0.5);
                    \draw [black,very thick] (0.5,1) to   (0.5,1.5);
                \draw [black,very thick] (-0.5,1) to   (-0.5,1.5);
          \draw [black,very thick] (0,0) to   (0.5,1);
          \draw [black,very thick] (0,0) to   (-0.5,1);
          \draw [black,very thick] (-1,0) to   (-0.5,1);
          \draw [black,very thick] (1,0) to   (0.5,1);
                                \draw [black,very thick] (-1.5,0) to
                                (1.5,0);
                                \filldraw [black] (-1,0) circle (2pt);
                                \filldraw [black] (1,0) circle (2pt);
                                \filldraw [black] (0,0) circle (2pt);
                                \filldraw [black] (-0.5,1) circle (2pt);
                                 \filldraw [black] (0.5,1) circle (2pt);
			\end{tikzpicture}}

\newcommand{\grapheightedgesFour}{
\begin{tikzpicture}[scale=0.6]
                                  \draw[very thick] (0,0.5) ellipse (0.20 and 0.5);
 \draw[very thick] (0,-0.5) ellipse (0.20 and 0.5);
  \draw [black,very thick] (0,0)                                to (0.5,0);
                                \draw [black,very thick] (0,-1)
                                to (-1,0);
                                \draw [black,very thick] (0,-1)
                                to (1,0);
                                \draw [black,very thick] (1,0)
                                to (0,1);
                                \draw [black,very thick] (-1,0)
                                to (0,1);
                                \filldraw [black] (0,0) circle (2pt);
				\filldraw [black] (1,0) circle (2pt);
                                \filldraw [black] (-1,0) circle
                                (2pt);
				\filldraw [black] (0,1) circle (2pt);
                                \filldraw [black] (0,-1) circle (2pt);
                                 \draw [black,very thick] (-2,0) to
                                  (-1,0);
                                    \draw [black,very thick] (2,0)
                                    to (1,0);
                                        \draw [black,very thick] (0,-2) to
                                        (0,-1);
                                         \draw [black,very thick] (0,2) to
                                (0,1);
			\end{tikzpicture}
}
\newcommand{\grapheightedgesFive}{  \begin{tikzpicture}[scale=0.5]
                                  \filldraw [color = black, fill=none,
                                  very thick] (0,1) circle (1cm);
 \draw[very thick] (0,1) ellipse (1 and 0.5);
                                \draw [black,very thick] (-1,1)
                                to (-1,-1);
                                \draw [black,very thick] (-1,-1)
                                to (1,-1);
                                \draw [black,very thick] (1,-1)
                                      to (1,1);
				\filldraw [black] (1,1) circle (2pt);
				\filldraw [black] (1,-1) circle (2pt);
				\filldraw [black] (-1,1) circle (2pt);
                                \filldraw [black] (-1,-1) circle
                                (2pt);
                                \draw [black,very thick] (-2,-2) to
                                (-1,-1);
                                  \draw [black,very thick] (-2,2) to
                                  (-1,1);
                                    \draw [black,very thick] (2,2)
                                    to (1,1);
                                      \draw [black,very thick] (1,1)
                                      to (1,-1);
                                        \draw [black,very thick]
                                        (2,-2) to (1,-1);
                                        \draw [black,very thick] (0,-2) to
                                (0,-1);
			\end{tikzpicture}
            }
\begin{figure}[htb]
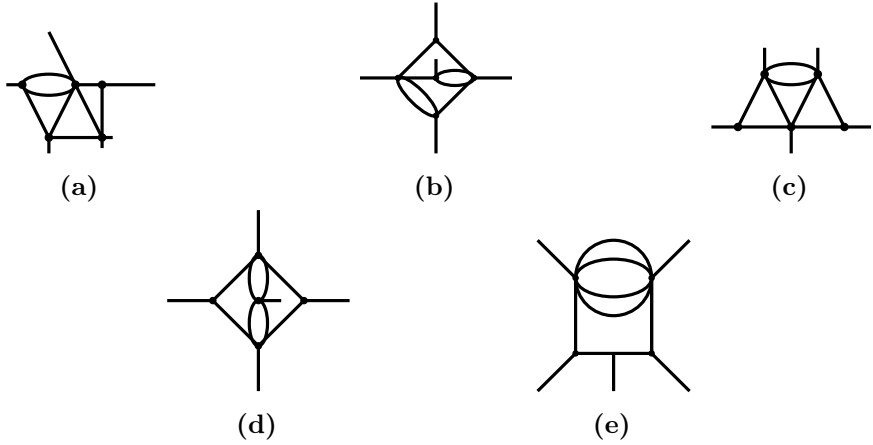

	\centering
\begin{subfigure}{0.3\textwidth}
\centering
             \grapheightedgesOne
              \caption{}
\end{subfigure}
\begin{subfigure}{0.3\textwidth}
\centering
             \grapheightedgesTwo
             \caption{}
\end{subfigure}
\begin{subfigure}{0.3\textwidth}
\centering
             \grapheightedgesThree
            \caption{}
\end{subfigure}
\begin{subfigure}{0.3\textwidth}
\centering
             \grapheightedgesFour
             \caption{}
\end{subfigure}
\begin{subfigure}{0.3\textwidth}
\centering
             \grapheightedgesFive
             \caption{}
             \label{dancing-boy}
\end{subfigure}
	\caption{Representative graphs with eight propagators that have a reflexive polytope. }
	\label{representatives8edges}
\end{figure}
	
\section{Period integrals}\label{sec:periods}
After classifying the reflexive or Fano polytopes we evaluate some of these integrals. Finite integrals and quasi-finite integrals are simpler to evaluate since by definition they do not allow divergences or subdivergences. They can be  evaluated analytically  with {\tt HyperInt}~\cite{Panzer:2014caa}.

The search of reflexive polytopes for graphs up to ten edges
has the case of integrals of the type $(n_\cU, n_\cF) = (n,0)$:
\begin{equation}
 I_\Gamma(L,2n;\undernotation{\nu})= \int_{\mathbb R^{\nprops-1}_+}
\left. {\prod_{i=1}^{\nprops} x_i^{\nu_i}\over \cU_\Gamma^n }\right|_{x_\nprops=1} 
\prod_{i=1}^{\nprops-1}{\dd x_i\over x_i}\, ,
\end{equation}
which are called  period integrals~\cite{Bloch:2005bh}.
The unique point in the interior of the associated polytope $n\Newton(\cU_\Gamma)$ is  $\undernotation{\nu}=(\nu_1,\dots,\nu_\nprops)$, with  $\nu_1+\cdots+\nu_\nprops=Ln$. Despite the fact that our list of graphs is limited by those having at most  ten edges, we observe that some patterns emerge.  In particular, we can discuss four families of graphs which lead to reflexive polytopes and evaluate their period integrals, namely, the one-loop $\nprops$-gon and the multiloop sunset in \cref{sec:Family1}, the family of graphs with external legs attached to one edge  of the multiloop sunset in \cref{sec:Family1n} and the family of graphs with external legs attached to two edges of the multiloop sunset in \cref{sec:Family1nn}.  They all evaluate to rational numbers obtained as product of factorials. In \cref{sec:Intzeta} we evaluate the integrals for the two reflexive cases that give zeta-values.

\subsection{One-loop $\nprops$-gon and multiloop sunset}\label{sec:Family1}
Let us begin by analyzing the reflexive polytopes associated with the $\nprops$-gon and the $(\nprops-1)$-loop sunset.   The former case has the polytope in $D=2\nprops$ dimensions 
\begin{equation}
    \Delta_{\nprops-\mathrm{gon}}(1,2\nprops;(1,\cdots,1))=\nprops \simplexwithpars{\nprops,1}\, , 
\end{equation}
while the latter, also in  $D=2\nprops$ dimensions, reads
\begin{equation}
    \Delta_{\circleddash}(\nprops-1,2\nprops;(\nprops-1,\dots,\nprops-1))=\nprops\Delta(\cU_{\circleddash}^\nprops) \, .
\end{equation}
Both polytopes are reflexive and isomorphic. Therefore, their respective associated Feynman integrals are equal. Indeed, the Feynman integral for the one-loop $\nprops$-gon is given by 
	\begin{equation}\label{e:evalUngon}
			 I_{\nprops-\rm gon}(1,\nprops;(1,\dots,1))= \int_{\mathbb R^{\nprops-1}_+}
		\left.	{1\over \left(\sum_{i=1}^\nprops x_i\right)^{\nprops} }\right|_{x_\nprops=1}\dd x_1\cdots \dd x_{\nprops-1} \,,
                      \end{equation}
and the one for the  $\nprops-1$-loop sunset integral is
\begin{equation}
 I_{\circleddash}(L,2L+2; (L,\dots,L))=   \int_{\mathbb R_+^L} { x_1^L\cdots x_{L}^L  \over
  \left(x_1\cdots x_{L}\left(1+\sum_{i=1}^Lx_i^{-1}\right)\right)^{L+1}} \prod_{i=1}^L {\dd x_i\over x_i}\,.
\end{equation}
Performing the Cremona transformation $x_i\to 1/x_i$, we find that the two integrals are equal:
\begin{equation}
 I_{\circleddash}(L,2L+2;(L,\dots,L))=  I_{\nprops-\rm gon}(1,\nprops;(1,\dots,1)) ={1\over (\nprops-1)!}\,.
\end{equation}

\subsection{Family of graphs of type $(1,\dots,1,n)$}\label{sec:Family1n}
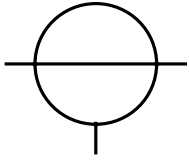
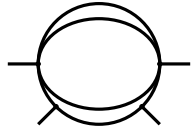
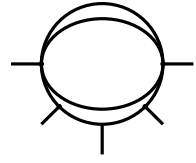
\begin{figure}[ht]
  \centering
  	\begin{subfigure}{0.3\textwidth}
    \centering
\begin{tikzpicture}[scale=0.4]
  \filldraw [color = black, fill=none, very thick] (0,0) circle (2cm);
 \draw [very thick] (-2,0) -- (2,0);
\filldraw [black] (2,0) circle (2pt);
\filldraw [black] (0,-2) circle (2pt);
\filldraw [black] (-2,0) circle (2pt);
\draw [black,very thick] (-2,0) to (-3,0);
\draw [black,very thick] (2,0) to (3,0);
\draw [black,very thick] (0,-2) to (0,-3);
\end{tikzpicture}
	\caption{The graph  $(1,1,2)$. }
        \label{fig:threeloop112}
      \end{subfigure} \hfill
  \begin{subfigure}{0.3\textwidth}
  \centering
      \begin{tikzpicture}[scale=0.4]
  \filldraw [color = black, fill=none, very thick] (0,0) circle (2cm);
  \draw[very thick] (0,0) [partial ellipse=0:180:2cm and 1.5cm];
  \draw[very thick] (0,0) [partial ellipse=180:360:2cm and 1.5cm];
\filldraw [black] (2,0) circle (2pt);
\filldraw [black] (-1.414,-1.414) circle (2pt);
\filldraw [black] (1.414,-1.414) circle (2pt);
\filldraw [black] (-2,0) circle (2pt);
\draw [black,very thick] (-2,0) to (-3,0);
\draw [black,very thick] (2,0) to (3,0);
\draw [black,very thick] (-1.414,-1.414) to (-2,-2);
\draw [black,very thick] (1.414,-1.414) to (2,-2);
\draw [white,very thick] (0,-2.1) to (0,-3);
\end{tikzpicture}
	\caption{The graph  $(1,1,1,3)$. }
        \label{fig:threeloop1113}
      \end{subfigure} \hfill
      	\begin{subfigure}{0.3\textwidth}
    \centering
       \begin{tikzpicture}[scale=0.4]
  \filldraw [color = black, fill=none, very thick] (0,0) circle (2cm);
  \draw[very thick] (0,0) [partial ellipse=0:180:2cm and 1.5cm];
  \draw[very thick] (0,0) [partial ellipse=180:360:2cm and 1.5cm];
\filldraw [black] (2,0) circle (2pt);
\filldraw [black] (-1.414,-1.414) circle (2pt);
\filldraw [black] (1.414,-1.414) circle (2pt);
\filldraw [black] (-2,0) circle (2pt);
\draw [black,very thick] (-2,0) to (-3,0);
\draw [black,very thick] (2,0) to (3,0);
\draw [black,very thick] (-1.414,-1.414) to (-2,-2);
\draw [black,very thick] (1.414,-1.414) to (2,-2);
\draw [black,very thick] (0,-2) to (0,-3);
\end{tikzpicture}
	\caption{The graph  $(1,1,1,1,4)$. }
        \label{fig:threeloop11114}
      \end{subfigure}
\caption{Family of graphs with edge label $(1,\dots,1,n)$.}\label{fig:edgetosunset}
\end{figure}

We consider now the family of graphs in \cref{fig:edgetosunset} composed by a $n$-loop sunset with $n-1$ legs attached to one edge, which we evaluate in dimension $D=2(n+1)$. These graphs have $\nprops=2n$ edges. For $n=2$ this is the ice-cream two-loop graph, for $n=3$ this is the last graph in \cref{example-symmetries}, and for $n=4$ this is the graph in \cref{dancing-boy} now labeled $(1,1,1,1,4)$

The first Symanzik graph polynomial for these graphs is given by
\begin{equation}
    \cU_{(1,\dots,1,n)}=(x_1+\cdots+x_{n})(x_{n+1}+\cdots+x_{2n})+x_{n+1}\cdots x_{2n}\,,
\end{equation}
with the associated  polytope 
\begin{equation}
  \Newton(n,2(n+1),\undernotation{\nu}_{1,n}) =(n+1) \Delta(\cU_{(1,\dots,1,n}) \, .
\end{equation}
This polytope is reflexive with the interior point given by
\begin{equation}
    \undernotation{\nu}_{1,n}=(1,\dots,1,n, \dots, n)\,.
\end{equation}
Their associated Feynman integrals are
\begin{equation}
    \label{e:In} I(n,2(n+1),\undernotation{\nu}_{1,N})=\int_{\mathbb R_+^{2n-1}}  \left.{x_1\cdots x_{n} x_{n+1}^{n}\cdots x_{2n}^{n} \over \cU_{(1,\dots,1,n)}^{n+1}}\right|_{x_{2n}=1}\,\prod_{i=1}^{2n-1}{dx_i\over x_i}\,.
\end{equation}
Performing the change of variables
\begin{equation}
    (x_{n+1},\dots,x_{n})\to  (1/x_{n+1},\dots,1/x_{n}),
\end{equation}
gives
\begin{equation}  I(n,2(n+1),\undernotation{\nu}_{1,n})=\int_{\mathbb R_+^{2n-1}}  \left.{x_1\cdots x_{2n}\over (1+(x_1+\cdots+x_{n})(x_{n+1}+\cdots +x_{2n}))^{n+1}}\right|_{x_{2n}=1}\,\prod_{i=1}^{2n-1}{dx_i\over x_i}\,,
\end{equation}
which evaluates to 
\begin{equation}\label{e:I1neval}
    I(n,2(n+1),\undernotation{\nu}_{1,n})={1\over n! (n-1)!}\,.
\end{equation}

\medskip

The first graph in fig.~\ref{example-symmetries} arises in the calculation of the beta function of the $\varphi^4$ theory. Its first Symanzik polynomial is $\cU_{(1,\dots,1,3)}$.  In $D=4-2\varepsilon$ the graph  is divergent but is quasi-finite   in $D=6-2\varepsilon$. Our analysis leads to  reflexivity for $(n_\cU, n_\cF)=(4,0)$.  The associated integral in \cref{e:I1neval} evaluates to $1/12$, which matches the leading order term of the result given in ref.~\cite{Kompaniets:2017yct}. 

\subsection{Family of graphs of type $(1,\dots,1,n,n)$}\label{sec:Family1nn}

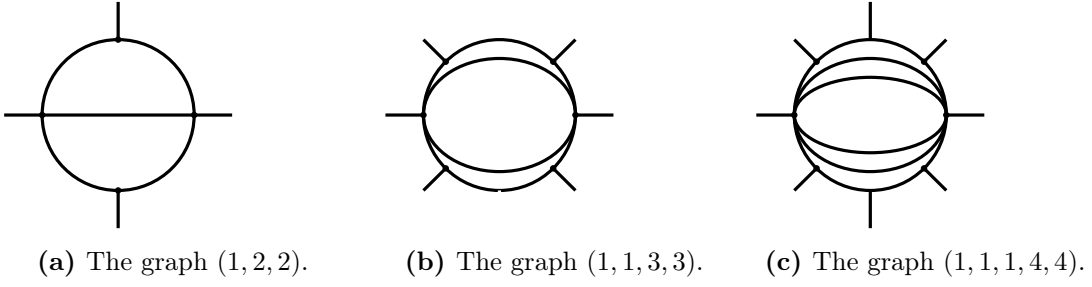
\begin{figure}[ht]
  \centering
  \begin{subfigure}{0.3\textwidth}
      \begin{tikzpicture}[scale=0.5]
				\filldraw [color = black, fill=none, very thick] (0,0) circle (2cm);
				\draw [black,very thick] (-2,0) to (2,0);
				\filldraw [black] (2,0) circle (2pt);
				\filldraw [black] (-2,0) circle (2pt);
				\filldraw [black] (0,2) circle (2pt);
				\filldraw [black] (0,-2) circle (2pt);
				\draw [black,very thick] (-2,0) to (-3,0);
				\draw [black,very thick] (2,0) to (3,0);
				\draw [black,very thick] (0,2) to (0,3);
				\draw [black,very thick] (0,-2) to (0,-3);
			\end{tikzpicture}
			\caption{The graph $(1,2,2)$.}\label{fig:twoloop122}
  \end{subfigure}
\hfil
  \begin{subfigure}{0.3\textwidth}
      \begin{tikzpicture}[scale=0.5]
  \filldraw [color = black, fill=none, very thick] (0,0) circle (2cm);
  \draw[very thick] (0,0) [partial ellipse=0:180:2cm and 1.5cm];
  \draw[very thick] (0,0) [partial ellipse=180:360:2cm and 1.5cm];
\filldraw [black] (2,0) circle (2pt);
\filldraw [black] (-1.414,-1.414) circle (2pt);
\filldraw [black] (1.414,-1.414) circle (2pt);
\filldraw [black] (1.414,1.414) circle (2pt);
\filldraw [black] (-1.414,1.414) circle (2pt);
\filldraw [black] (-2,0) circle (2pt);
\draw [black,very thick] (-2,0) to (-3,0);
\draw [black,very thick] (2,0) to (3,0);
\draw [black,very thick] (-1.414,-1.414) to (-2,-2);
\draw [black,very thick] (1.414,-1.414) to (2,-2);
\draw [black,very thick] (1.414,1.414) to (2,2);
\draw [black,very thick] (-1.414,1.414) to (-2,2);
\draw [white,very thick] (0,-2) to (0,-3);
\end{tikzpicture}
\caption{The graph  $(1,1,3,3)$.}\label{fig:threeloop1133}
  \end{subfigure}\hfil
  \begin{subfigure}{0.3\textwidth}
\begin{tikzpicture}[scale=0.5]
  \filldraw [color = black, fill=none, very thick] (0,0) circle (2cm);
  \draw[very thick] (0,0) [partial ellipse=0:180:2cm and 1.5cm];
  \draw[very thick] (0,0) [partial ellipse=180:360:2cm and 1.5cm];
  \draw[very thick] (0,0) [partial ellipse=0:180:2cm and 1cm];
  \draw[very thick] (0,0) [partial ellipse=180:360:2cm and 1cm];
\filldraw [black] (2,0) circle (2pt);
\filldraw [black] (-1.414,-1.414) circle (2pt);
\filldraw [black] (1.414,-1.414) circle (2pt);
\filldraw [black] (1.414,1.414) circle (2pt);
\filldraw [black] (-1.414,1.414) circle (2pt);
\filldraw [black] (-2,0) circle (2pt);
\draw [black,very thick] (-2,0) to (-3,0);
\draw [black,very thick] (2,0) to (3,0);
\draw [black,very thick] (-1.414,-1.414) to (-2,-2);
\draw [black,very thick] (1.414,-1.414) to (2,-2);
\draw [black,very thick] (1.414,1.414) to (2,2);
\draw [black,very thick] (-1.414,1.414) to (-2,2);
\draw [black,very thick] (0,2) to (0,3);
\draw [black,very thick] (0,-2) to (0,-3);
\end{tikzpicture}
\caption{The graph $(1,1,1,4,4)$. }
 \end{subfigure}
\caption{The multiloop graphs with edge label $(1,\dots,1,n,n)$.}\label{fig:generalisedeyes}
\end{figure}

The list of reflexive polytopes contains graphs built by attaching the same number of external legs to two lines of the multiloop sunset graph as shown in \cref{fig:generalisedeyes}.  We label the graph by the number of edges attached to the multiloop sunset graph. The graphs with less than ten edges in this family are the kite graph, shown in \cref{fig:twoloop122}, with label $(1,2,2)$, and the three-loop graph in \cref{fig:threeloop1133}, with label $(1,1,3,3)$.

The first Symanzik polynomial of this family of graphs  reads
\begin{equation}
    \cU_{(1,\dots,1,n,n)}=\left(\sum_{i=1}^n x_i\right) \left(\sum_{i=n+1}^{2n}x_i\right) \prod_{i=2n+1}^{3n-1}x_i \left(\sum_{i=2n+1}^{3n-1} {1\over x_i}\right)+\left(\sum_{i=1}^{2n}x_i\right) \prod_{i=2n+1}^{3n-1}x_i\,.
\end{equation}
    These graphs have $n$ loops and $\nprops=3n-1$ edges. In dimension $D=2(n+1)$, the Newton polytope $(n+1)\Newton(\cU_{(1,\dots,1,n,n)})$ is reflexive, with a single interior point $\undernotation{\nu}=(1,\dots,1,n,\dots,n)$.
The associated integral is given by
\begin{equation}
I_{(1,\dots,1,n,n)}=\int_{\mathbb R_+^{3n-2}} \left. x_1\cdots x_{2n} (x_{2n+1} \cdots x_{3n-1})^n\over \cU_{(1,\dots,1,n,n)}^{n+1}\right|_{x_{3n-1}=1} \prod_{i=1}^{3n-2}{dx_i\over x_i}\,.    
\end{equation}
Changing variables as $x_i\to 1/x_i$ for $2n+1\leq i\leq 3n-1$, one gets the integral
\begin{equation}
I_{(1,\dots,1,n,n)}=\int_{\mathbb R_+^{3n-2}} \left. x_1\cdots x_{3n-1} \over \hat \cU_{(1,\dots,1,n,n)}^{n+1}\right|_{x_{3n-1}=1} \prod_{i=1}^{3n-2}{dx_i\over x_i}\,,
\end{equation}
with 
\begin{equation}
    \hat \cU_{(1,\dots,1,n,n)}=\left(\sum_{i=1}^n x_i \right)\left(\sum_{i=n+1}^{2n}x_i\right) \prod_{i=2n+1}^{3n-1}x_i+\sum_{i=1}^{2n}x_i \,.
\end{equation}
This integral evaluates to
\begin{equation}
I_{(1,\dots,1,n,n)}={1\over (n-1)!^2 n!}\,.
\end{equation}

\subsection{Feynman graphs which evaluate to zeta values}\label{sec:Intzeta}

  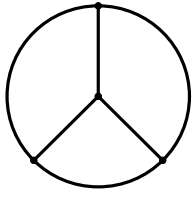
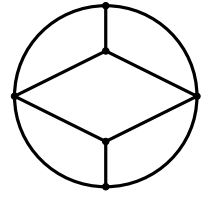
\begin{figure}[ht]
          \centering
          \begin{subfigure}{0.4\textwidth}
          \centering
			\begin{tikzpicture}[scale=0.6]
				\filldraw [color = black, fill=none, very thick] (0,0) circle (2cm);
				\filldraw [black] (0,0) circle (2pt);
				\filldraw [black] (0,2) circle (2pt);
				\filldraw [black] (1.414,-1.414) circle (2pt);
				\filldraw [black] (-1.414,-1.414) circle (2pt);
                                \draw [black,very thick] (1.414,-1.414) to (0,0);
                                \draw [black,very thick]   (-1.414,-1.414) to (0,0);
                                \draw [black,very thick] (0,0) to
                                (0,2);
                              \end{tikzpicture}\caption{The wheel
                                with three spokes. }\label{fig:Wheel3spokes} \end{subfigure}\hfill
                                \begin{subfigure}{0.4\textwidth}
                                          \centering\begin{tikzpicture}[scale=0.6]
				\filldraw [color = black, fill=none,  very thick] (0,0) circle (2cm);                        
				\filldraw [black] (0,2) circle (2pt);
				\filldraw [black] (2,0)  circle (2pt);
                                \filldraw [black] (-2,0) circle (2pt);
                                \filldraw [black] (0,-2) circle (2pt);
                                \filldraw [black] (0,1) circle (2pt);
                                \filldraw [black] (0,-1) circle (2pt);
				\draw [black,very thick]  (0,2) to (0,1);
                                \draw [black,very thick] (0,-2) to (0,-1);
				\draw [black,very thick]   (2,0) to (0,-1);
                                \draw [black,very thick]   (2,0) to (0,1);
                                \draw [black,very thick]   (-2,0) to (0,-1);
                                \draw [black,very thick]   (-2,0) to (0,1);
                                                     \end{tikzpicture}
                            \caption{The diamond circle graph. }\label{fig:diamondcircle}
                          \end{subfigure}
                          \caption{Two graphs that evaluate to zeta values.}
	\label{fig:graphzeta}	\end{figure}

The polytope associated with the graph in \cref{fig:Wheel3spokes} is $\Delta_{W_3}(3,4;(1,\dots,1))=4\Newton(\cU_{W_3})$
 with 
  \begin{multline}
  \cU_{W_3}= x_{1} x_{2}+x_{3} x_{2}+x_{1} x_{4} x_{2}+x_{3} x_{4} x_{2}+x_{1} x_{5} x_{2}+x_{3} x_{5}
  x_{2}+x_{4} x_{5} x_{2}+x_{5} x_{2}+x_{1} x_{3}\cr
  +x_{1} x_{4}+x_{1} x_{3} x_{4}+x_{3} x_{4}+x_{1}
   x_{5}+x_{1} x_{3} x_{5}+x_{3} x_{4} x_{5}+x_{4} x_{5}\,,
 \end{multline}
 and the polytope associated with the five-loop diamond~circle graph in \cref{fig:diamondcircle} is 
  $\Delta_{\textrm{diamond~circle}}(5,4;(1,\dots,1))=2\Newton(\cU_{\textrm{diamond~circle}})$  with the first Symanzik polynomial $\cU_{\textrm{diamond~circle}}$  given in \cref{e:Udiamondcircle}. 
  Both polytopes are reflexive with interior point $\undernotation{\nu}=(1,\dots,1)$.

Finally, the Feynman integral for the graph in \cref{fig:Wheel3spokes} in four dimensions evaluates to 
  \begin{equation}\label{e:IW3}
    I_{W_3}(3,4; (1,1,1,1,1,1))=    \int_{\mathbb R_+^5} \left. {\dd x_1\cdots \dd x_5\over \cU_{W_3}^2}\right|_{x_6=1}=6\zeta(3)
\end{equation}
  in agreement with~\cite{Brown:2021umn}.
The Feynman integral for the graph in \cref{fig:diamondcircle} in four dimensions evaluates to
  \begin{equation}
    I_{\rm diamond~circle}(5,4;(1,\dots,1))=\int_{\mathbb R_+^9} \left. {\dd x_1\cdots
      \dd x_{9}\over \cU_{\rm diamond~circle}^2}\right|_{x_{10}=1}=36\zeta(3)^2 \,.   
  \end{equation}
  This integral is the square of the one in \cref{e:IW3} as a consequence of conformal invariance~\cite{Broadhurst:1985tld, Panzer:2019yxl}.

\section{Calabi--Yau from reflexive polytopes of Feynman integrals}\label{sec:mirror}
In this \namecref{sec:mirror}, we elaborate on how Calabi--Yau varieties naturally arise from Feynman integrals when their Symanzik graph polynomials are interpreted through toric geometry.

In the toric framework, mirror symmetry emerges directly from the geometry of the Newton polytope. A reflexive polytope $\Delta$ determines a Gorenstein toric Fano variety $\mathbb P_\Delta$, while its dual $\PolarPolytope$ determines the mirror partner $\mathbb P_{\PolarPolytope}$, see e.g., ref.~\cite{Kasprzyk_2022}. The anticanonical hypersurfaces in these varieties form a mirror pair of Calabi--Yau manifolds in Batyrev's sense~\cite{Batyrev:1993oya}. When the Symanzik polynomials of a Feynman graph  define a reflexive polytope, the graph integral naturally embeds in this mirror pair geometry.

Products of powers of the first and second Symanzik polynomials define hypersurfaces in projective space, whose Newton polytopes (from scaled Minkowski sums) can be reflexive or Fano. When reflexive, the associated toric variety admits a Calabi--Yau hypersurface in its anticanonical linear system. The geometry of vanishing locus of the product of the powers of the Symanzik polynomials determines a degeneration of a del Pezzo surface, $K3$ surface, or higher-dimensional Calabi--Yau variety. 

This connection bridges quantum field theory and mirror symmetry. The Feynman integral computes a period of the graph hypersurface, while the reflexive polytope determines dual toric Fano varieties whose anticanonical hypersurfaces form a Batyrev mirror pair~\cite{Batyrev:1993oya}. Each Feynman graph with a reflexive Newton polytope, therefore, provides access to periods on both sides of a mirror pair. The integral's analytic behavior---rational, dilogarithmic, elliptic, or Calabi--Yau type---corresponds to the dimension and degeneration type of the underlying toric hypersurface.
We discuss  the  cases of the multiloop sunset and multileg one-loop graphs and their toric Calabi--Yau structures.
For one-loop graphs, the explicit computations of the
$\cU_\Gamma$ and $\cF_\Gamma$ polynomials are particularly transparent.  Their Newton
polytopes are  low-dimensional reflexive polytopes, and the corresponding
graph hypersurfaces may be interpreted as anticanonical hypersurfaces in
toric del Pezzo surfaces, quartic $K3$ surfaces, or (in the case of the
pentagon) degenerations of the quintic Calabi--Yau threefold.  These
examples illustrate how the analytic structure of the integral matches
the geometry of the toric Calab--Yau.

\subsection{The multiloop sunset graphs}\label{sec:sunsetmirror}

The Newton polytope for the generic massive multiloop sunset graph of \cref{sec:sunset}  is
\begin{equation}
  \Delta_{\circleddash}^\nprops=\simplexwithpars{\nprops,1} +\Newton(	\cU^{\nprops}_{\circleddash})\,.
\end{equation}
This polytope is detailed in \cref{app:sunsetpolytope} and  it is reflexive. In
		figure~\ref{twothreeloopssunset}, we present the
                reflexive polytopes for the two-loop and three-loop 
		sunset graphs. 	Verrill~\cite{verrill1996root} has studied the toric variety $\mathbb
		P_{\Newton_{\circleddash}^\nprops}$ defined by the Newton polytope $\Newton_{\circleddash}^\nprops$. The toric variety is obtained by blowing up the strict
		transformation of the space spanned by the coordinate points
		$(1,0,\dots,0), \dots, (0,\dots,0,1)$ in $\mathbb P^{\nprops-1}$. This
		defines a  family of (singular) Calabi--Yau  $(\nprops-2)$-folds. For
		$\nprops=3$ the toric variety is obtained by the blowing up of three points
		in $\mathbb P^2$ and leads to the del Pezzo surface $dP_6$ studied
		in refs.~\cite{Bloch:2013tra,Bloch:2016izu}; for $\nprops=4$ this is a $K3$
		surface studied
		in refs.~\cite{Bloch:2014qca,Broedel:2019kmn,Broedel:2021zij,Pogel:2022yat,Duhr:2025ppd};
		for $\nprops=5$ loops, we have the Calabi--Yau 3-fold~\cite{HulkerVerrill3fold}; and
		particular cases of the Calabi--Yau fourfold arise at  $\nprops=6$ loops~\cite{HulekVerrill4fold}.

		\begin{figure}[ht]
			\centering
			\begin{subfigure}{0.45\textwidth}
				\centering
				\includegraphics[width=5cm]{polytope-massless-triangle.png}
                                \caption{2d reflexive  \#9 in the {\tt sagemath}~\cite{sagemath}
                          database of polytopes.}\label{fig:2sunsetpolytope}
			\end{subfigure}\hfill
			\begin{subfigure}{0.45\textwidth}
				\centering
				\includegraphics[width=8cm]{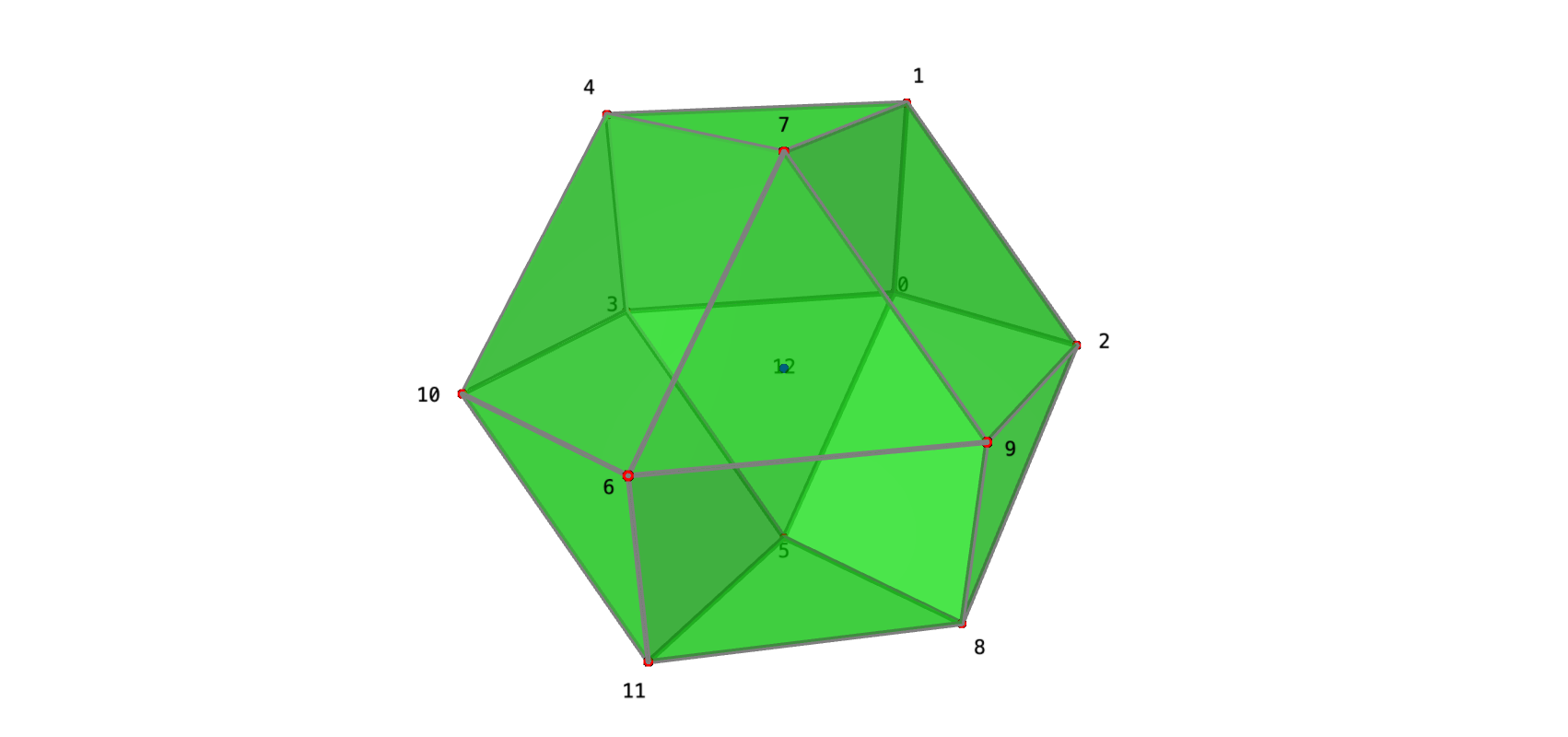}
                                \caption{3d reflexive  \#1529 in the {\tt sagemath}~\cite{sagemath}
                          database of polytopes.}\label{fig:3sunsetpolytope}
			\end{subfigure}
			\caption{Reflexive polytopes for the two- and
                          three-loop sunset graphs.}
			\label{twothreeloopssunset}
		\end{figure}

	\subsection{Hodge  numbers of $\nprops$-$\textrm{gons}$}\label{sec:Hodgenumbers}
    \begin{table}[tb]
			\centering
			\setlength{\tabcolsep}{1em}
			\begin{tabular}{@{}lrrrrrrr@{}}
                          $N$   &4 & 5 &6 &7 &8 &9&10  \\
                    				\midrule
				$h_{1,N-3}$  &20 & 101&426 &1667 &6371 & 24229& 92278\\
			\end{tabular}
			\caption{Non-vanishing Hodge numbers for $N$-gon. In all cases with $N>4$, $h_{11}=1$.}
			\label{Hodgengon}
		\end{table} 
		
Hodge numbers of a hypersurface and its mirror are  related by \cite{Batyrev:1993oya}
\begin{equation}
h_{1,1}(V)=h_{n-2,1}(V^{\mathrm{dual}}), \quad
h_{n-2,1}(V)=h_{1,1}(V^{\mathrm{dual}}) \,.
\end{equation}
These numbers can also be computed from the polytopes, e.g., using \texttt{PALP}.	

Reflexive polytopes associated to the massive one-loop $\nprops$-gon graph are given by the scaled standard simplex $\nprops \simplexwithpars{\nprops,1}$.
        The corresponding polar polytope has the simple expression
		\begin{equation}
			\nabla_{\nprops-\rm
                          gon}(1,D;(1,\dots,1))=\Conv\big\{
			e_1,e_2, \dots, e_{\nprops-1}, -\sum_{i=1}^{\nprops-1} e_i
			\big\}	\, , 
		\end{equation}
		where $e_i$ are unit vectors or $\mathbb R^{\nprops-1}$. We
                thus find that the one-loop $\nprops$-gon encodes the
                projective space $\mathbb{P}^{\nprops-1}$, which is the
                associated toric variety of the polar polytope, see e.g., ref.~\cite[Chapter 1]{Fulton93}.
                The
Hodge numbers are presented in \cref{Hodgengon}.
                They all vanish except for the numbers
                $h_{1,1}(\nprops)$ and $h_{1,N-3}(\nprops)$. The Hodge
                number $h_{\nprops,0}=h_{1,1}(\nprops)=1$  for all cases.
		The Hodge numbers $h_{1,N-3}(\nprops)$ can be extracted
                from the Hirzebruch generating formula for a degree
                $\nprops$ hypersurface $X$ in  $\mathbb{P}^{\nprops}$ given by~\cite[Theorem~17.3.4]{Arapura2012}
		\begin{equation}
			H(\nprops)=\sum_{p,q}(h_{p,q}(\nprops)-\delta_{pq}) x^p y^q=\frac{(1+y)^{\nprops-1}-(1+x)^{\nprops-1}}{(1+x)^{\nprops }y-(1+y)^{\nprops }x}\,.
			\label{formal-series-Hd}
                      \end{equation}

The varieties defined by the singular locus of
the one-loop $\nprops$-gon graphs are not smooth, although they share 
the same polar polytope as the smooth Calabi--Yau $(\nprops-2)$-folds with
the given Hodge numbers.
For instance, the reflexive polytope for the pentagon graph
($\nprops=5$) in $D=6$, $D=8$ or $D=10$ dimensions is the same as the one for a smooth quintic  Calabi--Yau  threefold  in $\mathbb P^4$.  But the one-loop pentagon variety is a
singular toric degeneration of a smooth  quintic threefold. This will
be further discussed in \cref{sec:pentagon-graph-CY}.

\subsection{The massless triangle and the del Pezzo surface $dP_6$}
\label{massless-triangle}
The massless triangle graph has three edge parameters $x_1,x_2,x_3$.
Its Symanzik polynomials read 
\begin{equation}
\cU_{\rm triangle}=x_1+x_2+x_3, \qquad
\cV_{\rm triangle}=x_1 x_2\, p_3^2 + x_1 x_3\, p_2^2 + x_2 x_3\, p_1^2 .
\end{equation}
In $D=4$ dimensions the finite integral is 

	\begin{equation}\label{e:triangle-massless-fourD}
			 I_{\rm triangle}^0(1,4;(1,1,1))=
                        \int_{\mathbb R_+^{2}}\left. {1\over \cU_{\rm
					triangle}\cV_{\rm triangle}}\right|_{x_3=1} \dd x_1  \dd x_2\, .  
                                  \end{equation}
The product $\cU_{\rm triangle} \cV_{\rm triangle}$ is a homogeneous cubic polynomial
in three variables. 
 
The Newton polytope $\Newton(\cU_{\rm triangle})+\Newton(\cV_{\rm triangle})$
is the two-dimensional reflexive
hexagon shown in \cref{fig:polytopemasslesstriangle}, the unique reflexive polytope in dimension two with six
vertices.  The associated toric Fano surface is the del Pezzo surface
$dP_6$ (see e.g. refs.~\cite{Manin1986cubic,Dolgachev_CAG} for a definition), obtained as the blow-up of $\mathbb{P}^2$ at three non-collinear
points. The 
anticanonical hypersurface on $dP_6$ is the cubic
\begin{multline}\label{e:dP6}
  a_{0} z_{0}^2 z_{1}^2 z_{2} z_{3}+a_{1} z_{0} z_{1}^2
   z_{3}^2 z_{4}+a_{2} z_{0}^2 z_{1} z_{2}^2 z_{5}+a_{3}
   z_{1} z_{3}^2 z_{4}^2 z_{5}+a_{4} z_{0} z_{2}^2 z_{4}
   z_{5}^2+a_{5} z_{2} z_{3} z_{4}^2 z_{5}^2\cr+a_{6} z_{0}
   z_{1} z_{2} z_{3} z_{4} z_{5}=0\, ,
\end{multline}
which defines a smooth elliptic curve.
The two-loop sunset hypersurface (cf.\ \cref{secondSymanzik-banana})
\begin{equation}
\{\cF_{\circleddash}^2=(x_1x_2+x_1x_3+x_2x_3)(m_1^2x_1+m_2^2x_2+m_3^2x_3)-p^2x_1x_2x_3=0\}
\end{equation}
is a specialization of the of the cubic in \cref{e:dP6} 
after setting $z_0 = x_1/x_2 z_3 z_4/z_2$, $ z_1 = x_2/x_3 z_2 z_5/z_3$ and  identifying the coefficients $a_0=a_2=m_1^2$, $a_1=a_3=m_2^2$, $a_4=a_5=m_3^2$ and $a_6=m_1^2+m_2^2+m_3^2-p^2$.
The two-loop sunset Feynman integral 
\begin{equation}
    I_{\circleddash}(2,2;(1,1,1))=\int_{\mathbb R_+^2}\left.{1\over \cF_{\circleddash}}\right|_{x_3=1} \dd x_1\dd x_2
\end{equation}
is then interpreted as a
regulator period, and the limiting mixed Hodge structure realizes a
version of local mirror symmetry for $dP_6$~\cite{Bloch:2016izu}.

\medskip

On the other hand, the hypersurface $\{\cU_{\rm triangle}\cF_{\rm triangle}=0\}$
is reducible  and it is obtained with the same identification as above except for the coefficient $a_6=m_1^2+m_2^2+m_3^2$ so that \cref{e:dP6} factorizes as
\begin{equation}
    (z_{0} z_{1} z_{3}+z_{0} z_{2} z_{5}+z_{3} z_{4}
   z_{5}) \left(m_{1}^2 z_{0} z_{1} z_{2}+m_{2}^2 z_{1}
   z_{3} z_{4}+m_{3}^2 z_{2} z_{4} z_{5}\right)=0\,.
\end{equation}
This singular hypersurface  is a degeneration
of the smooth sunset elliptic curve.
The integral in \cref{e:triangle-massless-fourD} evaluates to a single-valued
dilogarithm~\cite{Chavez:2012kn}. The value of this integral corresponds to the limit
$p^2=0$ of the two-loop sunset integral~\cite{Bloch:2013tra}. 

\medskip

Although simpler, the massless triangle fits in the same geometric
paradigm: it computes a period of an elliptic curve associated with an
(open) anticanonical divisor in $dP_6$.  The discriminant of the cubic
polynomial and its $j$--invariant encode the analytic behavior of the
integral in terms of elliptic polylogarithms, exactly mirroring the
sunset geometry.  In particular, the degeneration limits of the cubic
curve correspond to the expected boundary behaviors in the space of
kinematic invariants, and the variation of its $j$--invariant matches
the singularity structure of the Feynman integral.

\subsection{The massive triangle and $\mathbb P^2$}
\label{massive-triangle}
In the massive triangle the Symanzik polynomial and its product $\cU_{\rm triangle} \cF_{\rm triangle}$ remains a cubic polynomial in
three variables
\begin{equation}
\cF_{\rm triangle}=(x_1+x_2+x_3) \left(m_1^2x_1+m_2^2x_2+m_3^2x_3\right)-\left(x_1 x_2\, p_3^2 + x_1 x_3\, p_2^2 + x_2 x_3\, p_1^2\right).
\end{equation}
       The polytope for the massive triangle  in $D=4$ and $D=6$ is the translation of the standard simplex
       scaled by a factor of  three (cf.
\cref{sec:onelooppolytope} for details)
	   \begin{equation}\label{e:DeltaMassiveTriangle}
                  \Newton_{\rm triangle}(1,4;(1,1,1))  =3\simplexwithpars{3,1}\,. 
                \end{equation}
The toric Fano variety $\mathbb P_{\Delta_{\rm triangle}}$ is the projective plane $\mathbb{P}^2$, since the dual fan of any reflexive lattice triangle is the fan of $\mathbb{P}^2$ (see~\cite{Batyrev:1993oya,CoxKatz1999}). The  corresponding anticanonical hypersurface is the plane cubic
\begin{equation}\label{e:dP3}
a_0z_0^3 + a_1z_1^3+ a_2z_2^3 + a_3z_0z_1z_2=0\,.
\end{equation}
Thus, the Fano geometry encoded by the simplex $\Newton_{\rm triangle}(1,4;(1,1,1))$ is that of $\mathbb{P}^2$, and its anticanonical hypersurfaces are smooth elliptic curves.

However, the geometry undergoes a notable
degeneration.  Setting $a_1=a_2=0$,  $a_0=c_{11} c_{12} c_{13}$, $a_3=-c_{12}/c_{21}$, and $z_0=x_1+x_2+x_3$ with $z_i=\sum_{r=1}^3 c_{ir}x_i$, and  $p_1^2=c_{11}(c_{12}-c_{13})^2$, $p_2^2=c_{12}(c_{11}-c_{13})^2$ and $p_3^2=c_{13}(c_{11}-c_{12})^2$, the cubic in \cref{e:dP3} degenerates into the expression of the massive triangle graph hypersurface
\begin{equation}
\cU_{\rm triangle} \cF_{\rm triangle}
  = z_0\,(a_0 z_0^2 + a_3 z_1 z_2)\,. 
\end{equation}
This  union of a line and a
conic  is a degeneration of an elliptic curve to a rational
curve.
The massive triangle integral therefore computes periods of a rational
curve in a degeneration of the anticanonical elliptic curve, in accordance with the fact that the
integral
 \begin{equation}
                    \label{e:massivetriangle4d}
                    I_{3-\rm gon}(1,4;(1,1,1))= \int_{\mathbb R_+^2}
                    {1\over \cU_3\cF_3}\Big|_{x_3=1}\dd x_1\dd x_2\,
                  \end{equation}
can be expressed in terms of classical polylogarithms~\cite{tHooft:1978jhc} rather
than elliptic integrals.

\subsection{The massive box graph and quartic $K3$ surfaces}
\label{massive-box}

\begin{figure}[ht] 
  \centering
  \includegraphics[width=5cm]{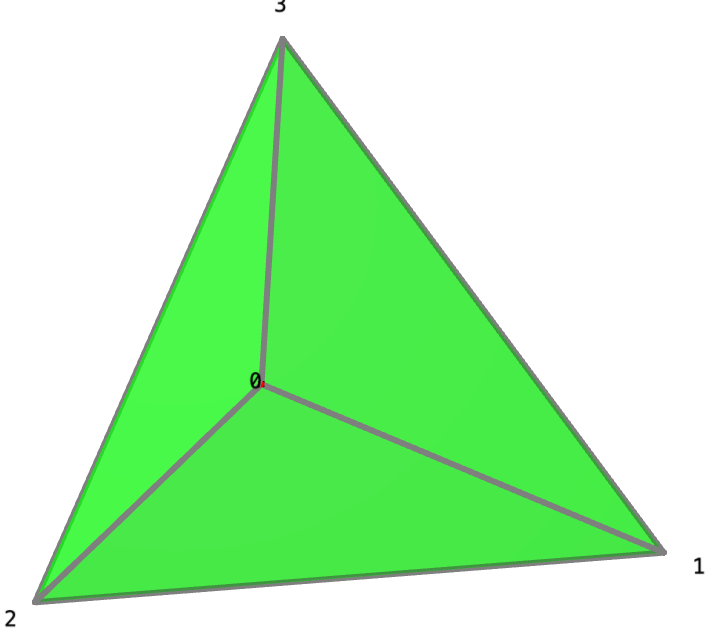}
  \caption{The polytope for the massive box in $D=4$. This polytope
    has the number \#4311 in the {\tt sagemath}~\cite{sagemath} database.}
  \label{fig:massivesboxpolytope}
\end{figure}

The one-loop massive box graph has four variables
$x_1,\dots,x_4$, and the graph polynomials are
\begin{equation}
\begin{aligned}
  \cU_{\rm box}&=x_1+x_2+x_3+x_4, \qquad \cV_{\rm box}=\sum_{1\leq
                 i<j\leq 4} (p_{i}+\cdots+p_{j-1})^2 x_ix_j,\cr
\cF_{\rm box}&=\cU_{\rm box}(m_1^2x_1+\cdots +m_4^2x_4) -\cV_{\rm box}\,  .              \end{aligned}
\end{equation}
The 
Newton polytopes for the massive box integral in $D=4$, $D=6$,  and
$D=8$ are all identical 
because of the relation in \cref{e:FMinkowski} $\Newton(\cF_{\rm
  box})=   \Newton(\cU_{\rm box})+ \Newton(x_1+x_2+x_3+x_4)$ and
$\Newton(\cU_{\rm box})= \Newton(x_1+x_2+x_3+x_4)$.
This polytope is 
four times the three-dimensional standard
simplex (see \cref{sec:onelooppolytope} for details):
\begin{equation}
2\Newton(\cF_{\rm box})=   \Newton(\cU_{\rm box}) +  \Newton(\cF_{\rm
  box})= 4 \Newton(\cU_{\rm box})= 4\simplexwithpars{4,1} \,.
\end{equation}
This polytope, shown in \cref{fig:massivesboxpolytope}, is
reflexive, and the corresponding toric Fano variety is $\mathbb{P}^3$.
Its anticanonical hypersurfaces are quartic surfaces, all of which are  (smooth or mildly singular) quartic $K3$ surfaces. 
The box hypersurface $\{\cF_{\rm box}^2=0\}$,
$\{\cU_{\rm box}^2\cF_{\rm box}=0\}$ or $\{\cU_{\rm box}^4=0\}$ are  singular
degenerations of the quartic $K3$. 
The massive box integrals thus   compute a period of a toric
degeneration of quartic $K3$ surface.

\subsection{The massless box graph and lattice-polarized $K3$ surfaces}
\label{massless-box}

\begin{figure}[ht] 
  \centering
  \includegraphics[width=5cm]{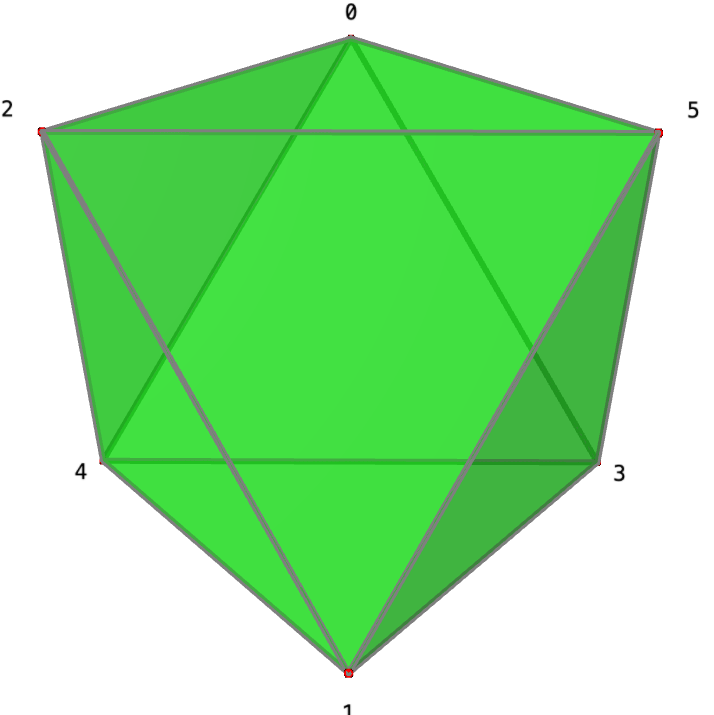}
  \caption{The polytope for the massless box in $D=4$. This polytope
    has the number \#3349 in the {\tt sagemath}~\cite{sagemath} database.}
  \label{fig:masslessboxpolytope}
\end{figure}

The Newton polytope for the massless one-loop box graph in four
dimensions, \cref{fig:masslessboxpolytope},  is given by $  \Newton(\cV_{\rm box}^2)$ which is a three-dimensional reflexive
polytope that is not a simplex.  Concretely, this polytope is the octahedron given by the convex hull of the point $(0,0,2,2)$ and its permutations. 
This is the unique reflexive polytope in dimension three invariant under the full
permutation group $\mathfrak S_4$ of the variables~\cite{Kreuzer:2000xy,Batyrev:1993oya}.  Its polar polytope
is the cube so that the toric ambiant variety is $\mathbb P^1\times \mathbb P^1\times\mathbb P^1$ and the corresponding anticanonical family of hypersurfaces are the  $(2,2,2)$ surfaces in $(\mathbb P^1)^3$.
 From the computation of
\cref{sec:Hodgenumbers} we have that the middle Hodge numbers are $h_{1,1}=20$.

The corresponding toric Fano variety is a smooth toric threefold of
Picard rank six, and a generic anticanonical hypersurface in this
variety is a $K3$ surface. Using Batyrev's mirror symmetry
construction for reflexive polytopes, the dual
reflexive polytope defines a mirror toric Fano threefold whose
anticanonical hypersurfaces form the mirror $K3$ family. The geometric Picard
lattice of a generic member of our $K3$ family has rank $17$, while
the mirror family has geometric Picard rank $3$, in perfect agreement with the
Dolgachev--Nikulin theory of lattice-polarized $K3$
surfaces~\cite{Dolgachev1996,Nikulin1980}.

The massless box integral therefore computes a period of a toric
degeneration of the  lattice-polarized $K3$ surface. The box integral
is given by a dilogarithm function~\cite{tHooft:1978jhc}.
This provides a natural geometric interpretation of the analytic
transition from the $K3$ periods integral to polylogarithmic (and more severe) singularities at special loci in the kinematic space.

\subsection{The pentagon graph and the quintic Calabi--Yau threefold}\label{sec:pentagon-graph-CY}

The graph hypersurfaces for the pentagon are given by
\begin{equation}\label{e:pentagon6d}\cU_{\rm pentagon}\cF_{\rm pentagon}^2=0,\end{equation}   \begin{equation}\label{e:pentagon8d}\cU_{\rm pentagon}^3\cF_{\rm pentagon}=0\, ,\end{equation}
\begin{equation}\label{e:pentagon10d}\cU_{\rm pentagon}^5=0\end{equation}
in $D=6$, $D=8$,
and $D=10$, respectively.  They share the same polytope given by the scaled standard $4$-simplex $5\simplexwithpars{5,1}$ (see
\cref{sec:onelooppolytope} for details), with the graph polynomials
\begin{align}
    \cU_{\rm pentagon}&=x_1+\cdots +x_5,\\
\nonumber    \cF_{\rm pentagon}&=\cU_{\rm pentagon}\left(m_1^2x_1+\cdots +m_5^2x_5\right)-\sum_{1\leq i<j\leq 5} (p_i+\cdots +p_{j-1})^2 x_ix_j\,.
\end{align}

This is the unique reflexive simplex in dimension four.
The associated toric variety is $\mathbb{P}^4$, and the anticanonical
linear system consists of smooth quintic hypersurfaces
\begin{equation}
    a_0z_0^5+\cdots +a_4z_4^5+a_5z_0\cdots z_4=0\,.
\end{equation}
The pentagon hypersurfaces of \cref{e:pentagon6d,e:pentagon8d,e:pentagon10d} are singular specializations of this quintic hypersurface after a linear change of variables.

As in the case of the massive triangle, the pentagon graph hypersurface \eqref{e:pentagon6d} arises from $a_1=a_2=a_3=a_4=0$ with $z_0=x_1+\cdots +x_5$ and a linear change of variables $z_i=\sum_{r=1}^5c_{i,r}x_r$, where the coefficients $c_{i,r}$ and $a_0$ and $a_5$ depend on the kinematic variables. 
The pentagon graph hypersurface in \cref{e:pentagon8d} is obtained by setting 
\begin{equation}
z_0=z_1=z_2=x_1+\cdots+x_5, \quad
a_1=a_2=a_3=a_4=0,  \quad a_0=-c_{3 5} c_{45} a_{5}\, ,
\end{equation}
and expressing the coefficients $c_{r,s}$ of the linear map between the variables $z_r=\sum_{s=1}^5c_{r,s}x_r$ and the edge variables $x_i$ in terms of the kinematic variables $c_{ij}$.
The extreme case in \cref{e:pentagon10d} arises as $a_1=\cdots=a_5=0$, $a_0\neq0$ and $z_0=x_1+\cdots +x_5$, and the other $z_i$ linearly related to the edge variables $x_i$.

In particular, the emergence of these pentagon graph hypersurfaces can be viewed as toric
degeneration of  smooth quintic threefold, and hence fits naturally into the
Batyrev construction~\cite{Batyrev:1993oya}.  As a consequence, the pentagon integral computes
a period of a singular Calabi--Yau threefold that is birational to a
toric degeneration of the quintic.

\subsection{Geometry of reflexive Feynman integrals}

Focusing on the combinatorics and geometry of graph hypersurfaces, the examples above demonstrate that multiloop sunset and one-loop Feynman integrals admit a natural interpretation in terms of smooth anticanonical Calabi--Yau varieties arising from reflexive polytopes via toric methods.
They illustrate a general phenomenon: even when the graph
hypersurface is singular, the reflexive polytope underlying the Symanzik
polynomials still determines a Calabi--Yau structure via Batyrev's
mirror symmetry.  The singularities of the hypersurface correspond to
physical boundary phenomena such as threshold singularities (the
discriminant vanishes at the location of threshold, see e.g. this
connection for two-loop graphs~\cite{Doran:2023yzu}), and mirror
symmetry provides a geometric language to interpret these analytic
features.
It could be interesting to further study the fibration structure on the mirror side and the connection with  torically induced Tyurin degenerations studied by Doran, Harder, and Thompson~\cite{Doran:2016uea}.

This perspective provides a unified framework in which the analytic structures of Feynman integrals---whether rational, polylogarithmic, elliptic, or of Calabi--Yau type---can be understood directly from the combinatorics of the Newton polytope. In this way, mirror symmetry emerges not as an exotic external structure imposed from outside, but as an intrinsic geometric feature of perturbative quantum field theory.

\section{Conclusions}\label{sec:conclusion}

In this work, we have analyzed Fano and reflexive polytopes arising as Newton polytopes of finite and quasi-finite Feynman integrals when the polytope has a single interior point.
We have found that the number of Fano and reflexive polytopes
is very small, which is a remarkable fact considering that the number of three-  and four-dimensional reflexive polytopes is huge~\cite{Kreuzer:2000xy}. We found that in generic kinematics the space of Fano and reflexive polytopes exhibits various families of Feynman graphs including, the well-known multiloop sunset, the $\nprops$-gon, among others. 

Fano and reflexive polytopes are naturally associated with Calabi--Yau varieties, and
our analysis shows that a subset of (quasi-)finite Feynman integrals are
naturally  Calabi--Yau period integrals. Their appearance in this context may provide a geometric organizing principle for the space of such Feynman integrals, allowing one to associate to each integral a toric Calabi--Yau variety or Fano variety that captures its singularity and period structure.
From a mathematical viewpoint, identifying the toric geometry associated with a Feynman graph offers a pathway to understanding the Hodge-theoretic~\cite{Bloch:2013tra,Bloch:2014qca,Bloch:2016izu,Doran:2023yzu}, motivic, and combinatorial structures underlying amplitudes~\cite{Brown:2015fyf, Schnetz:2016fhy, Panzer:2018tiv} and the role of mirror symmetry in evaluating Feynman integrals~\cite{Bloch:2014qca,Bloch:2016izu}.

In multiloop calculations, the standard approach relies on
large families of master integrals, many of which possess overlapping
ultraviolet and infrared divergences and require elaborate subtraction procedures or
dimensional regularization expansions before their finite
contributions can be extracted.  
A promising direction for improving the efficiency of amplitude
computations is the systematic identification of quasi-finite Feynman
integrals. The reason is that finite integrals are easier to evaluate numerically and, in addition to standard techniques, novel sampling methods are applicable to them~\cite{Borinsky:2023jdv, Borinsky:2020rqs}. Our examples show that some of them are also easy to evaluate analytically.

Finally, our analysis gives a systematic way of identifying a rather small number of quasi-finite master integrals that share a common geometrical property: having an associated  Fano variety. 
This might offer another avenue to constructing an integral basis for amplitudes based on algebraic geometry, which would  make the extraction of the finite part of the amplitude more direct.

\section*{Acknowledgements}
We thank Charles Doran for discussions.
The work of PV was funded by the Agence Nationale de la Recherche
(ANR) under the grant Observables (ANR-24-CE31-7996). The research of LDLC was supported by 
the European Research Council under grant ERC--AdG--885414.
The work of PPN was supported by the
European Research Council (ERC) under the European Union's Horizon Europe
research and innovation program grant agreement 101078449 (ERC Starting Grant
MultiScaleAmp).
Views and opinions expressed are however those of the authors only and do not necessarily reflect those of the European Union or the European Research Council Executive Agency. Neither the European Union nor the granting authority can be held responsible for them.

\appendix

\section{Graph polynomials of two-loop graphs}
\label{app:TwoloopGraph}

In this \namecref{app:TwoloopGraph}, we give the generic form of the graph polynomials
of two-loop integrals. The graphs are labeled by
                a triplet of integers $(a,b,c)$ indicating the number of
                external legs attached to each line of the
                skeleton of the graph following the notations of ref.~\cite{Doran:2023yzu}, see \cref{fig:a1cgraphsBis}. 

 	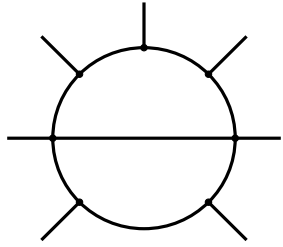
\begin{figure}[h]
			\centering
			\begin{tikzpicture}[scale=0.6]
				\filldraw [color = black, fill=none, very thick] (0,0) circle (2cm);
				\draw [black,very thick] (-2,0) to (2,0);
				\filldraw [black] (2,0) circle (2pt);
				\filldraw [black] (-2,0) circle (2pt);
				\filldraw [black] (0,2) circle (2pt);
				\filldraw [black] (1.414,1.414) circle (2pt);
				\filldraw [black] (-1.414,1.414) circle (2pt);
				\filldraw [black] (1.414,-1.414) circle (2pt);
				\filldraw [black] (-1.414,-1.414) circle (2pt);
				\draw [black,very thick] (-2,0) to (-3,0);
				\draw [black,very thick] (2,0) to (3,0);
				\draw [black,very thick] (0,2) to (0,3);
				\draw [black,very thick] (1.414,1.414) to (2.25,2.25);
				\draw [black,very thick] (-1.414,1.414) to (-2.25,2.25);
				\draw [black,very thick] (1.414,-1.414) to (2.25,-2.25);
				\draw [black,very thick] (-1.414,-1.414) to (-2.25,-2.25);
			\end{tikzpicture}
			\caption{A two-loop graph of type $(a,b,c)$ with $a=4$, $b=1$ and $c=3$.}\label{fig:a1cgraphsBis}
		\end{figure}
		
		The Symanzik polynomials for two-loop graphs are given by

		\begin{align}\label{eq:UVFtwolop}
			\cU_{(a,b,c)} & = \left(\sum_{i=1}^ax_i \right) \left( \sum_{i=1}^b
			y_i\right) + \left(\sum_{i=1}^ax_i \right)  \left(
			\sum_{i=1}^c z_i \right) +  \left( \sum_{i=1}^b
			y_i\right)  \left( \sum_{i=1}^c z_i \right) ,\cr
			\cL_{(a,b,c)}&=\sum_{i=1}^a m_{i}^2 x_i + \sum_{i=1}^b m^2_{i+a} y_i
			+\sum_{i=1}^c  m^2_{a+c+i} z_i,\cr
			\cV_{(a,b,c)} & =  \left( \sum_{i=1}^b
			y_i+ \sum_{i=1}^c z_i \right)\sum_{1\leq
                                        i<j\leq a} c^{x}_{ij} x_ix_j
                                        + \left( \sum_{i=1}^a
                                        x_i                                             + \sum_{i=1}^c z_i \right)\sum_{1\leq
                                        i<j\leq b} c^y_{ij} y_iy_j\cr
                                        &+ \left( \sum_{i=1}^a
			x_i+ \sum_{i=1}^b y_i \right)\sum_{1\leq
				i<j\leq c} c^z_{ij} z_iz_j
			 +\sum_{i=1}^a\sum_{j=1}^b\sum_{k=1}^c c^{xyz}_{ijk} x_iy_jz_k,\cr                  
			\ \cF_{(a,b,c)} & = \cU_{(a,b,c)}\cL_{(a,b,c)} - (\cV^2_{(a,b,c)}+\cV^3_{(a,b,c)}) \,,
		\end{align}
		where $m_i$ are the internal masses and the coefficients $c^x$,
		$c^y$, $c^z$ and $c^{xyz}$ are functions of the scalar products of the
		external legs.

\subsection{The Symanzik polynomial of the diamond circle  graph}\label{app:Udiamondcircle}
We give the first Symanzik polynomial for the diamond circle graph in \cref{fig:diamondcircle}:
 \begin{multline}\label{e:Udiamondcircle}
    \cU_{\rm diamond~circle}=\feynpar_{1} \feynpar_{2} \feynpar_{4} \feynpar_{6} \feynpar_{7}+\feynpar_{1} \feynpar_{2} \feynpar_{4}
    \feynpar_{6} \feynpar_{8}+\feynpar_{1} \feynpar_{2} \feynpar_{4} \feynpar_{6} \feynpar_{9}+\feynpar_{1} \feynpar_{2} \feynpar_{4} \feynpar_{6}
    \feynpar_{10}\cr
    +\feynpar_{1} \feynpar_{2} \feynpar_{4} \feynpar_{7} \feynpar_{8}+\feynpar_{1} \feynpar_{2} \feynpar_{4} \feynpar_{7}
    \feynpar_{10}+\feynpar_{1} \feynpar_{2} \feynpar_{4} \feynpar_{8} \feynpar_{9}+\feynpar_{1} \feynpar_{2} \feynpar_{4} \feynpar_{9}
    \feynpar_{10}+\feynpar_{1} \feynpar_{2} \feynpar_{5} \feynpar_{6} \feynpar_{7}+\feynpar_{1} \feynpar_{2} \feynpar_{5} \feynpar_{6}
    \feynpar_{8}\cr
    +\feynpar_{1} \feynpar_{2} \feynpar_{5} \feynpar_{6} \feynpar_{9}+\feynpar_{1} \feynpar_{2} \feynpar_{5} \feynpar_{6}
    \feynpar_{10}+\feynpar_{1} \feynpar_{2} \feynpar_{5} \feynpar_{7} \feynpar_{8}+\feynpar_{1} \feynpar_{2} \feynpar_{5} \feynpar_{7}
    \feynpar_{10}+\feynpar_{1} \feynpar_{2} \feynpar_{5} \feynpar_{8} \feynpar_{9}+\feynpar_{1} \feynpar_{2} \feynpar_{5} \feynpar_{9}
    \feynpar_{10}\cr
    +\feynpar_{1} \feynpar_{2} \feynpar_{6} \feynpar_{7} \feynpar_{9}+\feynpar_{1} \feynpar_{2} \feynpar_{6} \feynpar_{7}
    \feynpar_{10}+\feynpar_{1} \feynpar_{2} \feynpar_{6} \feynpar_{8} \feynpar_{9}+\feynpar_{1} \feynpar_{2} \feynpar_{6} \feynpar_{8}
    \feynpar_{10}+\feynpar_{1} \feynpar_{2} \feynpar_{7} \feynpar_{8} \feynpar_{9}+\feynpar_{1} \feynpar_{2} \feynpar_{7} \feynpar_{8}
    \feynpar_{10}\cr
    +\feynpar_{1} \feynpar_{2} \feynpar_{7} \feynpar_{9} \feynpar_{10}+\feynpar_{1} \feynpar_{2} \feynpar_{8} \feynpar_{9}
    \feynpar_{10}+\feynpar_{1} \feynpar_{3} \feynpar_{4} \feynpar_{6} \feynpar_{7}+\feynpar_{1} \feynpar_{3} \feynpar_{4} \feynpar_{6}
    \feynpar_{8}+\feynpar_{1} \feynpar_{3} \feynpar_{4} \feynpar_{6} \feynpar_{9}+\feynpar_{1} \feynpar_{3} \feynpar_{4} \feynpar_{6}
    \feynpar_{10}\cr
    +\feynpar_{1} \feynpar_{3} \feynpar_{4} \feynpar_{7} \feynpar_{8}+\feynpar_{1} \feynpar_{3} \feynpar_{4} \feynpar_{7}
    \feynpar_{10}+\feynpar_{1} \feynpar_{3} \feynpar_{4} \feynpar_{8} \feynpar_{9}+\feynpar_{1} \feynpar_{3} \feynpar_{4} \feynpar_{9}
    \feynpar_{10}+\feynpar_{1} \feynpar_{3} \feynpar_{5} \feynpar_{6} \feynpar_{7}+\feynpar_{1} \feynpar_{3} \feynpar_{5} \feynpar_{6}
    \feynpar_{8}\cr
    +\feynpar_{1} \feynpar_{3} \feynpar_{5} \feynpar_{6} \feynpar_{9}+\feynpar_{1} \feynpar_{3} \feynpar_{5} \feynpar_{6}
    \feynpar_{10}+\feynpar_{1} \feynpar_{3} \feynpar_{5} \feynpar_{7} \feynpar_{8}+\feynpar_{1} \feynpar_{3} \feynpar_{5} \feynpar_{7}
    \feynpar_{10}+\feynpar_{1} \feynpar_{3} \feynpar_{5} \feynpar_{8} \feynpar_{9}+\feynpar_{1} \feynpar_{3} \feynpar_{5} \feynpar_{9}
    \feynpar_{10}\cr
    +\feynpar_{1} \feynpar_{3} \feynpar_{6} \feynpar_{7} \feynpar_{9}+\feynpar_{1} \feynpar_{3} \feynpar_{6} \feynpar_{7}
    \feynpar_{10}+\feynpar_{1} \feynpar_{3} \feynpar_{6} \feynpar_{8} \feynpar_{9}+\feynpar_{1} \feynpar_{3} \feynpar_{6} \feynpar_{8}
    \feynpar_{10}+\feynpar_{1} \feynpar_{3} \feynpar_{7} \feynpar_{8} \feynpar_{9}+\feynpar_{1} \feynpar_{3} \feynpar_{7} \feynpar_{8}
    \feynpar_{10}\cr
    +\feynpar_{1} \feynpar_{3} \feynpar_{7} \feynpar_{9} \feynpar_{10}+\feynpar_{1} \feynpar_{3} \feynpar_{8} \feynpar_{9}
    \feynpar_{10}+\feynpar_{1} \feynpar_{4} \feynpar_{6} \feynpar_{7} \feynpar_{9}+\feynpar_{1} \feynpar_{4} \feynpar_{6} \feynpar_{7}
    \feynpar_{10}+\feynpar_{1} \feynpar_{4} \feynpar_{6} \feynpar_{8} \feynpar_{9}+\feynpar_{1} \feynpar_{4} \feynpar_{6} \feynpar_{8}
    \feynpar_{10}\cr
    +\feynpar_{1} \feynpar_{4} \feynpar_{7} \feynpar_{8} \feynpar_{9}+\feynpar_{1} \feynpar_{4} \feynpar_{7} \feynpar_{8}
    \feynpar_{10}+\feynpar_{1} \feynpar_{4} \feynpar_{7} \feynpar_{9} \feynpar_{10}+\feynpar_{1} \feynpar_{4} \feynpar_{8} \feynpar_{9}
    \feynpar_{10}+\feynpar_{1} \feynpar_{5} \feynpar_{6} \feynpar_{7} \feynpar_{9}+\feynpar_{1} \feynpar_{5} \feynpar_{6} \feynpar_{7}
    \feynpar_{10}\cr
    +\feynpar_{1} \feynpar_{5} \feynpar_{6} \feynpar_{8} \feynpar_{9}+\feynpar_{1} \feynpar_{5} \feynpar_{6} \feynpar_{8}
    \feynpar_{10}+\feynpar_{1} \feynpar_{5} \feynpar_{7} \feynpar_{8} \feynpar_{9}+\feynpar_{1} \feynpar_{5} \feynpar_{7} \feynpar_{8}
    \feynpar_{10}+\feynpar_{1} \feynpar_{5} \feynpar_{7} \feynpar_{9} \feynpar_{10}+\feynpar_{1} \feynpar_{5} \feynpar_{8} \feynpar_{9}
    \feynpar_{10}\cr
    +\feynpar_{2} \feynpar_{3} \feynpar_{4} \feynpar_{6} \feynpar_{7}+\feynpar_{2} \feynpar_{3} \feynpar_{4} \feynpar_{6} \feynpar_{8}+\feynpar_{2}
    \feynpar_{3} \feynpar_{4} \feynpar_{6} \feynpar_{9}+\feynpar_{2} \feynpar_{3} \feynpar_{4} \feynpar_{6} \feynpar_{10}+\feynpar_{2} \feynpar_{3}
    \feynpar_{4} \feynpar_{7} \feynpar_{8}+\feynpar_{2} \feynpar_{3} \feynpar_{4} \feynpar_{7} \feynpar_{10}\cr
    +\feynpar_{2} \feynpar_{3} \feynpar_{4} \feynpar_{8} \feynpar_{9}+\feynpar_{2} \feynpar_{3} \feynpar_{4} \feynpar_{9}
    \feynpar_{10}+\feynpar_{2} \feynpar_{3} \feynpar_{5} \feynpar_{6} \feynpar_{7}+\feynpar_{2} \feynpar_{3} \feynpar_{5} \feynpar_{6}
    \feynpar_{8}+\feynpar_{2} \feynpar_{3} \feynpar_{5} \feynpar_{6} \feynpar_{9}+\feynpar_{2} \feynpar_{3} \feynpar_{5} \feynpar_{6}
    \feynpar_{10}\cr
    +\feynpar_{2} \feynpar_{3} \feynpar_{5} \feynpar_{7} \feynpar_{8}+\feynpar_{2} \feynpar_{3} \feynpar_{5} \feynpar_{7}
    \feynpar_{10}+\feynpar_{2} \feynpar_{3} \feynpar_{5} \feynpar_{8} \feynpar_{9}+\feynpar_{2} \feynpar_{3} \feynpar_{5} \feynpar_{9}
    \feynpar_{10}+\feynpar_{2} \feynpar_{3} \feynpar_{6} \feynpar_{7} \feynpar_{9}+\feynpar_{2} \feynpar_{3} \feynpar_{6} \feynpar_{7}
    \feynpar_{10}\cr
    +\feynpar_{2} \feynpar_{3} \feynpar_{6} \feynpar_{8} \feynpar_{9}+\feynpar_{2} \feynpar_{3} \feynpar_{6} \feynpar_{8}
    \feynpar_{10}+\feynpar_{2} \feynpar_{3} \feynpar_{7} \feynpar_{8} \feynpar_{9}+\feynpar_{2} \feynpar_{3} \feynpar_{7} \feynpar_{8}
    \feynpar_{10}+\feynpar_{2} \feynpar_{3} \feynpar_{7} \feynpar_{9} \feynpar_{10}+\feynpar_{2} \feynpar_{3} \feynpar_{8} \feynpar_{9}
    \feynpar_{10}\cr
    +\feynpar_{2} \feynpar_{4} \feynpar_{5} \feynpar_{6} \feynpar_{7}+\feynpar_{2} \feynpar_{4} \feynpar_{5} \feynpar_{6} \feynpar_{8}+\feynpar_{2}
    \feynpar_{4} \feynpar_{5} \feynpar_{6} \feynpar_{9}+\feynpar_{2} \feynpar_{4} \feynpar_{5} \feynpar_{6} \feynpar_{10}+\feynpar_{2} \feynpar_{4}
    \feynpar_{5} \feynpar_{7} \feynpar_{8}+\feynpar_{2} \feynpar_{4} \feynpar_{5} \feynpar_{7} \feynpar_{10}\cr
    +\feynpar_{2} \feynpar_{4} \feynpar_{5} \feynpar_{8} \feynpar_{9}+\feynpar_{2} \feynpar_{4} \feynpar_{5} \feynpar_{9}
    \feynpar_{10}+\feynpar_{2} \feynpar_{5} \feynpar_{6} \feynpar_{7} \feynpar_{9}+\feynpar_{2} \feynpar_{5} \feynpar_{6} \feynpar_{7}
    \feynpar_{10}+\feynpar_{2} \feynpar_{5} \feynpar_{6} \feynpar_{8} \feynpar_{9}+\feynpar_{2} \feynpar_{5} \feynpar_{6} \feynpar_{8}
    \feynpar_{10}\cr
    +\feynpar_{2} \feynpar_{5} \feynpar_{7} \feynpar_{8} \feynpar_{9}+\feynpar_{2} \feynpar_{5} \feynpar_{7} \feynpar_{8}
    \feynpar_{10}+\feynpar_{2} \feynpar_{5} \feynpar_{7} \feynpar_{9} \feynpar_{10}+\feynpar_{2} \feynpar_{5} \feynpar_{8} \feynpar_{9}
    \feynpar_{10}+\feynpar_{3} \feynpar_{4} \feynpar_{5} \feynpar_{6} \feynpar_{7}+\feynpar_{3} \feynpar_{4} \feynpar_{5} \feynpar_{6}
    \feynpar_{8}\cr
    +\feynpar_{3} \feynpar_{4} \feynpar_{5} \feynpar_{6} \feynpar_{9}+\feynpar_{3} \feynpar_{4} \feynpar_{5} \feynpar_{6}
    \feynpar_{10}+\feynpar_{3} \feynpar_{4} \feynpar_{5} \feynpar_{7} \feynpar_{8}+\feynpar_{3} \feynpar_{4} \feynpar_{5} \feynpar_{7}
    \feynpar_{10}+\feynpar_{3} \feynpar_{4} \feynpar_{5} \feynpar_{8} \feynpar_{9}+\feynpar_{3} \feynpar_{4} \feynpar_{5} \feynpar_{9}
    \feynpar_{10}\cr
    +\feynpar_{3} \feynpar_{4} \feynpar_{6} \feynpar_{7} \feynpar_{9}+\feynpar_{3} \feynpar_{4} \feynpar_{6} \feynpar_{7}
    \feynpar_{10}+\feynpar_{3} \feynpar_{4} \feynpar_{6} \feynpar_{8} \feynpar_{9}+\feynpar_{3} \feynpar_{4} \feynpar_{6} \feynpar_{8}
    \feynpar_{10}+\feynpar_{3} \feynpar_{4} \feynpar_{7} \feynpar_{8} \feynpar_{9}+\feynpar_{3} \feynpar_{4} \feynpar_{7} \feynpar_{8}
    \feynpar_{10}\cr
    +\feynpar_{3} \feynpar_{4} \feynpar_{7} \feynpar_{9} \feynpar_{10}+\feynpar_{3} \feynpar_{4} \feynpar_{8} \feynpar_{9}
    \feynpar_{10}+\feynpar_{4} \feynpar_{5} \feynpar_{6} \feynpar_{7} \feynpar_{9}+\feynpar_{4} \feynpar_{5} \feynpar_{6} \feynpar_{7}
    \feynpar_{10}+\feynpar_{4} \feynpar_{5} \feynpar_{6} \feynpar_{8} \feynpar_{9}+\feynpar_{4} \feynpar_{5} \feynpar_{6} \feynpar_{8}
    \feynpar_{10}\cr
    +\feynpar_{4} \feynpar_{5} \feynpar_{7} \feynpar_{8} \feynpar_{9}+\feynpar_{4} \feynpar_{5} \feynpar_{7} \feynpar_{8} \feynpar_{10}+\feynpar_{4} \feynpar_{5} \feynpar_{7} \feynpar_{9} \feynpar_{10}+\feynpar_{4} \feynpar_{5} \feynpar_{8} \feynpar_{9} \feynpar_{10}    \,.
  \end{multline}
   
  \section{Bivariate Ehrhart polynomial for massless one-loop integrals}\label{app:oneLoopEhrhart}
  
  In this \namecref{app:oneLoopEhrhart}, we give a proof of the formula for the Ehrhart
  polynomial in \cref{e:EhNgonConj}.
  
  Let \( \nprops \ge 2 \) be an integer. 
For integers \( t_1,t_2 \ge 0 \), consider the Minkowski sum of the standard simplex in \cref{e:standardsimplex-definition} and the second hypersimplex in  \cref{e:secondhypersymplex-definition},
\begin{equation}
\Newton(t_1,t_2) := t_1    \simplexwithpars{\nprops,1} + t_2  \simplexwithpars{\nprops,2}  .
\end{equation}
We will describe \( \Newton(t_1,t_2) \) explicitly and derive its Ehrhart polynomial
\( \ehr(t_1,t_2) = \#(\Newton(t_1,t_2)\cap\mathbb{Z}^\nprops) \) which counts the number of lattice points of $\Newton(t_1,t_2)$.

For any matroids (or polymatroids) with rank functions \(r_1,r_2\),
the Minkowski sum of their base polytopes is the base polytope of the
rank function \(r=r_1+r_2\).
Hence
\begin{equation}
\Newton(t_1,t_2)
= \left\{ x \in \mathbb{R}_{\ge 0}^\nprops \, \middle| \,
   \sum_{i \in S} x_i \le r(S)\ \text{for all } S \subseteq [\nprops],
   \ \sum_{i=1}^\nprops x_i = r([\nprops]) \right\}
\end{equation}
is the base polytope of rank
\begin{equation}
r(S) = t_1 \min(|S|,1) + t_2 \min(|S|,2), \qquad S \subseteq [\nprops],
\end{equation}
which evaluates to
\begin{equation}
r(S) =
\begin{cases}
0, & |S|=0,\\[4pt]
t_1+t_2, & |S|=1,\\[4pt]
t_1+2t_2, & |S|\ge 2.
\end{cases}
\end{equation}
Therefore only the singleton inequalities are non-trivial, and 
we obtain the simple description
\begin{equation}
\Newton(t_1,t_2)
= \left\{ x \in \mathbb{R}_{\ge 0}^\nprops \, \middle| \,
      \sum_{i=1}^\nprops x_i = t_1+2t_2, \; x_i \le t_1+t_2 \ \text{for all } i \right\}.
\end{equation}
Thus the Minkowski sum depends only on the total sum \(t_1+2t_2\)
and the per-coordinate cap \(t_1+t_2\).
The lattice points of \(\Newton(t_1,t_2)\) are the integer solutions of
\begin{equation}
x_1+\cdots+x_\nprops = t_1+2t_2, \qquad 0 \le x_i \le t_1+t_2.
\end{equation}
Without the upper bounds, the number of non-negative integer solutions
is given by the classical stars-and-bars formula
\begin{equation}
\binom{t_1+2t_2+\nprops-1}{\nprops-1}.
\end{equation}
To impose the upper bounds \(x_i \le t_1+t_2\), we apply inclusion--exclusion.
For each subset \(J \subseteq [\nprops]\) of size \(j\),
consider the solutions with \(x_i \ge t_1+t_2+1\) for all \(i\in J\).
Setting \(x_i = y_i + (t_1+t_2+1)\) for \(i\in J\),
we obtain the equation
\begin{equation}
\sum_{i=1}^\nprops y_i = t_1+2t_2 - j(t_1+t_2+1),
\end{equation}
whose non-negative solutions number is
\begin{equation}
\binom{t_1+2t_2 - j(t_1+t_2+1) + \nprops - 1}{\nprops-1}.
\end{equation}
Using inclusion--exclusion, the total count is therefore
\begin{equation}\label{e:Lsum}
\ehr(t_1,t_2)
= \sum_{j=0}^{\lfloor (t_1+2t_2)/(t_1+t_2+1)\rfloor}
    (-1)^j \binom{\nprops}{j}
    \binom{t_1+2t_2 - j(t_1+t_2+1) + \nprops - 1}{\nprops - 1}.
\end{equation}
The upper limit of the sum arises because the binomial coefficient
vanishes when \(t_1+2t_2-j(t_1+t_2+1) < 0\).

For the case of the Minkowski sum arising from the massless one-loop
graph in \cref{e:DeltaOneLoopMassless} we have
$t_1+2t_2=\nu_1+\cdots+\nu_\nprops$ and $t_1+t_2=D/2$, implying that \(t_1+2t_2 < 2(t_1+t_2+1)\) because $t_1\geq0$.
Therefore only the terms \(j=0,1\) contribute to the sum in
\cref{e:Lsum}, so that
\begin{equation}
\ehr(t_1,t_2)
= \binom{t_1+2t_2+\nprops-1}{\nprops-1}
  - N \binom{t_2+\nprops-2}{\nprops-1}.
\end{equation}

\section{Bivariate Ehrhart polynomial for multiloop sunset integrals}
\label{app:sunsetEhrhart}

In this \namecref{app:sunsetEhrhart}, we compute the bivariate Ehrhart polynomial for the
massive multiloop sunset graphs. We review some properties of the
polytopes relevant for the multiloop sunset graphs and their relation to the
permutohedron, and derive the Ehrhart polynomial that we used in the
main text to identify the sunset polytopes with a single interior
point.

\subsection{Sunset polytope and permutohedron}\label{app:sunsetpolytope}

The polytope for the family of
multiloop sunset graphs is given in
\cref{e:SunsetPolytope}. Introducing   the notation $t_1=\nu_1+\cdots+\nu_{\nprops}-{(\nprops-1)D\over2}$ and
$t_2={D\over2}$ for the coefficients of the Newton polytope for the
sunset for better readability, the sunset Newton polytope
\begin{equation}
  \label{e:PolySun}
  \Delta_{\circleddash}(t_1,t_2;\nprops)=t_1 \simplexwithpars{\nprops,1} +t_2  \Newton(\cU_{\circleddash}^{\nprops})
\end{equation}
is a combination of the standard simplex polytope in 
\cref{e:standardsimplex-definition}
and the Newton polytope for the first Symanzik polynomial,
\begin{equation}
  \Newton(\cU_{\circleddash}^{\nprops})=\Newton\left(\sum_{i=1}^{\nprops} \prod_{1\leq 
    j\leq \nprops\atop j\neq i} x_j\right) =  \Conv(\mathbf 1-e_i) =\mathbf 1-\simplexwithpars{\nprops,1}\,,
\end{equation}
with $\mathbf 1=(1,\dots,1)$.
We notice that the Minkowski sum $\Delta_{\circleddash}(1,1;\nprops)=\simplexwithpars{\nprops,1}  + \Newton(\cU_{\circleddash}^{\nprops})$ has vertices among the points
\begin{equation}
e_i + (\mathbf{1} - e_j) = \mathbf{1} + e_i - e_j, \quad i, j \in \{1, \ldots, \nprops\}.
\end{equation}
The actual vertices are:
\begin{enumerate}[label=(\roman*)]
    \item for $i = j$: the point $\mathbf{1}$ (with multiplicity, but appears once as a vertex only for $n = 2$);
    \item for $i \neq j$: the points $\mathbf{1} + e_i - e_j$.
\end{enumerate}
For $\nprops \geq 3$, the vertices are precisely $\{\mathbf{1} + e_i - e_j : i \neq j\}$, giving $n(n-1)$ vertices.
This Newton polytope $\Delta_{\circleddash}(1,1;\nprops)$ is invariant under the action of the permutation group~$\mathfrak S_{\nprops}$ on the coordinates.

Let us introduce the permutohedron $\Pi_{\nprops-1}$ defined by~\cite{postnikov2005}
\begin{equation}
\Pi_{\nprops-1} = \Conv \set{(\sigma(1), \sigma(2), \ldots, \sigma(\nprops)) \, | \, \sigma
\in \mathfrak S_{\nprops}}.
\end{equation}
The sunset polytope $\Delta_{\circleddash}(1,1;\nprops)$ translated to be centered at the origin, is homothetic to the permutohedron. Specifically,
\begin{equation}
\Delta_{\circleddash}(1,1;\nprops) - \mathbf{1} = \Conv \set{e_i - e_j \, | \, i \neq j},
\end{equation}
which is the convex hull of the root system of type $A_{\nprops-1}$. This
was noticed in ref.~\cite{verrill1996root}.
The root system of type $A_{n-1}$ consists of vectors
\begin{equation}
\Phi_{A_{\nprops-1}} = \set{e_i - e_j \, | \, 1 \leq i \neq j \leq \nprops} \subset
\mathbb R^\nprops.
\end{equation}
The simple roots are $\alpha_i = e_i - e_{i+1}$ for $i = 1, \ldots, \nprops-1$.
The root polytope is the convex hull of the roots
\begin{equation}\label{e:RootAnpolytope}
\mathcal{R}_{A_{\nprops-1}} = \Conv(\Phi_{A_{\nprops-1}}) = \Conv \set{e_i - e_j \, | \, i \neq j}.
\end{equation}
We then have the relation to the sunset graph polytope
\begin{equation}
\Delta_{\circleddash}(1,1;\nprops) = \mathbf{1} + \mathcal{R}_{A_{\nprops-1}}.
\end{equation}

We remark that the root polytope $\mathcal{R}_{A_{\nprops-1}}$ is distinct from but closely related to the permutohedron. The permutohedron can be expressed as a Minkowski sum of line segments along the positive roots, while the root polytope is the convex hull of all roots.

\subsection{Ehrhart polynomials}

Before computing the bivariate Ehrhart polynomial, we compute the univariate Ehrhart polynomial for each polytope in the Minkowski sum.

The dilated simplex $t\simplexwithpars{\nprops,1}$ consists of points $(x_1, \ldots, x_\nprops)$ with $x_i \geq 0$ and $\sum x_i = t$. The number of non-negative integer solutions is $\binom{t+n-1}{n-1}$ by the stars-and-bars formula.
Therefore, the Ehrhart polynomial is~\cite{Beck:2015}
\begin{equation}
\ehr_{\simplexwithpars{\nprops,1}}(t) =\#\left(t\simplexwithpars{\nprops,1} \cap \mathbb Z^\nprops\right) = \binom{t + \nprops - 1}{\nprops - 1} = \frac{(t+1)(t+2)\cdots(t+\nprops-1)}{1 \cdot 2\cdots (\nprops-1)}.
\end{equation}

\medskip

The dilation of the second polytope, $t \Newton(\cU_{\circleddash}^{\nprops-1}) = t \left( \mathbf{1} - \simplexwithpars{\nprops,1} \right)$, is obtained from $t \simplexwithpars{\nprops,1}$ by an integer shift followed by a reflection. Since both operations preserve the number of lattice points, we have
\begin{equation}
\ehr_{\Newton(\cU_{\circleddash}^{\nprops-1})}(t) = \#\left(t
  \Newton(\cU_{\circleddash}^{\nprops-1}) \cap \mathbb Z^\nprops\right) = \ehr_{\simplexwithpars{\nprops,1}}(t).
\end{equation}

\medskip

The bivariate Ehrhart polynomial for the Minkowski sum
$\Delta_{\circleddash}(t_1,t_2; \nprops)$ in \cref{e:PolySun} is defined by
\begin{equation}
    \ehr_{\circleddash}(t_1, t_2) :=\# \left((t_1
      \simplexwithpars{\nprops,1} + t_2 ( \mathbf{1} - \simplexwithpars{\nprops,1})) \cap \mathbb Z^\nprops\right)\,.
  \end{equation}
  This polynomial satisfies the following properties:
  \begin{enumerate}[label=(\roman*)]
  \item it has special values for $t_1=0$ or $t_2=0$,
    \begin{equation}\label{e:EhrSunsetProp1}
      \ehr_{\circleddash}(t, 0) = \ehr_{\circleddash}(0, t) = \binom{t + \nprops - 1}{\nprops - 1};
    \end{equation}
  \item it has total degree $\nprops-1$ in $t_1$ and $t_2$;
  \item it is a symmetric function of $t_1$ and $t_2$,
    \begin{equation}
      \ehr_{\circleddash}(t_1, t_2)=\ehr_{\circleddash}(t_2, t_1) \,,
    \end{equation}
    so that it can be written in symmetric form
    \begin{equation}\label{e:Ehrsp}
  \ehr_{\circleddash}(t_1,t_2,\nprops) = \sum_{ r_1,r_2\geq 0\atop r_1+2r_2\leq \nprops-1}e^\nprops_{r_1, r_2} (t_1 + t_2)^{r_1} (t_1 t_2)^{r_2}\,.
\end{equation}
\end{enumerate}
Because the polytope  $ \Delta_{\circleddash}(t_1,t_2;\nprops)$ in  \cref{e:PolySun}
  is a generalized permutohedron, we can use
Postnikov's result on mixed volumes and his submodular function
formalism~\cite{postnikov2005} to determine the weight coefficients $e_{r_1,r_2}^\nprops$.

\medskip

We introduce the polynomial
\begin{equation}
  P(t_1,t_2,n):=  \sum_{0\leq i,j\leq n-1\atop i+j\neq
    n-1}(-1)^{-i-j+n-1} \binom{i+j}{i} \binom{i+t_1}{i} \binom{j+t_2}{j} \,.
\end{equation}
The Ehrhart polynomial for the multiloop sunset is given by
\begin{equation}\label{e:EhrSunsetResult}
    \ehr_{\circleddash}(t_1,t_2,n)=P(t_1,t_2,n)+\sum_{r=0}^{n-1} c(r,n) P(t_1,t_2,n)\,,
\end{equation}
where the coefficients $c(r,n)$ read~\cite[A097808]{oeis}
\begin{equation}
 c(r,n)=(-1)^n \textrm{coeff}_{x^{n-1}}   \left({2 x+1\over (1+x)^2}  \left(\frac{x}{x+1}\right)^{n-1-r}\right),
\end{equation}
and $\textrm{coeff}_{x^{r}}(f(x))$ means the coefficient of
$x^r$ in the series expansion of $f(x)$ around $x=0$.

\medskip

For fixed integer values of $\tau_1$ and $\tau_2$
one can compute the Ehrhart polynomial for the polytope $P=\tau_1\simplexwithpars{\nprops,1}+\tau_2 ( \mathbf{1} - \simplexwithpars{\nprops,1})$,
\begin{equation}
\ehr_P(t,\nprops)=  \#\left(t\Delta_{\circleddash}(\tau_1,\tau_2;\nprops) \cap \mathbb Z^n\right) \,,
\end{equation}
with {\tt polymake}~\cite{polymake:FPSAC_2009} and check that this
agrees with the bivariate Ehrhart polynomial of \cref{e:EhrSunsetResult} evaluated at $t_1=t \tau_1$ and $t_2=t\tau_2$.
We give the expressions of the bivariate Ehrhart polynomial for the multiloop sunset graphs up to ten edges in \cref{tab:EhrSunset}.

\begin{table}\begin{tblr}{
  colspec={|Q||X|},
  row{odd}={bg=lightgray},  
  row{1}={bg=brown,fg=white}}
  \hline
  $n$ edges & $\ehr(t_1,t_2;n)$ in \cref{e:EhrSunsetResult} with $s=t_1+t_2$ and $p=t_1t_2$\\
  \hline
  2&$1+s$\\[2ex]
  3&$1+p+\frac{3 s}{2}+\frac{s^2}{2}$\\[2ex]
  4&$1+p+\frac{11 s}{6}+p s+s^2+\frac{s^3}{6}$\\[2ex]
  5&$1+\frac{5 p}{4}+\frac{p^2}{4}+\frac{25 s}{12}+\frac{5 p s}{4}+\frac{35 s^2}{24}+\frac{p
   s^2}{2}+\frac{5 s^3}{12}+\frac{s^4}{24}$\\[2ex]
  6&$1+\frac{5 p}{4}+\frac{p^2}{4}+\frac{137 s}{60}+\frac{11 p s}{6}+\frac{p^2 s}{4}+\frac{15
   s^2}{8}+\frac{3 p s^2}{4}+\frac{17 s^3}{24}+\frac{p
   s^3}{6}+\frac{s^4}{8}+\frac{s^5}{120}$\\[2ex]
   7&$1+\frac{49 p}{36}+\frac{7 p^2}{18}+\frac{p^3}{36}+\frac{49 s}{20}+\frac{49 p
   s}{24}+\frac{7 p^2 s}{24}+\frac{203 s^2}{90}+\frac{91 p s^2}{72}+\frac{p^2
      s^2}{8}+\frac{49 s^3}{48}+\frac{7 p s^3}{24}
  +\frac{35 s^4}{144}+\frac{p
   s^4}{24}+\frac{7 s^5}{240}+\frac{s^6}{720}$\\[2ex]
  8&$1+\frac{49 p}{36}+\frac{7 p^2}{18}+\frac{p^3}{36}+\frac{363 s}{140}+\frac{877 p
   s}{360}+\frac{37 p^2 s}{72}+\frac{p^3 s}{36}+\frac{469 s^2}{180}+\frac{14 p
   s^2}{9}+\frac{p^2 s^2}{6}+\frac{967 s^3}{720}+\frac{5 p s^3}{9}+\frac{p^2
     s^3}{24}+\frac{7 s^4}{18}
  +\frac{p s^4}{12}+\frac{23 s^5}{360}+\frac{p
   s^5}{120}+\frac{s^6}{180}+\frac{s^7}{5040}$\\[2ex]
  9&$1+\frac{205 p}{144}+\frac{91 p^2}{192}+\frac{5 p^3}{96}+\frac{p^4}{576}+\frac{761
   s}{280}+\frac{417 p s}{160}+\frac{9 p^2 s}{16}+\frac{p^3 s}{32}+\frac{29531
   s^2}{10080}+\frac{971 p s^2}{480}+\frac{21 p^2 s^2}{64} +\frac{p^3 s^2}{72}+\frac{267
   s^3}{160}+\frac{3 p s^3}{4}+\frac{p^2 s^3}{16}+\frac{1069 s^4}{1920}+\frac{17 p
     s^4}{96}+\frac{p^2 s^4}{96} +\frac{9 s^5}{80}+\frac{3 p s^5}{160}+\frac{13
   s^6}{960}+\frac{p s^6}{720}+\frac{s^7}{1120}+\frac{s^8}{40320}$\\[2ex]
  \hline
10&$1+\frac{205 p}{144}+\frac{91 p^2}{192}+\frac{5 p^3}{96}+\frac{p^4}{576}+\frac{7129 s}{2520}+\frac{17531 p s}{6048}+\frac{6427 p^2 s}{8640}+\frac{7 p^3
   s}{108}+\frac{p^4 s}{576}+\frac{6515 s^2}{2016}+\frac{37 p s^2}{16}+\frac{25 p^2 s^2}{64}+\frac{5 p^3 s^2}{288}+\frac{4523 s^3}{2268}+\frac{4567 p
   s^3}{4320}+\frac{235 p^2 s^3}{1728}+\frac{p^3 s^3}{216}+\frac{95 s^4}{128}+\frac{25 p s^4}{96}+\frac{5 p^2 s^4}{288}+\frac{3013 s^5}{17280}+\frac{19 p
   s^5}{432}+\frac{p^2 s^5}{480}+\frac{5 s^6}{192}+\frac{p s^6}{288}+\frac{29 s^7}{12096}+\frac{p s^7}{5040}+\frac{s^8}{8064}+\frac{s^9}{362880}$\\[2ex]
   \hline
\end{tblr}
\caption{Bivariate Ehrhart polynomials in \cref{e:EhrSunsetResult} for $2\leq \nprops\leq 10$.}\label{tab:EhrSunset}
\end{table}

    \section{Tables of Fano graphs}
        \label{app:fanotable}

In the repository~\cite{reflexivefano}, we give the full list of representative graphs with the exponents that lead to Fano or reflexive polytopes.
They are given in \texttt{Mathematica} files named \texttt{FanoCases.m}, \texttt{ReflexiveCases.m} and \texttt{FanoNotReflexiveCases.m} for the list of polytopes with one interior point, the list of reflexive polytopes and the list of Fano polytopes that are not reflexive, respectively. Fano polytopes are accessed with the following command 
$$ \texttt{FanoCases} [\nprops] = \{ \{\texttt{Graph Association}, \mathtt{powers} \}, \{\texttt{Graph Association},  \mathtt{powers}  \}, \dots \}$$
for $\nprops=2,\dots, 10$.  The same format is used for reflexive cases, and the non-reflexive cases.
 The dimension and powers of the propagators can be obtained from the pair $(n_\cU, n_\cF)$ as they appear in the definition of the integral in eq.~\eqref{parametric-integral-projective}.

In addition, in \cref{Fanoupto9,Fanoupto10}  we give the list of graphs that lead to Fano polytopes  where both Symanzik polynomials appear. These are arguably the most relevant cases in Physics. The sunset and $N$-gons are not drawn since they were discussed at length in \cref{sec:reflexivefano,sec:dimloopscan}.  The graphs were drawn using the \texttt{Loopedia} tool~\emph{redraw} \cite{Bogner:2017xhp}.    

\begin{table}[htb!]
	\centering
	\begin{tabular}{|c||c||c|}
		\hline  
		\input{./drawings/tikz-drawings/graph5-6.tex} &  \input{./drawings/tikz-drawings/graph8-13.tex}&\input{./drawings/tikz-drawings/graph9-1534.tex}\\
		\smallNickel {5,\{2,1\}:(1,1)} &        \smallNickel{8,$\set{4,1}$:(1,2)}   & \smallNickel{9,$\set{14,1}$:(3,1)}\\
		\hline
		\input{./drawings/tikz-drawings/graph7-10.tex}&   \input{./drawings/tikz-drawings/graph9-16.tex} & \input{./drawings/tikz-drawings/graph9-1497.tex}\\ 
		\smallNickel{7, $\set{2,1}$:(2,1)} & \smallNickel{9,$\set{2,1}$:(3,1)} & \smallNickel{9,$\set{15,1}$:(1,1)} 
		\\		\hline
		\input{./drawings/tikz-drawings/graph7-11.tex}&  \input{./drawings/tikz-drawings/graph9-17.tex}& \input{./drawings/tikz-drawings/graph9-1488.tex}\\	
		\smallNickel{7,$\set{3,1}$:(2,1)} & \smallNickel{9,$\set{3,1}$:(3,1)} & \smallNickel{9,$\set{18,1}$:(1,1)}
		\\	\hline
		\input{./drawings/tikz-drawings/graph7-76.tex}& \input{./drawings/tikz-drawings/graph9-18.tex}& \input{./drawings/tikz-drawings/graph9-1493.tex} \\	
		\smallNickel{7,$\set{4,1}$:(1,1)}	& \smallNickel{9,$\set{4,1}$:(3,1)} & \smallNickel{9,$\set{19,1}$:(1,1)}
		\\		\hline
		\input{./drawings/tikz-drawings/graph7-71.tex}& \input{./drawings/tikz-drawings/graph9-1542.tex} & \input{./drawings/tikz-drawings/graph9-1507.tex} \\		
		\smallNickel{7,$\set{5,1}$:(1,1)} & \smallNickel{9,$\set{10,1}$:(1,1)} & \smallNickel{9,$\set{20,1}$:(1,1)}
		\\\hline
		\input{./drawings/tikz-drawings/graph7-68.tex}&  \input{./drawings/tikz-drawings/graph9-1543.tex} & \input{./drawings/tikz-drawings/graph9-1457.tex} \\	
		\smallNickel{7,$\set{7,1}$:(1,1)}	&	 \smallNickel{9,$\set{11,1}$:(1,1)} & \smallNickel{9,$\set{22,1}$:(1,1)}
		\\	\hline
		\input{./drawings/tikz-drawings/graph8-11.tex}& \input{./drawings/tikz-drawings/graph9-1527.tex} & \input{./drawings/tikz-drawings/graph9-1483.tex} \\
		\smallNickel{8,$\set{2,1}$:(1,2)}& \smallNickel{9,$\set{12,1}$:(1,1)}	& \smallNickel{9,$\set{23,1}$:(1,1)}
		\\		\hline 
		\input{./drawings/tikz-drawings/graph8-12.tex}&  \input{./drawings/tikz-drawings/graph9-1529.tex} &\input{./drawings/tikz-drawings/graph9-1451.tex} \\	
		\smallNickel{8,$\set{3,1}$:(1,2)} & \smallNickel{9,$\set{13,1}$:(1,1)} & \smallNickel{9,$\set{25,1}$:(1,1)}
		\\				\hline 
	\end{tabular}
	\caption{Fano graph representatives labeled by \smallNickel{edges, position:$(n_\cU, n_\cF )$}. The position refers to the list in the attached \texttt{Mathematica} file. We give only exponents $(n_\cU, n_\cF )$ with $n_\cU,\ne0$, $n_\cF,\ne0$ in generic kinematics up to 9 edges.  }
		\label{Fanoupto9}
\end{table}

\begin{table}[htb!]
	\centering
	\begin{tabular}{|c||c||c|}
		\hline 
		\input{./drawings/tikz-drawings/graph10-17.tex}&\input{./drawings/tikz-drawings/graph10-574.tex} &   \input{./drawings/tikz-drawings/graph10-571.tex}  \\
		\smallNickel{10,$\set{2,1}$:(2,2)}&  \smallNickel{10,$\set{11,1}$:(2,1)}	&	\smallNickel{10,$\set{10,1}$:(2,1)}
		\\		\hline
		\input{./drawings/tikz-drawings/graph10-18.tex} & 	\input{./drawings/tikz-drawings/graph10-577.tex} & \input{./drawings/tikz-drawings/graph10-570.tex}  \\
		\smallNickel{10,$\set{3,1}$:(2,2)}  &\smallNickel{10,$\set{12,1}$:(2,1)} & \smallNickel{10,$\set{9,1}$:(2,1)}
		\\\hline
		\input{./drawings/tikz-drawings/graph10-19.tex} & 	\input{./drawings/tikz-drawings/graph10-582.tex}& \input{./drawings/tikz-drawings/graph10-529.tex}  
		\\
		\smallNickel{10,$\set{4,1}$:(2,2)} & \smallNickel{10,$\set{2,1}$:(13,1)} & \smallNickel{10,$\set{18,1}$:(2,1)}
		\\\hline
		\input{./drawings/tikz-drawings/graph10-20.tex} &	\input{./drawings/tikz-drawings/graph10-584.tex}& \input{./drawings/tikz-drawings/graph10-567.tex}
		\\ 
		\smallNickel{10,$\set{5,1}$:(2,2)}  &\smallNickel{10,$\set{14,1}$:(2,1)} & \smallNickel{10,$\set{8,1}$:(2,1)}
		\\\hline
		\input{./drawings/tikz-drawings/graph10-564.tex} &	\input{./drawings/tikz-drawings/graph10-538.tex}& 	\input{./drawings/tikz-drawings/graph10-528.tex}\\
		\smallNickel{10,$\set{6,1}$:(2,1)} &\smallNickel{10,$\set{15,1}$:(2,1)} & \smallNickel{10,$\set{17,1}$:(2,1)}
		\\\hline
		\input{./drawings/tikz-drawings/graph10-566.tex} &	\input{./drawings/tikz-drawings/graph10-541.tex}& \\
		\smallNickel{10,$\set{7,1}$:(2,1)} &\smallNickel{10,$\set{16,1}$:(2,1)} &   
		\\\hline
	\end{tabular}
	\caption{
    Fano graph representatives labeled by \smallNickel{edges, position:$(n_\cU, n_\cF )$}. The position refers to the list in the attached \texttt{Mathematica} file.
   We give only exponents $(n_\cU, n_\cF )$ with $n_\cU,\ne0$, $n_\cF,\ne0$ in generic kinematics for ten edges. 
    }
	\label{Fanoupto10}
\end{table}     

\clearpage
    
\bibliographystyle{JHEP}

\begin{thebibliography}{100}

\bibitem{Travaglini:2022uwo}
G.~Travaglini et~al., \emph{{The SAGEX review on scattering amplitudes}},
  \href{https://doi.org/10.1088/1751-8121/ac8380}{\emph{J. Phys. A} {\bfseries
  55} (2022) 443001} [\href{https://arxiv.org/abs/2203.13011}{{\ttfamily
  2203.13011}}].

\bibitem{Bloch:2005bh}
S.~Bloch, H.~Esnault and D.~Kreimer, \emph{{On Motives associated to graph
  polynomials}}, \href{https://doi.org/10.1007/s00220-006-0040-2}{\emph{Commun.
  Math. Phys.} {\bfseries 267} (2006) 181}
  [\href{https://arxiv.org/abs/math/0510011}{{\ttfamily math/0510011}}].

\bibitem{Bogner:2007mn}
C.~Bogner and S.~Weinzierl, \emph{{Periods and Feynman integrals}},
  \href{https://doi.org/10.1063/1.3106041}{\emph{J. Math. Phys.} {\bfseries 50}
  (2009) 042302} [\href{https://arxiv.org/abs/0711.4863}{{\ttfamily
  0711.4863}}].

\bibitem{Brown:2009ta}
F.~C.~S. Brown, \emph{The massless higher-loop two-point function},
  {\emph{Commun. Math. Phys.} {\bfseries 287} (2009) 925–958}
  [\href{https://arxiv.org/abs/0804.1660}{{\ttfamily 0804.1660}}].

\bibitem{Schnetz:2013hqa}
O.~Schnetz, \emph{Graphical functions and single-valued multiple
  polylogarithms}, {\emph{Commun. Number Theory Phys.} {\bfseries 8} (2014)
  589–675} [\href{https://arxiv.org/abs/1302.6445}{{\ttfamily 1302.6445}}].

\bibitem{Laporta:2017okg}
S.~Laporta, \emph{High-precision calculation of the 4-loop contribution to the
  electron $g-2$ in qed}, {\emph{Phys. Lett. B} {\bfseries 772} (2017)
  232–238} [\href{https://arxiv.org/abs/1704.06996}{{\ttfamily 1704.06996}}].

\bibitem{Blum:2023qou}
T.~Blum and et~al., \emph{Hadronic contributions to the muon anomalous magnetic
  moment from lattice qcd}, {\emph{Phys. Rev. Lett.} {\bfseries 131} (2023)
  161802}.

\bibitem{Bjerrum-Bohr:2022blt}
N.~E.~J. Bjerrum-Bohr, P.~H. Damgaard, L.~Plante and P.~Vanhove, \emph{{The
  SAGEX review on scattering amplitudes Chapter 13: Post-Minkowskian expansion
  from scattering amplitudes}},
  \href{https://doi.org/10.1088/1751-8121/ac7a78}{\emph{J. Phys. A} {\bfseries
  55} (2022) 443014} [\href{https://arxiv.org/abs/2203.13024}{{\ttfamily
  2203.13024}}].

\bibitem{Bern:2021dqo}
Z.~Bern, J.~Parra-Martinez, R.~Roiban, M.~S. Ruf, C.-H. Shen, M.~P. Solon
  et~al., \emph{{Scattering Amplitudes and Conservative Binary Dynamics at
  ${\cal O}(G^4)$}},
  \href{https://doi.org/10.1103/PhysRevLett.126.171601}{\emph{Phys. Rev. Lett.}
  {\bfseries 126} (2021) 171601}
  [\href{https://arxiv.org/abs/2101.07254}{{\ttfamily 2101.07254}}].

\bibitem{Bern:2021yeh}
Z.~Bern, J.~Parra-Martinez, R.~Roiban, M.~S. Ruf, C.-H. Shen, M.~P. Solon
  et~al., \emph{{Scattering Amplitudes, the Tail Effect, and Conservative
  Binary Dynamics at O(G4)}},
  \href{https://doi.org/10.1103/PhysRevLett.128.161103}{\emph{Phys. Rev. Lett.}
  {\bfseries 128} (2022) 161103}
  [\href{https://arxiv.org/abs/2112.10750}{{\ttfamily 2112.10750}}].

\bibitem{Dlapa:2021npj}
C.~Dlapa, G.~K{\"a}lin, Z.~Liu and R.~A. Porto, \emph{{Dynamics of binary
  systems to fourth Post-Minkowskian order from the effective field theory
  approach}}, \href{https://doi.org/10.1016/j.physletb.2022.137203}{\emph{Phys.
  Lett. B} {\bfseries 831} (2022) 137203}
  [\href{https://arxiv.org/abs/2106.08276}{{\ttfamily 2106.08276}}].

\bibitem{Dlapa:2022lmu}
C.~Dlapa, G.~K{\"a}lin, Z.~Liu, J.~Neef and R.~A. Porto, \emph{{Radiation
  Reaction and Gravitational Waves at Fourth Post-Minkowskian Order}},
  \href{https://doi.org/10.1103/PhysRevLett.130.101401}{\emph{Phys. Rev. Lett.}
  {\bfseries 130} (2023) 101401}
  [\href{https://arxiv.org/abs/2210.05541}{{\ttfamily 2210.05541}}].

\bibitem{Bjerrum-Bohr:2022ows}
N.~E.~J. Bjerrum-Bohr, L.~Plant{\'e} and P.~Vanhove, \emph{{Effective Field
  Theory and Applications: Weak Field Observables from Scattering Amplitudes in
  Quantum Field Theory}},  \href{https://arxiv.org/abs/2212.08957}{{\ttfamily
  2212.08957}}.

\bibitem{Lellouch:2025rnz}
L.~Lellouch, A.~Lupo, M.~Sj{\"o}, K.~Szabo and P.~Vanhove, \emph{{Hadronic
  vacuum polarization to three loops in chiral perturbation theory}},
  \href{https://arxiv.org/abs/2510.12885}{{\ttfamily 2510.12885}}.

\bibitem{Klemm:2024wtd}
A.~Klemm, C.~Nega, B.~Sauer and J.~Plefka, \emph{{Calabi-Yau periods for black
  hole scattering in classical general relativity}},
  \href{https://doi.org/10.1103/PhysRevD.109.124046}{\emph{Phys. Rev. D}
  {\bfseries 109} (2024) 124046}
  [\href{https://arxiv.org/abs/2401.07899}{{\ttfamily 2401.07899}}].

\bibitem{Driesse:2024feo}
M.~Driesse, G.~U. Jakobsen, A.~Klemm, G.~Mogull, C.~Nega, J.~Plefka et~al.,
  \emph{{Emergence of Calabi{\textendash}Yau manifolds in high-precision
  black-hole scattering}},
  \href{https://doi.org/10.1038/s41586-025-08984-2}{\emph{Nature} {\bfseries
  641} (2025) 603} [\href{https://arxiv.org/abs/2411.11846}{{\ttfamily
  2411.11846}}].

\bibitem{Frellesvig:2023bbf}
H.~Frellesvig, R.~Morales and M.~Wilhelm,
``Calabi-Yau Meets Gravity: A Calabi-Yau Threefold at Fifth Post-Minkowskian Order,''
Phys. Rev. Lett. \textbf{132} (2024) no.20, 201602
doi:10.1103/PhysRevLett.132.201602
[arXiv:2312.11371 [hep-th]].

\bibitem{Brammer:2025rqo}
D.~Brammer, H.~Frellesvig, R.~Morales and M.~Wilhelm,
``Classification of Feynman integral geometries for black-hole scattering at 5PM order,''
JHEP \textbf{10} (2025), 212
doi:10.1007/JHEP10(2025)212
[arXiv:2505.10274 [hep-th]].

\bibitem{2010arXiv1010.5060N}
L.~{Nilsson} and M.~{Passare}, \emph{{Mellin transforms of multivariate
  rational functions}},
  \href{https://doi.org/10.1007/s12220-011-9235-7}{\emph{Journal of Geometric
  Analysis} {\bfseries 23} (2010) 24}
  [\href{https://arxiv.org/abs/1010.5060}{{\ttfamily 1010.5060}}].

\bibitem{2011arXiv1103.6273B}
C.~{Berkesch}, J.~{Forsg{\r{a}}rd} and M.~{Passare}, \emph{{Euler--Mellin
  integrals and A-hypergeometric functions}}, {\emph{Michigan Math. J.}
  {\bfseries 1} (2014) 101} [\href{https://arxiv.org/abs/1103.6273}{{\ttfamily
  1103.6273}}].

\bibitem{Ananthanarayan:2018tog}
B.~Ananthanarayan, A.~Pal, S.~Ramanan and R.~Sarkar, \emph{{Unveiling Regions
  in multi-scale Feynman Integrals using Singularities and Power Geometry}},
  \href{https://doi.org/10.1140/epjc/s10052-019-6533-x}{\emph{Eur. Phys. J. C}
  {\bfseries 79} (2019) 57} [\href{https://arxiv.org/abs/1810.06270}{{\ttfamily
  1810.06270}}].

\bibitem{Gardi:2022khw}
E.~Gardi, F.~Herzog, S.~Jones, Y.~Ma and J.~Schlenk, \emph{{The on-shell
  expansion: from Landau equations to the Newton polytope}},
  \href{https://doi.org/10.1007/JHEP07(2023)197}{\emph{JHEP} {\bfseries 07}
  (2023) 197} [\href{https://arxiv.org/abs/2211.14845}{{\ttfamily
  2211.14845}}].

\bibitem{Gardi:2024axt}
E.~Gardi, F.~Herzog, S.~Jones and Y.~Ma, \emph{{Dissecting polytopes: Landau
  singularities and asymptotic expansions in 2 \textrightarrow{} 2
  scattering}}, \href{https://doi.org/10.1007/JHEP08(2024)127}{\emph{JHEP}
  {\bfseries 08} (2024) 127}
  [\href{https://arxiv.org/abs/2407.13738}{{\ttfamily 2407.13738}}].

\bibitem{Ma:2023hrt}
Y.~Ma, \emph{{Identifying regions in wide-angle scattering via
  graph-theoretical approaches}},
  \href{https://doi.org/10.1007/JHEP09(2024)197}{\emph{JHEP} {\bfseries 09}
  (2024) 197} [\href{https://arxiv.org/abs/2312.14012}{{\ttfamily
  2312.14012}}].

\bibitem{Borowka:2015mxa}
S.~Borowka, G.~Heinrich, S.~P. Jones, M.~Kerner, J.~Schlenk and T.~Zirke,
  \emph{{SecDec-3.0: numerical evaluation of multi-scale integrals beyond one
  loop}}, \href{https://doi.org/10.1016/j.cpc.2015.05.022}{\emph{Comput. Phys.
  Commun.} {\bfseries 196} (2015) 470}
  [\href{https://arxiv.org/abs/1502.06595}{{\ttfamily 1502.06595}}].

\bibitem{Heinrich:2021dbf}
G.~Heinrich, S.~Jahn, S.~P. Jones, M.~Kerner, F.~Langer, V.~Magerya et~al.,
  \emph{{Expansion by regions with pySecDec}},
  \href{https://doi.org/10.1016/j.cpc.2021.108267}{\emph{Comput. Phys. Commun.}
  {\bfseries 273} (2022) 108267}
  [\href{https://arxiv.org/abs/2108.10807}{{\ttfamily 2108.10807}}].

\bibitem{GELFAND1990255}
I.~Gelfand, M.~Kapranov and A.~Zelevinsky, \emph{Generalized {E}uler integrals
  and {A}-hypergeometric functions},
  \href{https://doi.org/https://doi.org/10.1016/0001-8708(90)90048-R}{\emph{Advances
  in Mathematics} {\bfseries 84} (1990) 255 }.

\bibitem{2016arXiv160504970N}
E.~{Nasrollahpoursamami}, \emph{{Periods of Feynman Diagrams and GKZ
  D-Modules}},  \href{https://arxiv.org/abs/1605.04970}{{\ttfamily
  1605.04970}}.

\bibitem{Schultka:2018nrs}
K.~Schultka, \emph{{Toric geometry and regularization of Feynman integrals}},
  \href{https://arxiv.org/abs/1806.01086}{{\ttfamily 1806.01086}}.

\bibitem{delaCruz:2019skx}
L.~de~la Cruz, \emph{{Feynman integrals as A-hypergeometric functions}},
  \href{https://doi.org/10.1007/JHEP12(2019)123}{\emph{JHEP} {\bfseries 12}
  (2019) 123} [\href{https://arxiv.org/abs/1907.00507}{{\ttfamily
  1907.00507}}].

\bibitem{Klausen:2019hrg}
R.~P. Klausen, \emph{{Hypergeometric Series Representations of Feynman
  Integrals by GKZ Hypergeometric Systems}},
  \href{https://doi.org/10.1007/JHEP04(2020)121}{\emph{JHEP} {\bfseries 04}
  (2020) 121} [\href{https://arxiv.org/abs/1910.08651}{{\ttfamily
  1910.08651}}].

\bibitem{Klausen:2021yrt}
R.~P. Klausen, \emph{{Kinematic singularities of Feynman integrals and
  principal A-determinants}},
  \href{https://doi.org/10.1007/JHEP02(2022)004}{\emph{JHEP} {\bfseries 02}
  (2022) 004} [\href{https://arxiv.org/abs/2109.07584}{{\ttfamily
  2109.07584}}].

\bibitem{Ananthanarayan:2022ntm}
B.~Ananthanarayan, S.~Banik, S.~Bera and S.~Datta, \emph{{FeynGKZ: A
  Mathematica package for solving Feynman integrals using GKZ hypergeometric
  systems}}, \href{https://doi.org/10.1016/j.cpc.2023.108699}{\emph{Comput.
  Phys. Commun.} {\bfseries 287} (2023) 108699}
  [\href{https://arxiv.org/abs/2211.01285}{{\ttfamily 2211.01285}}].

\bibitem{cox2011toric}
D.~Cox, J.~Little and H.~Schenck, \emph{Toric Varieties}, Graduate studies in
  mathematics. American Mathematical Society, 2011.

\bibitem{Batyrev:1993oya}
V.~V. Batyrev, \emph{{Dual polyhedra and mirror symmetry for Calabi-Yau
  hypersurfaces in toric varieties}}, {\emph{J. Alg. Geom.} {\bfseries 3}
  (1994) 493} [\href{https://arxiv.org/abs/alg-geom/9310003}{{\ttfamily
  alg-geom/9310003}}].

\bibitem{Kreuzer:2000xy}
M.~Kreuzer and H.~Skarke, \emph{{Complete classification of reflexive polyhedra
  in four-dimensions}},
  \href{https://doi.org/10.4310/ATMP.2000.v4.n6.a2}{\emph{Adv. Theor. Math.
  Phys.} {\bfseries 4} (2000) 1209}
  [\href{https://arxiv.org/abs/hep-th/0002240}{{\ttfamily hep-th/0002240}}].

\bibitem{Bloch:2014qca}
S.~Bloch, M.~Kerr and P.~Vanhove, \emph{{A Feynman integral via higher normal
  functions}}, \href{https://doi.org/10.1112/S0010437X15007472}{\emph{Compos.
  Math.} {\bfseries 151} (2015) 2329}
  [\href{https://arxiv.org/abs/1406.2664}{{\ttfamily 1406.2664}}].

\bibitem{Bloch:2016izu}
S.~Bloch, M.~Kerr and P.~Vanhove, \emph{{Local mirror symmetry and the sunset
  Feynman integral}},
  \href{https://doi.org/10.4310/ATMP.2017.v21.n6.a1}{\emph{Adv. Theor. Math.
  Phys.} {\bfseries 21} (2017) 1373}
  [\href{https://arxiv.org/abs/1601.08181}{{\ttfamily 1601.08181}}].

\bibitem{Schimmrigk:2024xid}
R.~Schimmrigk, \emph{{Special Fano geometry from Feynman integrals}},
  \href{https://doi.org/10.1016/j.physletb.2025.139420}{\emph{Phys. Lett. B}
  {\bfseries 864} (2025) 139420}
  [\href{https://arxiv.org/abs/2412.20236}{{\ttfamily 2412.20236}}].

\bibitem{Gambuti:2023eqh}
G.~Gambuti, D.~A. Kosower, P.~P. Novichkov and L.~Tancredi, \emph{{Finite
  Feynman integrals}},
  \href{https://doi.org/10.1103/PhysRevD.110.116026}{\emph{Phys. Rev. D}
  {\bfseries 110} (2024) 116026}
  [\href{https://arxiv.org/abs/2311.16907}{{\ttfamily 2311.16907}}].

\bibitem{delaCruz:2024xsm}
L.~de~la Cruz, D.~A. Kosower and P.~P. Novichkov, \emph{{Finite integrals from
  Feynman polytopes}},
  \href{https://doi.org/10.1103/PhysRevD.111.105013}{\emph{Phys. Rev. D}
  {\bfseries 111} (2025) 105013}
  [\href{https://arxiv.org/abs/2410.18014}{{\ttfamily 2410.18014}}].

\bibitem{vonManteuffel:2014qoa}
A.~von Manteuffel, E.~Panzer and R.~M. Schabinger, \emph{{A quasi-finite basis
  for multi-loop Feynman integrals}},
  \href{https://doi.org/10.1007/JHEP02(2015)120}{\emph{JHEP} {\bfseries 02}
  (2015) 120} [\href{https://arxiv.org/abs/1411.7392}{{\ttfamily 1411.7392}}].

\bibitem{Heinrich:2008si}
G.~Heinrich, \emph{{Sector Decomposition}},
  \href{https://doi.org/10.1142/S0217751X08040263}{\emph{Int. J. Mod. Phys. A}
  {\bfseries 23} (2008) 1457}
  [\href{https://arxiv.org/abs/0803.4177}{{\ttfamily 0803.4177}}].

\bibitem{Nakanishi:1971}
N.~Nakanishi, \emph{Graph theory and Feynman integrals}, volume~11. Routledge,,
  New York :, 1971.

\bibitem{Tarasov:1996br}
O.~V. Tarasov, \emph{{Connection between Feynman integrals having different
  values of the space-time dimension}},
  \href{https://doi.org/10.1103/PhysRevD.54.6479}{\emph{Phys. Rev. D}
  {\bfseries 54} (1996) 6479}
  [\href{https://arxiv.org/abs/hep-th/9606018}{{\ttfamily hep-th/9606018}}].

\bibitem{Cheng:1987ga}
H.~Cheng and T.~Wu, \emph{Expanding Protons: Scattering at High Energies}. MIT
  Press, Cambridge, 1987.

\bibitem{Panzer:2015ida}
E.~Panzer, \emph{{Feynman integrals and hyperlogarithms}}, Ph.D. thesis,
  Humboldt U., 2015.
\newblock \href{https://arxiv.org/abs/1506.07243}{{\ttfamily 1506.07243}}.
\newblock 10.18452/17157.

\bibitem{Weinzierl:2022eaz}
S.~Weinzierl, \emph{{Feynman Integrals}}. 1, 2022,
  \href{https://doi.org/10.1007/978-3-030-99558-4}{10.1007/978-3-030-99558-4},
  [\href{https://arxiv.org/abs/2201.03593}{{\ttfamily 2201.03593}}].

\bibitem{Haase:2012}
C.~Haase, \emph{Lecture Notes on Lattice Polytopes}. 2012.

\bibitem{Beck:2015}
M.~Beck and S.~Robins, \emph{Computing the Continuous Discretely: Integer-point
  Enumeration in Polyhedra}, Undergraduate Texts in Mathematics. Springer, New
  York, 2nd~ed., 2015,
  \href{https://doi.org/10.1007/978-1-4939-2969-3}{10.1007/978-1-4939-2969-3}.

\bibitem{Arkani-Hamed:2022cqe}
N.~Arkani-Hamed, A.~Hillman and S.~Mizera, \emph{{Feynman polytopes and the
  tropical geometry of UV and IR divergences}},
  \href{https://doi.org/10.1103/PhysRevD.105.125013}{\emph{Phys. Rev. D}
  {\bfseries 105} (2022) 125013}
  [\href{https://arxiv.org/abs/2202.12296}{{\ttfamily 2202.12296}}].

\bibitem{Borinsky:2023jdv}
M.~Borinsky, H.~J. Munch and F.~Tellander, \emph{{Tropical Feynman integration
  in the Minkowski regime}},
  \href{https://doi.org/10.1016/j.cpc.2023.108874}{\emph{Comput. Phys. Commun.}
  {\bfseries 292} (2023) 108874}
  [\href{https://arxiv.org/abs/2302.08955}{{\ttfamily 2302.08955}}].

\bibitem{postnikov2005}
A.~Postnikov, \emph{Permutohedra, associahedra, and beyond},
  {\emph{International Mathematics Research Notices} (2005) }
  [\href{https://arxiv.org/abs/math/0507163}{{\ttfamily math/0507163}}].

\bibitem{Cara:2019}
C.~Monical, N.~Tokcan and A.~Yong, \emph{Newton polytopes in algebraic
  combinatorics}, {\emph{Selecta Mathematica} {\bfseries 25} (2019) 66}
  [\href{https://arxiv.org/abs/1703.02583}{{\ttfamily 1703.02583}}].

\bibitem{Lee:2013hzt}
R.~N. Lee and A.~A. Pomeransky, \emph{{Critical points and number of master
  integrals}}, \href{https://doi.org/10.1007/JHEP11(2013)165}{\emph{JHEP}
  {\bfseries 11} (2013) 165} [\href{https://arxiv.org/abs/1308.6676}{{\ttfamily
  1308.6676}}].

\bibitem{Doran:2023yzu}
C.~F. Doran, A.~Harder, P.~Vanhove and E.~Pichon-Pharabod, \emph{{Motivic
  Geometry of two-Loop Feynman Integrals}},
  \href{https://doi.org/10.1093/qmath/haae015}{\emph{Quart. J. Math. Oxford
  Ser.} {\bfseries 75} (2024) 901}
  [\href{https://arxiv.org/abs/2302.14840}{{\ttfamily 2302.14840}}].

\bibitem{Gruber2007}
P.~M. Gruber, \emph{Convex and Discrete Geometry}. Springer, Berlin, 2007.

\bibitem{polymake:FPSAC_2009}
M.~Joswig, B.~M\"uller and A.~Paffenholz, \emph{{\tt polymake} and lattice
  polytopes},  in \emph{21st {I}nternational {C}onference on {F}ormal {P}ower
  {S}eries and {A}lgebraic {C}ombinatorics ({FPSAC} 2009)}, Discrete Math.
  Theor. Comput. Sci. Proc., AK, pp.~491--502.
\newblock Assoc. Discrete Math. Theor. Comput. Sci., Nancy, 2009.

\bibitem{Haase:2017}
C.~Haase, J.~Hofscheier, B.~Nill and T.~Theobald, \emph{Mixed ehrhart
  polynomials},
  \href{https://doi.org/10.1007/s00209-016-1818-6}{\emph{Mathematische
  Zeitschrift} {\bfseries 286} (2017) 1157}.

\bibitem{BrandenburgEtAl:2020}
M.-C. Brandenburg, S.~Elia and L.~Zhang, \emph{Multivariate volume, ehrhart,
  and $h^*$-polynomials of polytropes},
  \href{https://doi.org/10.1016/j.jsc.2022.04.011}{\emph{Journal of Symbolic
  Computation} {\bfseries 114} (2023) 209–230}
  [\href{https://arxiv.org/abs/2006.01920}{{\ttfamily 2006.01920}}].

\bibitem{CoxKatz1999}
D.~A. Cox and S.~Katz, \emph{Mirror Symmetry and Algebraic Geometry}, vol.~68
  of \emph{Mathematical Surveys and Monographs}. American Mathematical Society,
  Providence, RI, 1999.

\bibitem{Nill2005}
B.~Nill, \emph{Gorenstein toric fano varieties}, {\emph{manuscripta
  mathematica} {\bfseries 116} (2005) 183}
  [\href{https://arxiv.org/abs/math/0405448}{{\ttfamily math/0405448}}].

\bibitem{Kasprzyk12}
A.~M. Kasprzyk and B.~N. fibrations~on Calabi-Yau~manifolds.

\bibitem{Telen:2022}
S.~Telen, \emph{Introduction to toric geometry},
  \href{https://arxiv.org/abs/2203.01690}{{\ttfamily 2203.01690}}.

\bibitem{verrill1996root}
H.~A. Verrill, \emph{Root lattices and pencils of varieties}, {\emph{Journal of
  Mathematics of Kyoto University} {\bfseries 36} (1996) 423}.

\bibitem{sagemath}
{The Sage Developers}, \emph{{S}ageMath, the {S}age {M}athematics {S}oftware
  {S}ystem ({V}ersion 10.5)}, 2024.

\bibitem{Kreuzer:2002uu}
M.~Kreuzer and H.~Skarke, \emph{{PALP: A Package for analyzing lattice
  polytopes with applications to toric geometry}},
  \href{https://doi.org/10.1016/S0010-4655(03)00491-0}{\emph{Comput. Phys.
  Commun.} {\bfseries 157} (2004) 87}
  [\href{https://arxiv.org/abs/math/0204356}{{\ttfamily math/0204356}}].

\bibitem{Nogueira:1991ex}
P.~Nogueira, \emph{{Automatic Feynman Graph Generation}},
  \href{https://doi.org/10.1006/jcph.1993.1074}{\emph{J. Comput. Phys.}
  {\bfseries 105} (1993) 279}.

\bibitem{Bremner_2014}
D.~Bremner, M.~Dutour~Sikirić, D.~V. Pasechnik, T.~Rehn and A.~Schürmann,
  \emph{Computing symmetry groups of polyhedra},
  \href{https://doi.org/10.1112/s1461157014000400}{\emph{LMS Journal of
  Computation and Mathematics} {\bfseries 17} (2014) 565–581}.

\bibitem{grinis2013normalformsconvexlattice}
R.~Grinis and A.~Kasprzyk, \emph{Normal forms of convex lattice polytopes},
  \href{https://arxiv.org/abs/1301.6641}{{\ttfamily 1301.6641}}.

\bibitem{delaCruz:2024ssb}
L.~de~la Cruz, \emph{{Polytope symmetries of Feynman integrals}},
  \href{https://doi.org/10.1016/j.physletb.2024.138744}{\emph{Phys. Lett. B}
  {\bfseries 854} (2024) 138744}
  [\href{https://arxiv.org/abs/2404.03564}{{\ttfamily 2404.03564}}].

\bibitem{Pak:2011xt}
A.~Pak, \emph{{The toolbox of modern multi-loop calculations: novel analytic
  and semi-analytic techniques}},
  \href{https://doi.org/10.1088/1742-6596/368/1/012049}{\emph{J. Phys. Conf.
  Ser.} {\bfseries 368} (2012) 012049}
  [\href{https://arxiv.org/abs/1111.0868}{{\ttfamily 1111.0868}}].

\bibitem{Shtabovenko:2023idz}
V.~Shtabovenko, R.~Mertig and F.~Orellana, \emph{{FeynCalc 10: Do multiloop
  integrals dream of computer codes?}},
  \href{https://doi.org/10.1016/j.cpc.2024.109357}{\emph{Comput. Phys. Commun.}
  {\bfseries 306} (2025) 109357}
  [\href{https://arxiv.org/abs/2312.14089}{{\ttfamily 2312.14089}}].

\bibitem{BagnaraHZ08SCP}
R.~Bagnara, P.~M. Hill and E.~Zaffanella, \emph{The {Parma Polyhedra Library}:
  Toward a complete set of numerical abstractions for the analysis and
  verification of hardware and software systems}, {\emph{Science of Computer
  Programming} {\bfseries 72} (2008) 3}.

\bibitem{Kompaniets:2017yct}
M.~V. Kompaniets and E.~Panzer, \emph{{Minimally subtracted six loop
  renormalization of $O(n)$-symmetric $\phi^4$ theory and critical exponents}},
  \href{https://doi.org/10.1103/PhysRevD.96.036016}{\emph{Phys. Rev. D}
  {\bfseries 96} (2017) 036016}
  [\href{https://arxiv.org/abs/1705.06483}{{\ttfamily 1705.06483}}].

\bibitem{Abreu:2024fei}
S.~Abreu, P.~F. Monni, B.~Page and J.~Usovitsch, \emph{{Planar six-point
  Feynman integrals for four-dimensional gauge theories}},
  \href{https://doi.org/10.1007/JHEP06(2025)112}{\emph{JHEP} {\bfseries 06}
  (2025) 112} [\href{https://arxiv.org/abs/2412.19884}{{\ttfamily
  2412.19884}}].

\bibitem{Panzer:2014caa}
E.~Panzer, \emph{{Algorithms for the symbolic integration of hyperlogarithms
  with applications to Feynman integrals}},
  \href{https://doi.org/10.1016/j.cpc.2014.10.019}{\emph{Comput. Phys. Commun.}
  {\bfseries 188} (2015) 148}
  [\href{https://arxiv.org/abs/1403.3385}{{\ttfamily 1403.3385}}].

\bibitem{Brown:2021umn}
F.~Brown, \emph{{Invariant Differential Forms on Complexes of Graphs and
  Feynman Integrals}},
  \href{https://doi.org/10.3842/SIGMA.2021.103}{\emph{SIGMA} {\bfseries 17}
  (2021) 103} [\href{https://arxiv.org/abs/2101.04419}{{\ttfamily
  2101.04419}}].

\bibitem{Broadhurst:1985tld}
D.~J. Broadhurst, \emph{{Massless scalar Feynman diagrams: five loops and
  beyond}},  \href{https://arxiv.org/abs/1604.08027}{{\ttfamily 1604.08027}}.

\bibitem{Panzer:2019yxl}
E.~Panzer, \emph{{Hepp\textquoteright{}s bound for Feynman graphs and
  matroids}}, \href{https://doi.org/10.4171/aihpd/126}{\emph{Ann. Inst. H.
  Poincare D Comb. Phys. Interact.} {\bfseries 10} (2022) 31}
  [\href{https://arxiv.org/abs/1908.09820}{{\ttfamily 1908.09820}}].

\bibitem{Kasprzyk_2022}
A.~Kasprzyk and V.~Przyjalkowski, \emph{Laurent polynomials in mirror symmetry:
  why and how?},
  \href{https://doi.org/10.22199/issn.0717-6279-5279}{\emph{Proyecciones
  (Antofagasta)} {\bfseries 41} (2022) 481–515}.

\bibitem{Bloch:2013tra}
S.~Bloch and P.~Vanhove, \emph{{The elliptic dilogarithm for the sunset
  graph}}, \href{https://doi.org/10.1016/j.jnt.2014.09.032}{\emph{J. Number
  Theor.} {\bfseries 148} (2015) 328}
  [\href{https://arxiv.org/abs/1309.5865}{{\ttfamily 1309.5865}}].

\bibitem{Broedel:2019kmn}
J.~Broedel, C.~Duhr, F.~Dulat, R.~Marzucca, B.~Penante and L.~Tancredi,
  \emph{{An analytic solution for the equal-mass banana graph}},
  \href{https://doi.org/10.1007/JHEP09(2019)112}{\emph{JHEP} {\bfseries 09}
  (2019) 112} [\href{https://arxiv.org/abs/1907.03787}{{\ttfamily
  1907.03787}}].

\bibitem{Broedel:2021zij}
J.~Broedel, C.~Duhr and N.~Matthes, \emph{{Meromorphic modular forms and the
  three-loop equal-mass banana integral}},
  \href{https://doi.org/10.1007/JHEP02(2022)184}{\emph{JHEP} {\bfseries 02}
  (2022) 184} [\href{https://arxiv.org/abs/2109.15251}{{\ttfamily
  2109.15251}}].

\bibitem{Pogel:2022yat}
S.~P{\"o}gel, X.~Wang and S.~Weinzierl, \emph{{The three-loop equal-mass banana
  integral in {\ensuremath{\varepsilon}}-factorised form with meromorphic
  modular forms}}, \href{https://doi.org/10.1007/JHEP09(2022)062}{\emph{JHEP}
  {\bfseries 09} (2022) 062}
  [\href{https://arxiv.org/abs/2207.12893}{{\ttfamily 2207.12893}}].

\bibitem{Duhr:2025ppd}
C.~Duhr and S.~Maggio, \emph{{Feynman integrals, elliptic integrals and
  two-parameter K3 surfaces}},
  \href{https://doi.org/10.1007/JHEP06(2025)250}{\emph{JHEP} {\bfseries 06}
  (2025) 250} [\href{https://arxiv.org/abs/2502.15326}{{\ttfamily
  2502.15326}}].

\bibitem{HulkerVerrill3fold}
K.~Hulek and H.~Verrill, \emph{On the modularity of calabi-yau threefolds
  containing elliptic ruled surfaces},
  \href{https://arxiv.org/abs/math/0502158}{{\ttfamily math/0502158}}.

\bibitem{HulekVerrill4fold}
K.~Hulek and H.~Verrill, \emph{On modularity of rigid and nonrigid calabi-yau
  varieties associated to the root lattice a4}, {\emph{Nagoya Mathematical
  Journal} {\bfseries 179} (2005) 103–146}.

\bibitem{Fulton93}
W.~Fulton, \emph{Introduction to Toric Varieties. (AM-131)}. Princeton
  University Press, 1993.

\bibitem{Arapura2012}
D.~Arapura, \emph{Algebraic geometry over the complex numbers}, Universitext,.
  Springer,, New York, c2012.

\bibitem{Manin1986cubic}
Y.~I. Manin, \emph{Cubic Forms: Algebra, Geometry, Arithmetic}, vol.~4 of
  \emph{North-Holland Mathematical Library}. North-Holland, Amsterdam, 2nd~ed.,
  1986.

\bibitem{Dolgachev_CAG}
I.~V. Dolgachev, \emph{Classical Algebraic Geometry: A Modern View}, Cambridge
  University Press. Cambridge University Press, 2012.

\bibitem{Chavez:2012kn}
F.~Chavez and C.~Duhr, \emph{{Three-mass triangle integrals and single-valued
  polylogarithms}}, \href{https://doi.org/10.1007/JHEP11(2012)114}{\emph{JHEP}
  {\bfseries 11} (2012) 114} [\href{https://arxiv.org/abs/1209.2722}{{\ttfamily
  1209.2722}}].

\bibitem{tHooft:1978jhc}
G.~'t~Hooft and M.~J.~G. Veltman, \emph{{Scalar One Loop Integrals}},
  \href{https://doi.org/10.1016/0550-3213(79)90605-9}{\emph{Nucl. Phys. B}
  {\bfseries 153} (1979) 365}.

\bibitem{Dolgachev1996}
I.~Dolgachev, \emph{Mirror symmetry for lattice polarized k3 surfaces},
  \href{https://doi.org/10.1007/BF02362332}{\emph{Journal of Mathematical
  Sciences} {\bfseries 81} (1996) 2599}.

\bibitem{Nikulin1980}
V.~V. Nikulin, \emph{Integral symmetric bilinear forms and some of their
  applications},
  \href{https://doi.org/10.1070/IM1980v014n01ABEH001060}{\emph{Mathematics of
  the USSR-Izvestija} {\bfseries 14} (1980) 103}.

\bibitem{Doran:2016uea}
C.~F. Doran, A.~Harder and A.~Thompson, \emph{{Mirror symmetry, Tyurin
  degenerations and fibrations on Calabi-Yau manifolds}}, {\emph{Proc. Symp.
  Pure Math.} {\bfseries 96} (2017) 93}
  [\href{https://arxiv.org/abs/1601.08110}{{\ttfamily 1601.08110}}].

\bibitem{Brown:2015fyf}
F.~C.~S. Brown and O.~Schnetz, \emph{Modular forms in quantum field theory},
  {\emph{Commun. Number Theory Phys.} {\bfseries 7} (2013) 293–325}
  [\href{https://arxiv.org/abs/1304.5342}{{\ttfamily 1304.5342}}].

\bibitem{Schnetz:2016fhy}
O.~Schnetz, \emph{Quantum field theory over $\mathbb{F}_q$}, {\emph{Electron.
  J. Combin.} {\bfseries 18} (2011) P102}
  [\href{https://arxiv.org/abs/0909.0905}{{\ttfamily 0909.0905}}].

\bibitem{Panzer:2018tiv}
E.~Panzer and O.~Schnetz, \emph{The galois coaction on $\varphi^4$ periods},
  \href{https://doi.org/10.4310/CNTP.2017.v11.n3.a3}{\emph{Commun. Number
  Theory Phys.} {\bfseries 11} (2017) 657–705}
  [\href{https://arxiv.org/abs/1603.04289}{{\ttfamily 1603.04289}}].

\bibitem{Borinsky:2020rqs}
M.~Borinsky, \emph{{Tropical Monte Carlo quadrature for Feynman integrals}},
  \href{https://doi.org/10.4171/aihpd/158}{\emph{Ann. Inst. H. Poincare D Comb.
  Phys. Interact.} {\bfseries 10} (2023) 635}
  [\href{https://arxiv.org/abs/2008.12310}{{\ttfamily 2008.12310}}].

\bibitem{oeis}
{OEIS Foundation Inc.}, \emph{The {O}n-{L}ine {E}ncyclopedia of {I}nteger
  {S}equences},  2025.

\bibitem{reflexivefano}
\url{https://github.com/pierrevanhove/ReflexiveFanoPolytopes}

\bibitem{Bogner:2017xhp}
C.~Bogner, S.~Borowka, T.~Hahn, G.~Heinrich, S.~P. Jones, M.~Kerner et~al.,
  \emph{{Loopedia, a Database for Loop Integrals}},
  \href{https://doi.org/10.1016/j.cpc.2017.12.017}{\emph{Comput. Phys. Commun.}
  {\bfseries 225} (2018) 1} [\href{https://arxiv.org/abs/1709.01266}{{\ttfamily
  1709.01266}}].
  
\end{thebibliography}

\providecommand{\href}[2]{#2}\begingroup\raggedright\endgroup
\end{document}